%% file: main.tex
\documentclass[5p,preprint,dvipsnames]{elsarticle}

\usepackage{hyperref}
\usepackage{amssymb, amsmath}
\usepackage{stmaryrd}
\usepackage{graphicx}

\usepackage{microtype}

\usepackage{csquotes} 

\usepackage{fontawesome}
\usepackage{txfonts}

\usepackage{subfigure}
\usepackage{tikz,subfigure}
\usetikzlibrary{arrows,arrows.meta,shapes,shapes.multipart,automata,backgrounds,petri,positioning,shadows,matrix,decorations.pathmorphing,fit,positioning,calc,backgrounds,shapes.geometric}

\usepackage{multirow}
\usepackage{threeparttable}

\usepackage{listings}
\usepackage{todonotes}

\usepackage[capitalise,nameinlink]{cleveref}

\definecolor{amber}{rgb}{1.0,0.75,0.0}
\definecolor{amaranth}{rgb}{0.9, 0.17, 0.31}

\usepackage{paralist}
\usepackage{mathtools}

\newtheorem{definition}{Definition}
\newtheorem{theorem}{Theorem}
\newtheorem{lemma}{Lemma}

\newtheorem{example}{Example}

\newproof{proof}{Proof}
\newcommand{\qedwhite}{\hfill \ensuremath{\Box}}
\newcommand{\qedblack}{\hfill \ensuremath{\blacksquare}}

\usepackage{xspace}

\newcommand{\tdbnet}{timed DB-net\xspace}
\newcommand{\tdbnets}{timed DB-nets\xspace}

\newcommand{\labeltitle}[1]{\vskip 0.03in \noindent\textbf{#1}} 
\newcommand{\labelsubtitle}[1]{\vskip 0.03in \noindent\emph{#1}.}

\newcommand{\eg}{e.g.,\ }
\newcommand{\ie}{i.e.,\ }

\newcommand{\Bool}{\mathbb{B}}
\newcommand{\BExp}{\mathsf{BExp}}
\newcommand{\Exp}{\mathsf{Exp}}

\newcommand{\mathvar}[1]{\mathord{\mathit{#1}}}

\newcommand{\type}{\mathvar{type}}

\newcommand{\pchar}{\mathvar{char}}
\newcommand{\PC}{\mathvar{PC}}

\newcommand{\CPT}{\mathvar{CPT}}
\newcommand{\EL}{\mathvar{EL}}

\newcommand{\iC}{\mathvar{inContr}}
\newcommand{\oC}{\mathvar{outContr}}
\newcommand{\contr}{\mathvar{contr}}

\newcommand{\cost}{\mathvar{cost}}

\newcommand{\match}{\mathsf{match}}

\newcommand{\inn}{\mathsf{in}}

\newcommand{\sem}[1]{\llbracket#1\rrbracket}

\newcommand{\List}[1]{\mathsf{List}\,#1}

\newcommand{\cin}[2]{\text{color}_{\text{in}}(#1)_{#2}}
\newcommand{\cout}[2]{\text{color}_{\text{out}}(#1)_{#2}}

\newcommand{\cptize}[2]{#1_{\CPT(#2)}}

\usepackage{enumitem}

\newcounter{reqno}
\newcommand{\req}[1]{\refstepcounter{reqno}\label{#1}REQ-\thereqno}
\newcommand{\refreq}[1]{REQ-\ref{#1}}

\newcommand{\iotshort}{IOT\xspace}
\newcommand{\Ce}{CE\xspace}
\newcommand{\mt}{MT\xspace}
\newcommand{\Aggregator}{aggregator\xspace}

\newcommand{\prg}[1]{\text{#1}} 
\newcommand{\cnd}[1]{\text{#1}} 

\crefname{section}{Sec.}{Sec.}
\Crefname{section}{Section}{Sections}
\crefname{figure}{Fig.}{Figs.}
\Crefname{figure}{Figure}{Figures}
\crefname{table}{Tab.}{Tab.}
\Crefname{table}{Table}{Tables}
\crefname{lstlisting}{List.}{Lists.}
\Crefname{lstlisting}{Listing}{Listings}

\mathchardef\UrlBreakPenalty=10000

\begin{document}

\begin{frontmatter}

\title{Responsible Composition and Optimization of Integration Processes under Correctness Preserving Guarantees}

 \author[add1]{Daniel Ritter}{\corref{mycorrespondingauthor}}
 \ead{daniel.ritter@sap.com}

 \author[add2]{Fredrik Nordvall Forsberg}
 \ead{fredrik.nordvall-forsberg@strath.ac.uk}
 
 \author[add3]{Stefanie Rinderle-Ma}{\corref{mycorrespondingauthor}}
 \ead{stefanie.rinderle-ma@tum.de}

 \address[add1]{SAP, Walldorf (Baden), Germany}
 \address[add2]{Department of Computer and Information Sciences, University of Strathclyde, Glasogw, Scotland}
 \address[add3]{School of Computation, Information and Technology, Technische Universit\"at M\"unchen, Garching, Germany}

\begin{abstract}
Enterprise Application Integration deals with the problem of connecting heterogeneous applications, and is the centerpiece of current on-premise, cloud and device integration scenarios.
For integration scenarios, structurally correct composition of patterns into processes and improvements of integration processes are crucial.
In order to achieve this, we formalize compositions of integration patterns based on their characteristics, and describe optimization strategies that help to reduce the model complexity, and improve the process execution efficiency using design time techniques.
Using the formalism of timed DB-nets --- a refinement of Petri nets --- we model integration logic features such as control- and data flow, transactional data storage, compensation and exception handling, and time aspects that are present in reoccurring solutions as separate integration patterns.
We then propose a realization of optimization strategies using graph rewriting, and prove that the optimizations we consider preserve both structural and functional correctness.
We evaluate the improvements on a real-world catalog of pattern compositions, containing over $900$ integration processes, and illustrate the correctness properties in case studies based on two of these processes.
\end{abstract}


\end{frontmatter}

\input{introduction}
\input{optimization_strategies}
\input{pattern_composition_formalization}
\input{petri_net_semantics}
\input{optimization_realization_short}
\input{evaluation}
\input{relatedwork}
\input{discussion}

\bibliographystyle{elsarticle-num}
\bibliography{optimization}

\end{document}

%% file: introduction.tex
\section{Introduction}
\label{sec:intro}
In a highly digitized and connected world, in which enterprises get more and more intertwined with each other, the integration of applications scattered across on-premises, cloud and devices is crucial for enabling innovation, improved productivity, and more accessible information \cite{DBLP:conf/ahfe/JeskeWLWS20}.
This is facilitated by process technology 
based on integration building blocks called integration patterns~\cite{hohpe2004enterprise, 
DBLP:conf/caise/0001H15,ritter2016exception,Ritter201736}.
The composition of these patterns into integration processes can result in complex models that are often vendor-specific, informal and ad-hoc~\cite{Ritter201736}; optimizing such integration processes is often desirable, but hard.
In most cases complex process control flows are further complicated by data flow, transactional data storage, compensation, exception handling, and time aspects \cite{RITTER2019101439}.

In previous work \cite{DBLP:conf/debs/0001MFR18}, we found that already simple integration processes show improvement potential, \eg when considering data dependencies that allow for (sub-)process parallelization.
In order to consider such improvements, it is crucial to also consider data flow in the model, but approaches for verification and formal analysis of \enquote{realistic data-aware} integration processes are currently missing, as recent surveys on event data~\cite{DBLP:journals/corr/AbiteboulABBCDH17,eyers_et_al:DR:2017:6910}, workflow management~\cite{kougka2014optimization}, and in particular application integration~\cite{Ritter201736} report.
Such approaches are needed in order to formally prove the structural and functional correctness of compositions of patterns and their optimizations, which in turn is needed to enable a responsible development of integration scenarios where integration processes behave as intended.

To enable such approaches for the process modeller, we propose a \emph{responsible  composition and optimization (ReCO) process} for patterns, that covers the following objectives: (i) inherently correct structural process representation, (ii) means for representing and proving functional process execution correctness, (iii) semantic integration pattern aspects of control and data flow, transactional data storage, compensation, exception handling, and time, (iv) automatic identification and application of optimizations, and (v) correctness-preserving process changes.
We argue that existing approaches do not fully support responsible integration pattern composition and optimization with correctness preserving guarantees (cf.\ related work in \cref{sec:relatedwork}). 

\begin{figure}[bt]
	\centering
	\includegraphics[width=1\linewidth]{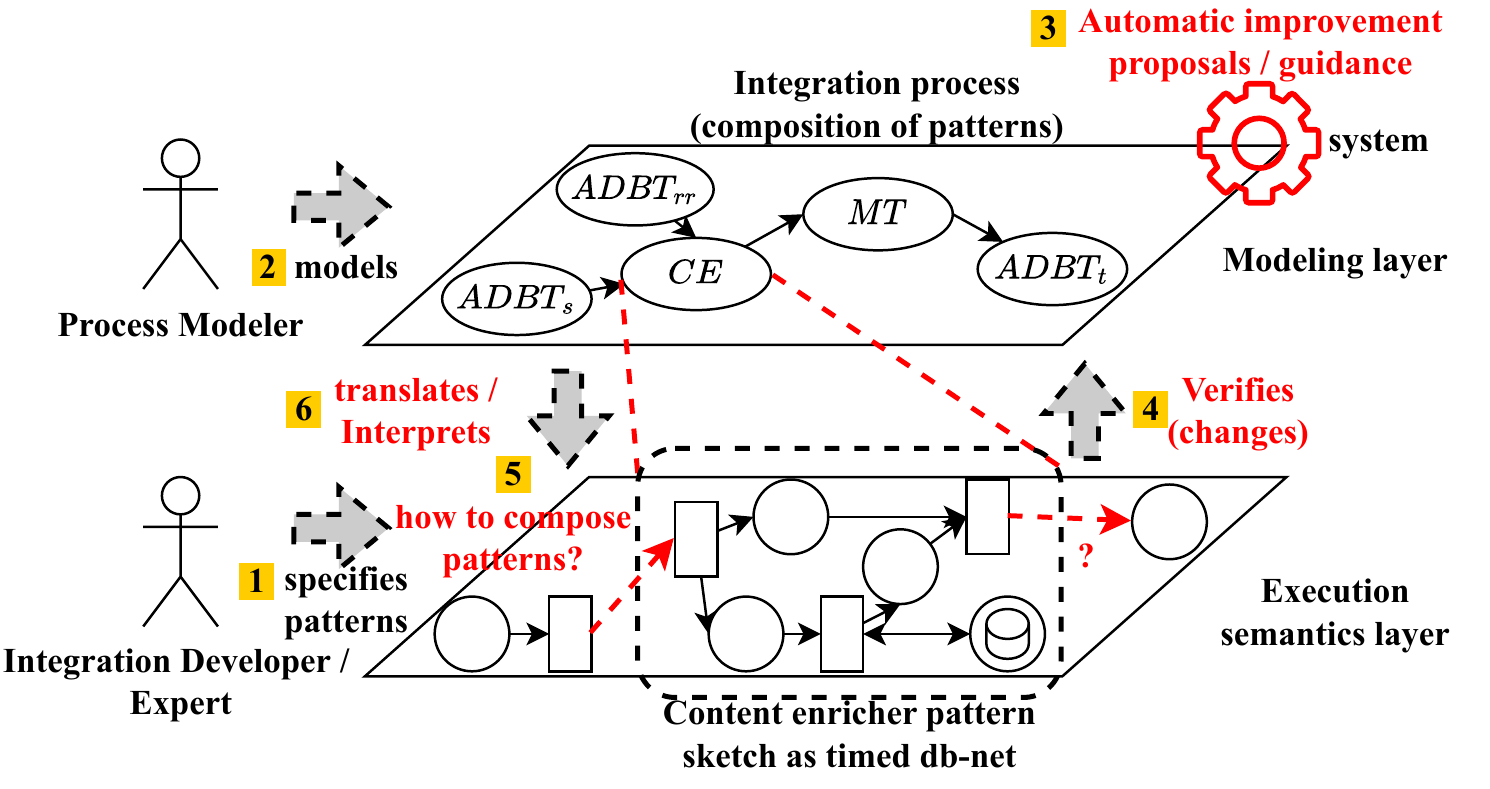}
	\caption{End-to-end perspective from integration process modeling to verifiable execution semantics and automatic, correctness-preserving improvements (current gaps or missing aspects in red color).}
	\label{fig:reco:problem_statement}
\end{figure}
\newcommand{\yellowbox}[1]{\colorbox{amber!70}{\textbf{#1}}}
\Cref{fig:reco:problem_statement} visualises the challenges that need to be overcome to enable such a responsible composition and optimization process achieving objectives (i)--(v).
Integration developers and experts provide semantically meaningful pattern realizations in expressive specialised languages such as timed db-nets~\cite{RITTER2019101439} (\yellowbox{1}). However, these languages are cumbersome to use for process modelers, and makes an automatic identification of improvements difficult.
Process modelers on the other hand like higher level languages and notations to specify processes as a composition of integration patterns~\cite{DBLP:conf/caise/0001H15,ritter2016exception,Ritter201736,DBLP:conf/debs/RitterDMR17} (\yellowbox{2}).
While improvements are conveniently defined on the higher level modeling layer for automatic changes to processes~\cite{DBLP:conf/debs/0001MFR18} (\yellowbox{3}) in manual modeling and automatic improvement cases, it is currently not possible to verify that the improved process has the same behaviour as the original process with respect to the execution semantics of the composed integration patterns (\yellowbox{4}).
This is because the composition of pattern definitions is currently not formally specified (\yellowbox{5}), and the modeling and execution semantics layers are not connected (\yellowbox{6}).

This work aims to fill these gaps, based on the following research questions that guide the design and development of a responsible composition and optimization process:
\begin{itemize}
    \itemsep0em
    \item[\textbf{Q1}] How can the user be supported and guided during pattern composition and process modeling?
    \item[\textbf{Q2}] When are pattern compositions correct?
    \item[\textbf{Q3}] How to responsibly determine and apply optimizations?
\end{itemize}
Question Q1 is related to objective (i), Q2 to objectives (ii)--(iii), and Q3 to (iv)--(v). In the conference version of this paper~\cite{DBLP:conf/debs/0001MFR18}, we provided the foundations for Q1 (and partially Q3).
Pattern compositions were represented as typed pattern graphs, based on pattern characteristics and contracts, which inherently guarantee structurally correct compositions, and thus guide and support the user.
We could not use existing languages / notations like BPMN or EIP icon notation \cite{hohpe2004enterprise}, since they either structurally or semantically do not provide the required notions of data flow, persistence (\eg  \cite{DBLP:conf/caise/0001H15,ritter2016exception,DBLP:conf/debs/RitterDMR17,DBLP:conf/ecmdafa/Ritter14}) and boundaries for checking structural correctness (\eg \cite{Ritter201736}).
Furthermore, the graph-based representation of integration patterns allows for the realization of optimizations as graph rewriting rules.
Our evaluation showed that effective improvements could be identified and applied to real-world integration processes, while structural correctness was preserved.
However, functional correctness was not considered, meaning that process changes might not be \emph{responsible} (cf.\ objectives (ii)--(iii), (v)).

In this work, we extend our user-facing and structural correctness guaranteeing graph-based representation with an execution semantics using timed DB-nets~\cite{RITTER2019101439}.
To support the same notion of correctness based on pattern contracts as in \cite{DBLP:conf/debs/0001MFR18}, we define a new notion of \emph{open} timed db-nets that are capable of representing the data exchange between patterns. We then show how they can be composed, and specify their execution semantics.
By interpreting pattern compositions in graph representation as compositions of open timed db-nets, and by proving that the translation results in structurally correct and semantically well-behaved nets, we can answer Question Q2.
All in all, this makes automatic optimization of integration processes feasible, 
but now also taking functional correctness into account, thus answering question Q3 fully. Hence this enables the study of ReCO for the first time.

\paragraph{Methodology}
We follow the principles of design science research methodology by Peffers et al.~\cite{peffers2007design} to answer the research questions above: \emph{\enquote{Activity 1: Problem identification and motivation}} is based on a literature review and the assessment of vendor driven solutions~\cite{Ritter201736}, as well as quantitative analysis of integration pattern characteristics of EAI building blocks (existing catalogs of 166 integration patterns) and process improvements~\cite{DBLP:conf/debs/0001MFR18}, resulting in requirements to a suitable formalism.
We then address \emph{\enquote{Activity 2: Define the objectives for a solution}} by formulating objectives (i)--(v). For \emph{\enquote{Activity 3: Design and development}}, we create several artifacts/contributions to answer questions Q1--Q3 and realize objectives (i)--(v):
\begin{enumerate}[label=(\alph*)]
    \item a specification of an extensible structural correctness enforcing representation that allows for efficient application of improvements,
    \item an extension of the definition of the formalism of execution semantics by inter-pattern data exchange analysis capabilities based on open timed db-nets (cf. \yellowbox{4}, \yellowbox{5}),
    \item an interpretation procedure of graph representation as open timed db-nets (cf. \yellowbox{6}), and
    \item optimization realizations on the graph representation leveraging the interpretation to prove their correctness (cf. \yellowbox{3}).
\end{enumerate}

\paragraph{Outline}
We introduce the ReCO process in \cref{sec:ReCO}, together with an integration process modeling example.
In \cref{sec:background}, we analyze recurring integration pattern characteristics, which are relevant for developing our formalisms, and collect optimization strategies for integration processes. We also identify eight requirements for formalizing integration pattern compositions.
In \cref{sec:formalization}, we describe integration pattern graphs, and use them to specify pattern compositions with inherent structural correctness. We also give an abstract cost model which can be used to determine if an optimization is an improvement or not.
\Cref{sec:semantics} extends the timed DB-net formalism~\cite{RITTER2019101439} to open nets to introduce compositional aspects, and uses this extended formalism to capture the dynamics of integration patterns --- that is, how data flows through the system. We also briefly describe an implementation of simulation of timed DB-nets as an extension of CPN Tools.
In \cref{sec:interpretation}, we combine the two formalisms by showing how integration pattern graphs can be interpreted as open timed DB-nets.
Next in \cref{sec:realization} we realize optimization strategies as rewrite rules for integration pattern graphs, and show 
that these optimisations preserve the functional correctness of patterns when interpreted as timed DB-nets.
In \cref{sec:evaluation}, we evaluate the optimizations on a large body of integration processes, and apply the ReCO process to two integration processes.
We discuss related work in \cref{sec:relatedwork} and conclude in \cref{sec:conclusion}.

 \paragraph{History of this paper} This paper extends our previous conference paper~\cite{RITTER2019101439} with several new contributions. Firstly, this paper is structured along
 the ReCO process, whereas the conference version was not. Secondly, the dynamic semantics in the form of open timed DB-nets and their implementation as an extension of CPN Tools (\cref{sec:semantics}), and the interpretation of integration pattern graphs into it (\cref{sec:interpretation}) are new, together with the proof that the optimizations proposed preserve the functional behaviour of the corresponding timed DB-nets (\cref{sub:op:correctness}).
The case studies of the ReCO process in \cref{sub:pc:evaluation} are also new in this version.

\section{Responsible Composition and Optimization Process for Patterns}
\label{sec:ReCO}
To reason about responsible pattern composition and optimization, one first needs to consider formalizations of \emph{patterns}, \emph{compositions}, and \emph{composition improvements}.
In this section, we introduce a pattern composition and optimization process that covers all of these topics, and we describe the process modeling in more detail by an integration process example.
Finally, we specify the problem of responsible pattern composition and optimization.

\newcommand{\bluebox}[1]{\colorbox{Aquamarine!70}{\textbf{#1}}}

\subsection{Responsible Composition and Optimization Process}
Our approach to a responsible pattern composition and optimization process is described in \cref{fig:opt:process}.
We briefly discuss the main concepts involved.

\begin{figure}[bt]
	\centering
	\includegraphics[width=1\linewidth]{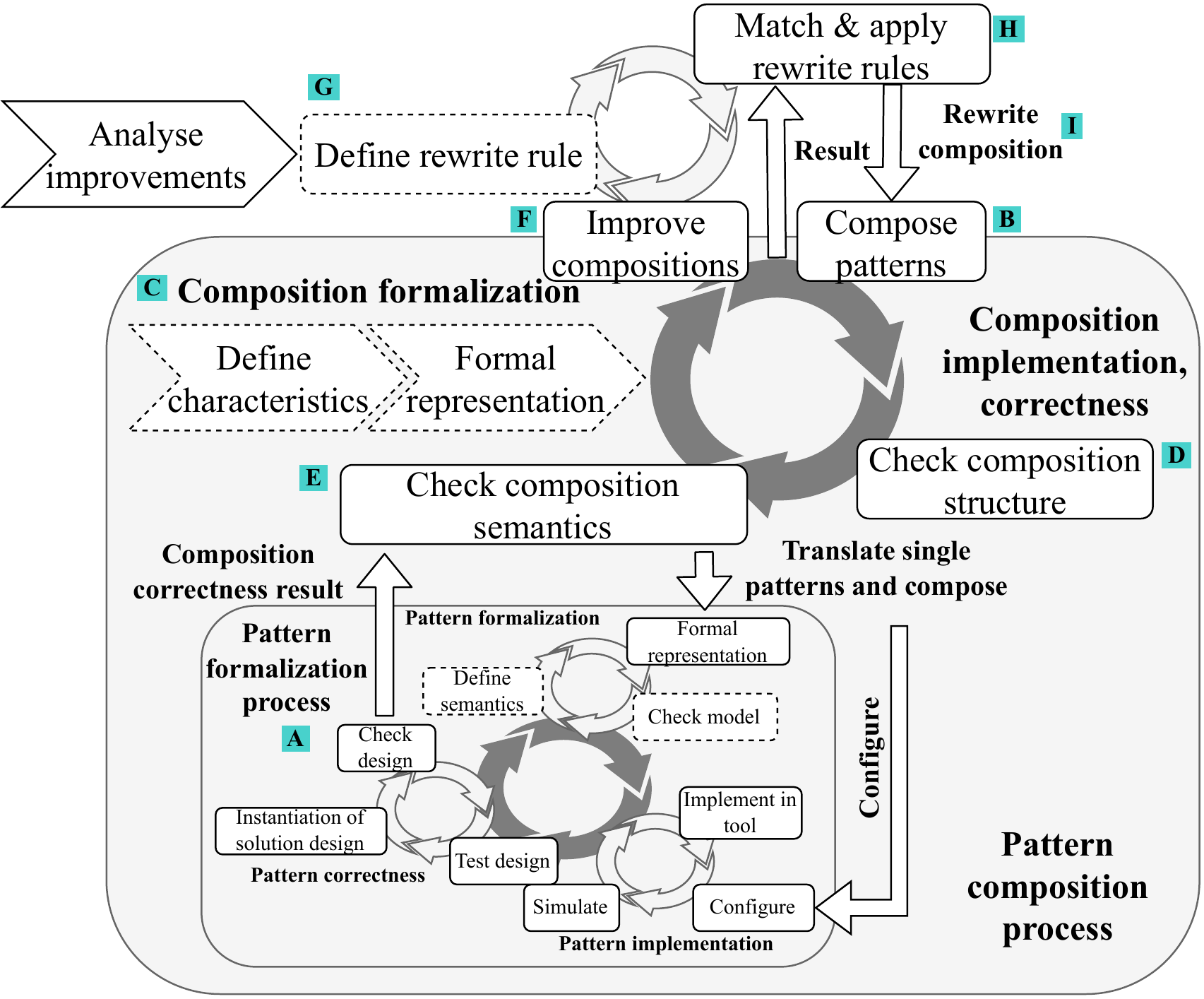}
	\caption{Responsible pattern composition and optimization process (white colored boxes denote the main steps, grey colored box shows the \emph{pattern formalization process} from \cite{RITTER2019101439}, to which we bridge through translation of single patterns and compositions)}
	\label{fig:opt:process}
\end{figure}
\begin{figure*}[bt]
	\centering
	\includegraphics[width=.8\linewidth]{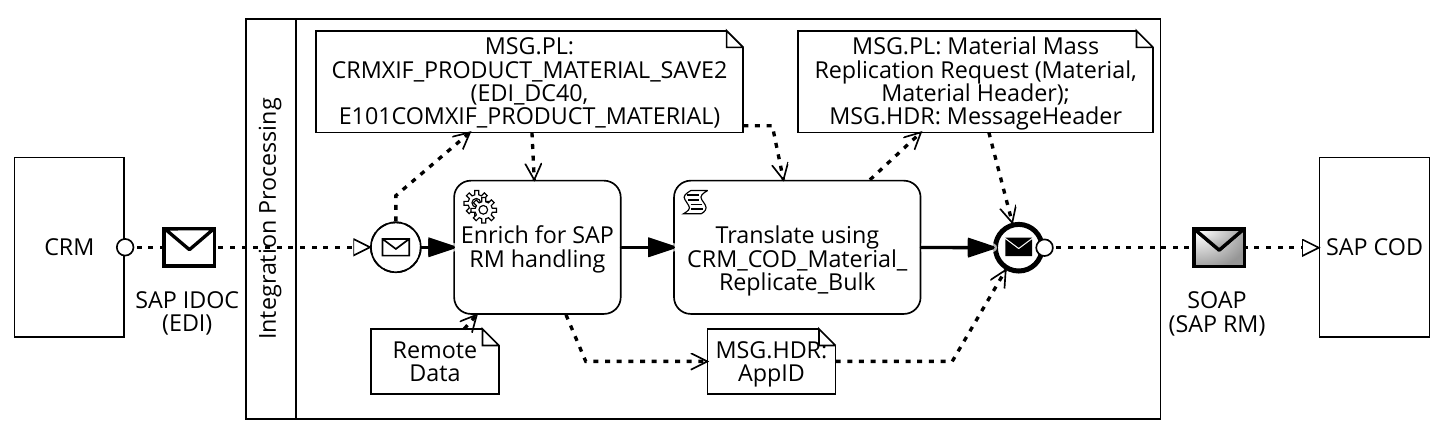}
	\caption[Replicate material from SAP Business Suite (in BPMN)]{Replicate material from SAP Business Suite (a \emph{hybrid integration} scenario in BPMN)}
	\label{fig:pn:c4c}
\end{figure*}
\begin{description}
    \item[Patterns] The \emph{pattern formalization process} (cf. \bluebox{A} in \cref{fig:opt:process}) is specified in \cite{RITTER2019101439}.
    For ReCO, patterns are provided as modeling building blocks by the integration platform.
    Integration experts and platform developers start developing patterns by defining their semantics 
    in formal representation that allows to formally analyze the pattern models. 
    In addition, the realization of a pattern can be configured and simulated as well as tested.
    \item[Composition] In ReCO, a process is modeled and configured by a user/modeler in a process modeling tool by composing (cf. \bluebox{B} \emph{compose patterns}) existing, single, semantically correct, integration patterns, 
    from which the model is then translated into a formal representation by the integration platform (cf. \bluebox{C}).
    With a thorough understanding of pattern characteristics within a process, composition can be formalized with inherent structural correctness (cf. \bluebox{D} \emph{check composition structure}).
    Semantic correctness of the whole composition of patterns can be checked through a notion of composition added to the pattern 
    formalization (cf. \bluebox{E} \emph{check composition semantics}).
    \item[Improvements] The formal foundation of structurally and semantically correct processes allows for proposed improvements that preserve the processes' correctness (cf. \bluebox{F} \emph{improve compositions}).
    For ReCO, the improvements could be defined by integration experts, but also users in a formalism provided by the integration platform that allows for process rewriting suitable to the underlying formalisms (cf. \bluebox{G} \emph{define rewrite rule}).
    The improvements are automatically matched (cf. \bluebox{H} \emph{match \& apply rewrite rules}), and if applicable, it leads to the application of an improvement (cf. \bluebox{I} \emph{rewrite composition}).
\end{description}

To summarise, our vision is that integration experts and platform developers specify patterns as building blocks of an integration platform (cf. \bluebox{A}),
where the aspect of correct composition of multiple patterns into processes in terms of structure and semantics is based on a formalism used by the platform (cf. \bluebox{B}--\bluebox{E}).
Correctness-preserving process improvements should be developed and proved to be correct by integration experts, and automatically applied by the platform (cf. \bluebox{F}--\bluebox{I}).

\begin{figure*}[bt]
	\centering
	\includegraphics[width=0.7\linewidth]{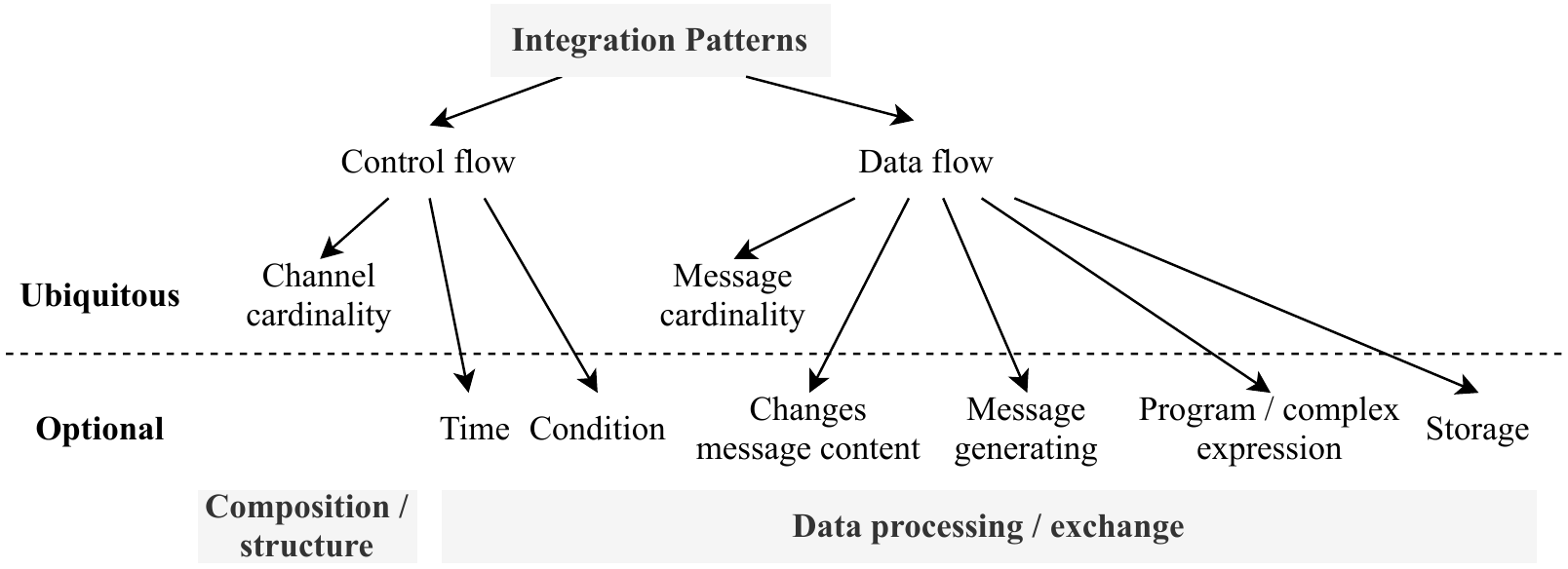}
	\caption{Categorizing integration pattern characteristics according to control and data flow}
	\label{fig:characteristics0}
\end{figure*}

\subsection{Potential Process Optimization by Example}
We now introduce the composition stage by example of an integration process that could be improved using ReCO.

Many organizations have started to connect their on-premise applications such as Customer Relationship Management (CRM) systems with cloud applications such as SAP Cloud for Customer (COD) using integration processes similar to the one shown in \cref{fig:pn:c4c}.
A \emph{CRM Material} is sent from the CRM system via EDI (more precisely the SAP IDOC transport protocol) to an integration process running on SAP Cloud Platform Integration (CPI) \cite{sap-hci-content}.
The integration process enriches the message header (MSG.HDR) with additional information based on a document number for reliable messaging (\ie AppID), which allows redelivery of the message in an exactly-once service quality~\cite{ritter2016exception}.
The IDOC structure is then mapped to the COD service description and sent to the COD receiver.

Already in this simple integration process, 
an obvious improvement can be applied: the data-independent Content Enricher and Message Translator patterns~\cite{hohpe2004enterprise} could be executed in parallel.
Importantly, such a change does not alter the behaviour of the integration process.

In this paper, we seek to find a mechanism to combine the inherently, structurally correct pattern composition formalism from \cite{DBLP:conf/debs/0001MFR18} with the work on timed db-nets \cite{RITTER2019101439} that allow for semantically correct definitions of integration patterns, and to prove that improvements are correctness-preserving.
The interaction between the user/modeler and the integration system requires a ReCO process that addresses the objectives (i)--(v).


%% file: optimization_strategies.tex
\section{Background and Requirements}
\label{sec:background}
In this section, we give a brief background on application integration patterns and their optimizations, by analyzing recurring pattern characteristics as well as collecting existing optimizations as strategies.
We also derive and discuss requirements for a suitable formalism for pattern compositions in the context of the optimization strategies.

\subsection{Integration Pattern Characteristics}
\label{sec:characteristics}
Enterprise integration patterns (EIPs)~\cite{hohpe2004enterprise} with recent additions~\cite{ritter2016exception,Ritter201736} form a suitable and important abstraction when implementing application integration scenarios.
Besides their original differentiation in functional categories such as \emph{message channels}, \emph{message routers}, \emph{message transformations}, and \emph{messaging endpoints} among others, there are more subtle means of classifying patterns by \emph{pattern characteristics} that consider the control and data flow within and between integration patterns, and thus help greatly when formalizing pattern compositions.

We analyzed all patterns from the literature \cite{hohpe2004enterprise,ritter2016exception,Ritter201736} regarding their control and data flow characteristics.
Our findings are summarised in \cref{fig:characteristics0}.
The characteristics of \emph{channel} and \emph{message cardinality} (CC and MC, respectively) are ubiquitous and can be found in every pattern. We also identified a number of optional and non-exclusive characteristics: if the pattern \emph{changes message contents} (CHG), if it is \emph{message generating} (MG), if it has \emph{conditions} (CND), and if it has \emph{programs / complex expressions} (PRG).
Additionally, \emph{time} and \emph{storage} (also found in \cite{RITTER2019101439}) are important requirements, especially for our later considerations on runtime semantics, but on the level of compositions, they are not very relevant due to their pattern local nature, and are subsumed under control and data flow.

Together these characteristics summarize all relevant control and data flow properties of the considered integration patterns.
In this work, composition and structure becomes relevant for checking structural correctness properties, while data processing and exchange characteristics are required mostly for a composition's semantic correctness.
Both notions of correctness are especially relevant for modeling as well as any improvements to the composition (\eg in the form of optimizations).

\subsubsection{Ubiquitous Characteristics}
Every pattern has both a channel and a message cardinality, covering control and data flow aspects that are relevant for composing patterns.

\paragraph{Channel Cardinality}
The channel cardinality specifies the number of incoming and outgoing message channels of an integration pattern.
It is especially important for the structural correctness of a pattern composition.
The relative number of patterns of each channel cardinality can be found in \cref{fig:characteristics1}.
A \emph{zero-to-one} or \emph{one-to-zero} cardinality is exclusively found in 
\emph{start}- and \emph{end} components of a composition, like message endpoints (\eg commutative endpoint \cite{Ritter201736}) or integration adapters (\eg event-driven consumer \cite{hohpe2004enterprise}).
Most of the transformation patterns and some of the routing patterns are \emph{message processors} that have a channel cardinality of \emph{one-to-one} (\eg aggregator, splitter \cite{hohpe2004enterprise}).
The remainder of the routing patterns are either \emph{forks} with cardinality \emph{one-to-many} (\eg multicast), \emph{conditional forks} with \emph{one-to-many (cond.)} (\eg content-based router \cite{hohpe2004enterprise}), or \emph{joins} with cardinality \emph{many-to-one} (\eg join router \cite{Ritter201736}).
We found no \emph{many-to-many} patterns at all.
\begin{figure}[bt]
	\centering
	\includegraphics[width=0.9\columnwidth]{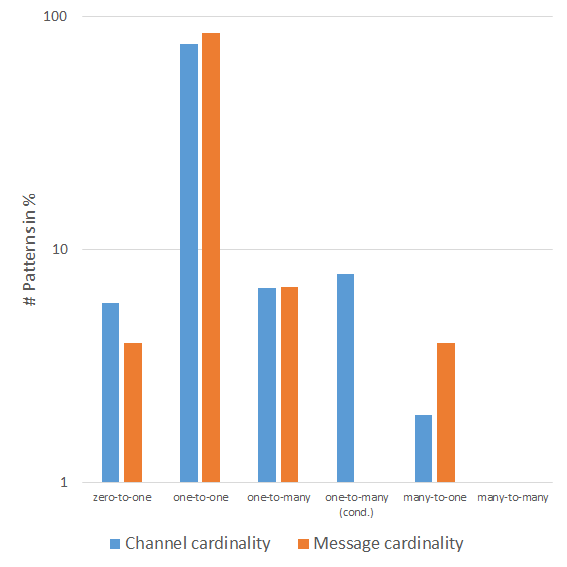}
	\caption{Channel and message cardinalities (\enquote{zero-to-one} includes \enquote{one-to-zero})}
	\label{fig:characteristics1}
\end{figure}

\paragraph{Message cardinality} Similar to the channel cardinality, the message cardinality specifies the number of incoming and outgoing messages.
However, the message cardinality is relevant for specifying the data transfer between patterns in a composition.
As summarised in \cref{fig:characteristics1}, we found that most of the integration patterns receive or require one and emit one message (\eg message translator, message signer).
There are also patterns that require one message and emit several messages (\eg splitter, multicast) and similarly receive many and emit only one (\eg aggregator).
Finally, there are patterns that require zero and emit one or vice versa (\eg event-driven consumer or producer).
Note that there are no patterns with a message cardinality of \emph{one-to-many (cond.)} or \emph{many-to-many}.

\subsubsection{Optional Characteristics}
We found that certain data flow characteristics are only present in some patterns.
The number of patterns that have each identified characteristic can be found in \cref{fig:characteristics2}.
\begin{figure}[bt]
	\centering
	\includegraphics[width=0.9\columnwidth]{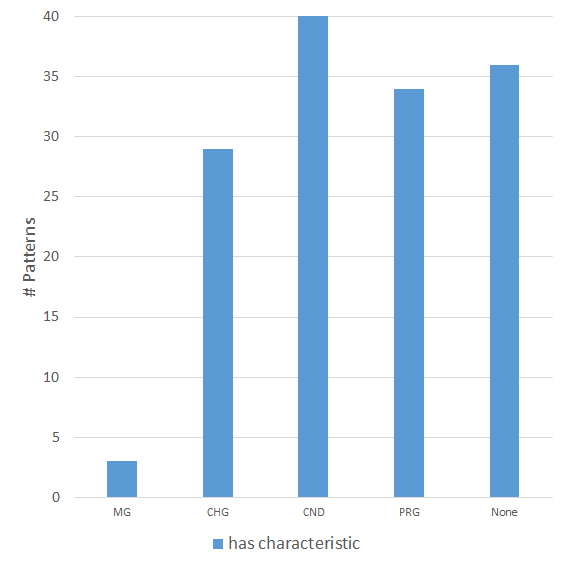}
	\caption{Further characteristics and configurations}
	\label{fig:characteristics2}
\end{figure}

\paragraph{Message Generating} A  pattern is message generating, if it does not simply pass or alter the received message, but creates a completely new message (\eg aggregator, splitter).
The new message might preserve data and structure from a received message, but will have a new message identifier.
If a newly generated message is composed out of several received messages, the lineage to the original messages is preserved by adding a message history \cite{hohpe2004enterprise}.
We identified only three message generating patterns.

\paragraph{Changing Message Content} In some cases it might be interesting to distinguish between integration patterns that actually change the content of received messages (\ie mostly, but not exclusively, message transformation patterns) and read-only patterns, which either read data from the message for evaluating a condition (\eg content-based router) or do not look into the message at all (\eg multicast \cite{Ritter201736}, claim check \cite{hohpe2004enterprise}).
We found 29 message changing patterns.

\paragraph{Conditions} A condition is a read-only, user-defined function that returns a Boolean valuation.
Conditions are mostly used in routing patterns, which decide on route or no-route, when receiving a message.
We found 40 patterns that require a condition to function.

\paragraph{Program/Complex Expressions} A program or complex expression is an arbitrary, user-defined function that might change the content of a message by local or remote program execution (incl. remote procedure and database calls), which we found in $34$ patterns.
As can be seen in \cref{fig:characteristics2}, there are more patterns requiring an expression than there are patterns changing the content of a message.
Consequently, there is a small number of read-only patterns that require more complex processing, potentially with side-effects (\eg load balancer \cite{Ritter201736}).

\paragraph{None} Since the identified characteristics are optional and non-exclusive, there is also a significant number of $36$ patterns that do not have any of the characteristics (\eg the detour \cite{Ritter201736}, wire tap \cite{hohpe2004enterprise}).

\subsubsection{Summary}
Analysing all known application integration patterns from the literature \cite{hohpe2004enterprise,ritter2016exception,Ritter201736}, we identified and categorized several control and data flow characteristics, relevant for the structural composition of those patterns and their internal and inter-pattern data flow.
Channel cardinality and message cardinality are relevant to all patterns, but other characteristics are optional and non-exclusive.

\subsection{Static Optimization Strategies}
\label{sec:stategies}
We consider improvements very important in the context of EAI, and thus we briefly survey recent attempts to optimize composed EIPs, in order to motivate the need to formalize their semantics.
As a result, we derive so far unexplored prerequisites 
for optimizing compositions of EIPs.
A more detailed collection of optimization strategies can be found in a non-mandatory supplementary material~\cite{DBLP:journals/corr/abs-1901-01005}.
\begin{table*}[tb]
	\centering
	\small
	\caption{Optimizations in related domains --- horizontal search}
	\label{tab:systemreview}
	\begin{tabular}{lccll}
		\hline
		\parbox[t]{3.5cm}{Keyword} & hits & selected & Selection criteria & Selected Papers \\
		\hline
		\parbox[t]{3.5cm}{Business Process Optimization} & $159$ & $3$ & \parbox[t]{4cm}{data-aware processes} & \parbox[t]{7.5cm}{survey \cite{vergidis2008business}, optimization patterns~\cite{niedermann2011business,niedermann2011deep}} \\
		\parbox[t]{3.5cm}{Workflow Optimization} & $396$ & $6$ & \parbox[t]{4cm}{data-aware processes} & \parbox[t]{7.5cm}{instance scheduling \cite{agrawal2010scheduling,bittencourt2011hcoc,tirapat2013cost}, scheduling and partitioning for interaction \cite{ahmad2014data}, scheduling and placement \cite{benoit2012throughput}, operator merge~\cite{habib2013adapting}} \\
		\parbox[t]{3.5cm}{Data Integration Optimization} & $61$ & $2$ & \parbox[t]{4cm}{data-aware processes optimization, (no schema-matching)} & \parbox[t]{7.5cm}{instance scheduling, parallelization \cite{zhang2012cost}, ordering, materialization, arguments, algebraic \cite{getta2011static}} \\
		\parbox[t]{3.5cm}{} & & & &\\
		\parbox[t]{3.5cm}{Added} & n/a & $13$ & expert knowledge & \parbox[t]{7.5cm}{business process \cite{vrhovnik2007approach}, workflow survey \cite{kougka2014optimization,DBLP:journals/corr/KougkaGS17}, data integration \cite{bohm2008model}, distributed applications \cite{DBLP:journals/is/BohmK11}, EAI \cite{ritter2016exception,Ritter201736,DBLP:conf/debs/RitterDMR17,DBLP:conf/bncod/Ritter17}, placement \cite{DBLP:conf/edoc/000120,DBLP:journals/is/Ritter22}, resilience \cite{Nygard:2007:RDD:1200767}} \\
		\parbox[t]{3.5cm}{Removed} & - & $1$ &  & \parbox[t]{7.5cm}{classification only \cite{vergidis2008business}} \\
		\parbox[t]{3.5cm}{} &          &               & &  \\

		\parbox[t]{3.5cm}{Overall} & $616$ & $23$ &  &\\
	\end{tabular}
\end{table*}

\subsubsection{Identifying Optimization Strategies}
Since a formalization of the EAI foundations in the form of integration patterns for static optimization of \enquote{data-aware} pattern processing is missing~\cite{Ritter201736}, we conducted a horizontal literature search~\cite{kitchenham2004procedures} to identify optimization techniques in related domains.
For EAI, the domains of business processes, workflow management and data integration are of particular interest.
The results of our analysis are summarized in~\cref{tab:systemreview}.
Out of the resulting $616$ hits, we selected $18$ papers according to the search criteria \enquote{data-aware processes}, and excluded work on unrelated aspects.
This resulted in the \emph{seven} papers listed in \cref{tab:optimization_strategies}.
The mapping of techniques from related domains to EAI was done by for instance taking the idea of projection push-downs~\cite{niedermann2011business,habib2013adapting,getta2011static,vrhovnik2007approach,DBLP:journals/is/BohmHPLW11} and deriving the early-filter or early-mapping technique in EAI.
We categorized the techniques according to their impact (\eg structural or process, data-flow) in context of the objectives for which they provide solutions.

In the following subsections, we now briefly discuss the optimization strategies listed in \cref{tab:optimization_strategies}, in order to derive prerequisites needed for optimizing compositions of EIPs.
To relate to their practical relevance and coverage so far (in the form of evaluations on \enquote{real-world} integration scenarios), we also discuss existing \enquote{data-aware} message processing solutions for each group of strategies.

\newcommand{\improved}{$\uparrow$} 
\newcommand{\slightlyImproved}{$\nearrow$} 
\newcommand{\notImproved}{$\searrow$} 

\begin{table*}[bt]
	\centering
	\small
	\caption{Optimization Strategies in the Context of Integration Processes}
	\label{tab:optimization_strategies}
	\begin{tabular}{ll|ccc||c}
		\hline
          Strategy & Optimization & Throughput & Latency & Complexity & Practical Studies \\
		\hline
		\multirow{2}{*}{\parbox[t]{3.0cm}{OS-1: Process \\ Simplification}} & {Redundant Sub-process Removal \cite{DBLP:journals/is/BohmHPLW11}} &   & \improved & \improved & -\\
			& {Combine Sibling Patterns \cite{habib2013adapting,DBLP:journals/is/BohmHPLW11}} &   & \improved & \improved & (\cite{DBLP:conf/debs/RitterDMR17}) \\
			& {Unnecessary conditional fork \cite{vrhovnik2007approach,DBLP:journals/is/BohmHPLW11}} & \slightlyImproved & \improved & \improved & -\\
		\hline
		\multirow{2}{*}{\parbox[t]{3.0cm}{OS-2: Data \\ Reduction}} & {Early-Filter \cite{niedermann2011business,habib2013adapting,getta2011static,vrhovnik2007approach,DBLP:journals/is/BohmHPLW11}} & \improved &   &   & \cite{DBLP:conf/debs/RitterDMR17}\\
		& {Early-Mapping \cite{habib2013adapting,getta2011static,DBLP:journals/is/BohmHPLW11}} & \improved &   &   & \cite{DBLP:conf/debs/RitterDMR17,edoc2017}\\
		& {Early-Aggregation \cite{habib2013adapting,getta2011static,DBLP:journals/is/BohmHPLW11}} & \improved &   &   & \cite{edoc2017}\\
		& {Claim Check \cite{getta2011static,DBLP:journals/is/BohmHPLW11}} & \improved &   & \notImproved & -\\
		& {Early-Split \cite{DBLP:conf/debs/RitterDMR17}} & \improved &   & \notImproved & \cite{DBLP:conf/debs/RitterDMR17,edoc2017}\\
		\hline
		\parbox[t]{3.0cm}{OS-3: Parallelization} &
		{Sequence to parallel \cite{niedermann2011business,zhang2012cost,vrhovnik2007approach,DBLP:journals/is/BohmHPLW11}} & \improved &   & \notImproved & \cite{DBLP:conf/debs/RitterDMR17,DBLP:conf/bncod/Ritter17}\\
		& {Merge parallel sub-processes \cite{niedermann2011business,zhang2012cost,vrhovnik2007approach,DBLP:journals/is/BohmHPLW11}} &   & \improved & \improved & \cite{DBLP:conf/debs/RitterDMR17}\\
        & Heterogeneous parallelization \cite{benoit2012throughput} & \improved &   & \notImproved & - \\ 
        \hline
        {\parbox[t]{3.0cm}{OS-4: Pattern \\ Placement}} & {\parbox[t]{4.5cm}{Pushdown to Endpoint \\ (extending OS-2)}} & \improved  & & (\improved) & \cite{DBLP:conf/bncod/Ritter17,DBLP:conf/edoc/000120,DBLP:journals/is/Ritter22} \\ 
        \hline
        \multirow{2}{*}{\parbox[t]{3.0cm}{OS-5: Reduce \\ Interaction}} & Ignore Failing Endpoints \cite{ritter2016exception,Ritter201736,Nygard:2007:RDD:1200767} & & \improved & & - \\ 
        & Reduce Requests \cite{vrhovnik2007approach} & \slightlyImproved & \improved & & -\\
	\end{tabular}
\begin{tablenotes}
	\centering
	\small
	\item \improved: improvement, \slightlyImproved: slight improvement, \notImproved: slight deterioration, \textvisiblespace: no effect
\end{tablenotes}
\end{table*}

\subsubsection{Process Simplification}
We grouped together techniques whose main goal is reducing model complexity (\ie number of patterns) under the heading of process simplification.
The cost reduction of these techniques can be measured by pattern processing time (latency, \ie time required per operation) and model complexity metrics~\cite{sanchez2010prediction}.
Process simplification can be achieved by removing redundant patterns like \emph{Redundant Subprocess Removal} (\eg removing one of two identical sub-flows), \emph{Combine Sibling Patterns} (\eg removing one of two identical patterns), or \emph{Unnecessary Conditional Fork} (\eg removing redundant branching).
As far as we know, the only practical study of combining sibling patterns can be found in Ritter et al.~\cite{DBLP:conf/debs/RitterDMR17}, showing moderate throughput improvements.
The simplifications requires a formalization of patterns as a control graph structure, which helps to identify and deal with the structural changes. 
Previous work targeting process simplification include B{\"{o}}hm et al.~\cite{DBLP:journals/is/BohmHPLW11} and Habib, Anjum and Rana~\cite{habib2013adapting}, who use evolutionary search approaches on workflow graphs, and Vrhovnik et al.~\cite{vrhovnik2007approach}, who use a rule formalization on query graphs.

\subsubsection{Data Reduction}
The reduction of data can be facilitated by pattern push-down optimizations of message-element-cardinality-reducing patterns, which we call \emph{Early-Filter} (for data; \eg removing elements from the message content), \emph{Early-Mapping} (\eg applying message transformations), as well as message-reducing optimization patterns like \emph{Early-Filter} (for messages; \eg removing messages), \emph{Early-Aggregation} (\eg combining multiple messages), \emph{Early-Claim Check} (\eg storing content and claiming it later without passing it through the pipeline), and \emph{Early-Split} (\eg cutting one large message into several smaller ones).
Measuring data reduction requires a cost model based on the characteristics of the patterns, as well as the data and element cardinalities.
For example, the practical realizations for multimedia~\cite{edoc2017} and hardware streaming~\cite{DBLP:conf/debs/RitterDMR17} show improvements especially for early-filter, split and aggregation, as well as moderate improvements for early-mapping.
This requires a formalization that is able to represent data or element flow. 
Data reduction optimizations target message throughput improvements (\ie processed messages per time unit), however, some have a negative impact on the model complexity.
Previous work on data reduction include Getta~\cite{getta2011static}, based on relational algebra, and Niedermann, Radesch{\"u}tz and Mitschang~\cite{niedermann2011business}, who define optimizations algorithmically for a graph-based model.

\subsubsection{Parallelization}
Parallelization of processes can be facilitated through transformations such as \emph{Sequence to Parallel} (\eg duplicating a pattern or sequence of pattern processing), or, if not beneficial, reverted, \eg by \emph{Merge Parallel}.
For example, good practical results have been shown for vectorization~\cite{DBLP:conf/bncod/Ritter17} and hardware parallelization~\cite{DBLP:conf/debs/RitterDMR17}.
Formalizing such operations again require a control graph structure.
Although the main focus of parallelization is message throughput, heterogeneous variants also improve latency. 
In both cases, parallelization requires additional patterns, which negatively impacts the model complexity, whereas merging parallel processes improves the model complexity and latency. 
Previous work on pattern parallelization include Zhang et al.~\cite{zhang2012cost}, who defines a service composition model, to which algorithmically defined optimizations are applied.

\subsubsection{Pattern Placement}
For all of the data reduction optimizations (cf. OS-2) \enquote{Pushdown to Endpoint} can be applied by extending the placement to the message endpoints, 
which improves message throughput and reduces the complexity of the integration process, while increasing the complexity at the message endpoints.
For example, good practical results have been shown for vectorization~\cite{DBLP:conf/bncod/Ritter17} and cost efficient, security-aware placement \cite{DBLP:conf/edoc/000120,DBLP:journals/is/Ritter22}.

\subsubsection{Reduce Interaction}
The resilience and robustness of integration processes is crucial -- especially in the cloud.
Dependencies to resources used by an integration process (\eg database, message queuing) and the message endpoints (\eg mobile, cloud) has to be dealt with during the message processing.
Optimizations like \emph{Ignore Failing Endpoints} and \emph{Reduce Requests} help dealing with potentially unreliable network communication, and allow for smart network usage and reaction to exceptional situations or unavailabilities.
This requires a formalism that is able to represent data flow and time.
These optimizations target more stable processes, improved latency and potentially higher throughput.

\subsubsection{Summary}
We found several optimizations in related domains like data integration, business process and workflow, and re-interpreted them in the context of EAI.
The optimizations have different effects regarding relevant categories like throughput, latency and model complexity, which we used to categorize them into three disjoint classes of optimization strategies, namely process simplification, data reduction and parallelization.
For most of the optimizations we identified implementations that report a successful application to problems in the EAI domain.

\subsection{Discussion: Requirements for Formalizing Integration Pattern Compositions}
\label{sec:requirements}
Based on our analyses of characteristics of single patterns and optimization strategies for pattern compositions, we briefly discuss requirements for the formalization of pattern compositions and their improvements as optimizations.
These requirements are listed and set into context to the closest related approaches known from the application and data integration domains in \cref{tab:requirements}.
We also give examples of optimization strategies (cf. OS-x) that are co-enabled when fulfilling the requirements.
\newcommand{\covered}{\checkmark} 
\newcommand{\partiallyCovered}{(\checkmark)} 
\newcommand{\textinvisiblespace}{}
\begin{table*}[bt]
	\centering
	\small
	\caption{Formalization requirements}
	\label{tab:requirements}
	\begin{tabular}{l p{5.2cm} ccc}
		\hline
        ID & Requirement & Fahland et al. \cite{DBLP:conf/caise/FahlandG13} & Ritter et al. \cite{RITTER2019101439} & B\"ohm et al. \cite{DBLP:journals/is/BohmHPLW11} \\
		\hline
		\req{controlflow} & Control flow (pipes and filter) & \covered & \covered & \partiallyCovered \\
		\req{time} & Time & \textinvisiblespace & \covered & \textinvisiblespace \\
		\hline
		\req{dataflow} & Data flow & \covered & \covered & \covered \\
        \hline
		\hline
		\req{database} & Database, Transactions & \textinvisiblespace & \covered & \textinvisiblespace \\
		\req{exceptions} & Exceptions, Compensations & \textinvisiblespace & \covered & \textinvisiblespace \\
		\hline
		\req{compositional} & Compositional & \partiallyCovered & \partiallyCovered & \partiallyCovered \\
		\req{improvements} & Improvements (control and data flow) & \textinvisiblespace & \textinvisiblespace & \partiallyCovered \\
		\req{correctness} & Preserving correctness  & \textinvisiblespace & \textinvisiblespace & \textinvisiblespace \\
	\end{tabular}
\begin{tablenotes}
	\centering
	\small
	\item \covered: covered, \partiallyCovered: partially covered, \textvisiblespace: not covered or out of scope
\end{tablenotes}
\end{table*}

We found that a fundamental support for control flow is mandatory for pattern compositions, due to the pipes-and-filters nature of integration scenarios \cite{hohpe2004enterprise,Ritter201736} (\eg co-enabling OS-1).
Hence, a suitable formalism representing pattern compositions requires a formal representation (\textbf{\refreq{controlflow}: Formal representation of control flow}), \eg as found in first formalization attempts using Coloured Petri Nets by Fahland and Gierds \cite{DBLP:conf/caise/FahlandG13} and our recent work on timed db-nets \cite{RITTER2019101439}, which essentially denotes an improvement of previous work regarding formal representation and analysis properties like data flow, time, transactional data storage and exceptions.
The work by B\"ohm et al. \cite{DBLP:journals/is/BohmHPLW11} stems from the data integration domain, and thus only has a weak notion of control flow and none of integration patterns.
As also found in \cite{RITTER2019101439}, the known integration patterns require different aspects of time like timeouts, delays, and time-based rate limits to be functional (\textbf{\refreq{time}: Formal treatment of time}), \eg co-enabling OS-5.

The concept of pipes-and-filters also requires support for the flow and processing of messages (\textbf{\refreq{dataflow}: Formal representation of data flow}), which is also supported by the closest known formalizations, using Coloured Petri nets in \cite{DBLP:conf/caise/FahlandG13}, plus an extension to db-nets \cite{DBLP:journals/topnoc/MontaliR17} in \cite{RITTER2019101439}, and a data dependency tree in \cite{DBLP:journals/is/BohmHPLW11} (\eg co-enabling OS-2).

Another batch of requirements from related work \cite{RITTER2019101439} target properties that were only recently formally represented, and thus are not considered in \cite{DBLP:conf/caise/FahlandG13,DBLP:journals/is/BohmHPLW11}.
Let us briefly summarize these pattern-level requirements in the context of this work.
To be operational, some of the patterns like commutative and idempotent receivers as well as aggregator require the ability to persistently store data (\textbf{\refreq{database}: Formal treatment of databases and transactions}), \eg co-enabling OS-2.
Finally, potential exceptions have to be handled and compensated for within a pattern and on a composition level (\textbf{\refreq{exceptions}: Formal treatment of exceptions and compensations}).

To represent and reason about pattern compositions, an adequate formalism must include the ability to structurally and semantically compose patterns (\textbf{\refreq{compositional}: Formalism must be compositional}), co-enabling OS-1--5.
This core requirement is only partially supported in related work: in the work on Petri nets \cite{DBLP:conf/caise/FahlandG13,RITTER2019101439}, composition is assumed through the composable nature of Petri nets, but no formal construction of such compositions is given.
Similarly, \cite{DBLP:journals/is/BohmHPLW11} introduces a composition of data integration operations, but again without a formal construction.

Once patterns are composed, the compositions will be subject to frequent changes such as extensions, adaptation due to changing legal requirements, or improvements and optimizations.
In this work we focus on optimizations representing a comprehensive set of change operations that introduce change to pattern compositions.
For a formal analysis of such changes, the optimizations themselves shall be represented in a formal way, such that compositions and their change operations can be formally analyzed (\textbf{\refreq{improvements}: Formal treatment of improvements (of control and data flow)}), co-enabling OS-1--5.
The notion of change is partially considered in data integration by \cite{DBLP:journals/is/BohmHPLW11}, but not in recent application integration work \cite{DBLP:conf/caise/FahlandG13,RITTER2019101439}.

Finally, a suitable formal representation of pattern compositions shall allow for a structural and semantic analysis of the correctness of a composition.
This requirement also contains the notion of remaining correct after applying changes to the compositions, \eg in the form of optimizations (\textbf{\refreq{correctness}: Formal specification of preserving correctness (structurally and semantically) of compositions}), co-enabling OS-1--5.
In the context of structural composition correctness, we identified the channel cardinality characteristic as decisive correctness criteria based on the control flow.
For example, a content-based router with a cardinality of $1$:$n$ channels, can only be connected to 1 input and $n$ output channels, otherwise the composition is structurally incorrect.
The semantic correctness has to go deeper into the fundamental execution semantics of integration patterns as defined in \cite{RITTER2019101439}, and thus we build on top of that formalism.
However, neither of the current approaches \cite{DBLP:conf/caise/FahlandG13,RITTER2019101439,DBLP:journals/is/BohmHPLW11} addresses the requirement of structural and semantic correctness for compositions, nor do they guarantee a general correctness-preserving property when changing compositions.

%% file: pattern_composition_formalization.tex
\section{Graph-based Pattern Compositions}
\label{sec:formalization}
We now introduce a formalization of pattern compositions, and an abstract cost model for them.
This is needed in order to rigorously reason about optimizations.

\subsection{Integration Pattern Graphs}
\label{sub:ipcgs}
Taking requirement \refreq{controlflow} of having a formal representation of control
flow from \cref{sec:requirements} into account, it is natural to model
pattern compositions as extended control flow
graphs~\cite{allen1970control}, as we do in \cref{def:ipg}. This gives
a high level modelling language that is easy to use and understand, and is close to informal notation used by practitioners \cite{ritter2016exception,Ritter201736}. To
take requirement \refreq{dataflow} of having a formal representation of data flow
into account, we will further enrich the vertices of the graph with
additional information in
\cref{def:characteristics,def:pattern-contract}.
Requirements \refreq{time}, \refreq{database} and \refreq{exceptions} are pattern local \cite{RITTER2019101439}, and thus not relevant at the
abstraction level of pattern compositions. They will become important
when we consider the runtime semantics of compositions later in
\cref{sec:semantics,sec:interpretation}. Control flow graphs can
easily be composed into larger graphs, and hence requirement \refreq{compositional} of
composability is fulfilled. Requirements \refreq{improvements} and \refreq{correctness} will be
addressed in \cref{sec:realization}, of course building on the
definitions in this section.

Compared to for example colored Petri Nets, integration pattern graphs represents pattern compositions at a higher and more specialised level of abstraction, which is more easily understood also for non-technical users.

Before we get to the definition of the kind of graph we need to model
pattern compositions, let us fix some notation: a directed graph is
given by a set of nodes $P$ and a set of edges
$E \subseteq P\times P$.  For a node $p \in P$, we write $\bullet p$ =
\{$p' \in P\ |\ (p',p) \in E$\} for the set of direct predecessors of
$p$, and $p \bullet$ = \{$p'' \in P\ |\ (p,p'') \in E$\} for the set
of direct successors of $p$.

\begin{definition}[Integration pattern type graph]
  \label{def:ipg}
  An \emph{integration pattern typed graph (IPTG)} is a directed graph  with set of nodes $P$ and set of edges $E \subseteq P\times P$, together with a function $\type : P \to T$, where $T = \{$start, end, message processor, fork, structural join, condition, merge, external call$\}$.
  An IPTG $(P, E, \type)$ is \emph{correct} if
  \begin{itemize}
    \itemsep0em
    \item there exists $p_1$, $p_2$ $\in$ P with $\type(p_1)$ = start and $\type(p_2)$ = end;
    \item if $\type(p) = start$ then $|\bullet p| = 0$, and if $\type(p) = end$ then $|p \bullet| = 0$.
    \item if $\type(p) \in$ \{fork, condition\} then $|\bullet p|=1$ and $|p \bullet| > 1$, and if $\type(p) = join$ then $|\bullet p| > 1$ and $|p \bullet| = 1$;
    \item if $\type(p) \in$ \{message processor, merge\} then $|\bullet p| = 1$ and $|p \bullet| = 1$;
    \item if $\type(p) \in$ \{external call\} then $|\bullet p| = 2$ and $|p \bullet| = 2$;
    \item The graph $(P, E)$ is connected and acyclic. \qedwhite
  \end{itemize}
\end{definition}

In the definition, we think of $P$ as a set of extended integration patterns that are connected by message channels represented as edges in $E$, as in a pipes and filter architecture.
The function $\type$ records what type of pattern each node represents.
The first correctness condition says that an integration pattern has at least one source and one target, while the next three states the cardinality of the involved patterns coincide with the in- and out-degrees of the nodes in the graph representing them.
The last condition states that the graph represents one integration pattern, not multiple unrelated ones, and that messages do not loop back to previous patterns.

To represent the data flow, \ie the basis for the optimizations and requirement \refreq{dataflow}, the control flow has to be enhanced with (a) the data that is processed by each pattern, and (b) the data exchanged between the patterns in the composition.
The data processed by each pattern (a) is described as a set of \textsl{pattern characteristics}: 

\begin{definition}[Pattern characteristics]
  \label{def:characteristics}
  A \emph{pattern characteristic} assignment for a
  graph $(P, E)$ is a function
  $\pchar: P \to 2^{\PC}$, assigning to each vertex a subset of the set
  \begin{align*}
    \PC =\
    & (\{\text{MC}\}\times \mathbb{N}^2) \cup {} 
     (\{\text{ACC}\} \times \{\text{ro},\text{rw}\}) \cup {} 
     (\{\text{MG}\} \times \Bool) \cup {} \\
    & (\{\text{CND}\} \times 2^{\BExp}) \cup {} 
     (\{\text{PRG}\} \times \Exp ) \cup {} \\
    & (\{\text{S}\} \times \Exp ) \cup {} (\{\text{QRY}\} \times 2^{\Exp} ) \cup {} (\{\text{ACTN}\} \times 2^{\Exp} ) \cup {} \\
    & (\{\text{TM}\} \times (\mathbb{Q}^{\geq 0} \times (\mathbb{Q}^{\geq 0} \cup \{\infty\})) \enspace ,  
  \end{align*}
  where $\Bool$ is the set of Booleans, $\BExp$ and $\Exp$ the sets of Boolean and 
  program expressions, respectively, and $\text{MC}$, $\text{ACC}$, $\text{MG}$, $\text{CND}$, $\text{PRG}$, $\text{S}$, $\text{QRY}$, $\text{ACTN}$, $\text{TM}$  some distinct symbols. \qedwhite
\end{definition}

The property and value domains in~\cref{def:characteristics} are based on the pattern characteristics identified in \cref{sec:characteristics}, and could of course be be extended if future patterns required it.
We briefly explain the intuition behind the characteristics: the characteristic $(\text{MC}, n, k)$ represents a message cardinality of $n$:$k$, $(\text{ACC}, x)$ the message access,
depending on if $x$ is read-only $\text{ro}$ or read-write $\text{rw}$, and the characteristic $(\text{MG}, y)$ represents whether the pattern is message generating depending on the Boolean $y$.
A $(\text{CND}, \{c1, ..., cn\})$ represents the conditions $c_1$, \ldots, $c_n$ used by the pattern to route messages, and $(\text{PRG}, p)$ represents a program $p$ used by the pattern (\eg for message translation).
The storage aspects are denoted by a schema $(\text{S}, (p_s))$ created by a program $p_s$, expressions $(\text{QRY}, \{q1, ..., qn\})$ denoting a set of distinct queries $q1$, \ldots, $qn$, and a set of actions $(\text{ACTN}, \{a1, ..., an\})$ with distinct $a1$, \ldots, $an$.
Finally, $(\text{TM}, (\tau_s, \tau_e))$ represents a timing window from $\tau_s$ to $\tau_e$.

\begin{example}
The characteristics of a content-based router $CBR$ is $\pchar(CBR) = $\{(MC, 1:1), (ACC, ro), (MG, false), (CND,\{$cnd_1$, \ldots, $cnd_{n-1}$\}), (PRG,\texttt{null}), (S, \texttt{null}), (QRY, $\emptyset$), (ACTN, $\emptyset$), (TM, $(0,0)$)\}, because of the workflow of the router:  it receives exactly one message, then evaluates up to $n-1$ routing conditions $cnd_1$ up to $cnd_{n-1}$ (one for each outgoing channel), until a condition matches.
The original message is then rerouted read-only (in particular, the router is not message generating) on the selected output channel, or forwarded to the default channel, if no condition matches.
The router does not require programs, storage or time configurations. \qedblack
\end{example}

The data exchange between the patterns (b) is based on \textsl{input and output contracts} (similar to data parallelization contracts in \cite{DBLP:conf/cloud/BattreEHKMW10}).
These contracts specify how the data is exchanged in terms of required message properties of a pattern during the data exchange:
\begin{definition}[Pattern contract]
  \label{def:pattern-contract}
  A \emph{pattern contract} assignment for a graph $(P, E)$ is a function $\contr : P \to {\CPT} \times 2^{\EL}$, assigning to each vertex a function of type
  \[
    \CPT = {\{\text{signed}, \text{encrypted}, \text{encoded}\} \to \{\text{yes}, \text{no}, \text{any}\}}
  \]
  and a subset of the set
  \[
    \EL = (\{\text{HDR}\} \times 2^{D}) \cup (\{\text{PL}\} \times 2^{D}) \cup (\{\text{ATTCH}\} \times 2^{D}) \enspace ,
  \]
  where $D$ is a set of  data elements (the concrete elements of $D$ will vary with the application domain). We represent the function of type $\CPT$ by its graph, leaving out the attributes that are sent to $\text{any}$, when convenient. \qedwhite
\end{definition}

The set $\CPT$ in a contract represents integration concepts, while
the set $\EL$ represents data elements and the structure of the
message: its headers $(\text{HDR}, H)$, its payload $(\text{PL}, Y)$
and its attachments $(\text{ATTCH}, A)$.
Each pattern will have an inbound and an outbound pattern contract,
describing the format of the data it is able to receive and send
respectively --- the role of pattern contracts is to make sure that adjacent inbound and outbound contracts match.

\begin{example}
A content-based router is not able to process encrypted messages.
Recall that its pattern characteristics included a collection of routing conditions: these might require read-only access to message elements such as certain headers $h_1$ or payload elements $e_1$, $e_2$.
Hence the input contract for a router mentioning these message elements is
\begin{multline*}
  \iC(CBR)= (\{(\text{encrypted},\text{no})\},\{(\text{HDR},\{h_1\}), \\ (\text{PL},\{e_1, e_2\})\}) \enspace .
\end{multline*}
Since the router forwards the original message, the output contract is the same as the input contract. \qedblack
\end{example}

\begin{definition}
\label{def:matching-contracts}
Let $(C, E) \in {\CPT} \times 2^{\EL}$ be a pattern contract, and $X \subseteq {\CPT} \times 2^{\EL}$ a set of pattern contracts. Write $X_{\CPT} = \{ C'\ | \ (\exists E')\ (C', E') \in X \}$ and $X_{\EL} = \{ E'\ | \ (\exists C')\ (C', E') \in X \}$.
We say that $(C, E)$ \emph{matches} $X$, in symbols $\match((C, E), X)$, if following condition holds:
\begin{align*}
&(\forall x)\big(C(x) \neq \text{any} \implies \\ & \qquad\quad(\forall C' \in X_{\CPT})(C'(x) = C(x) \vee C'(x) = \text{any}) \big) \land {} \\
	& (\forall (m, Z) \in E)\big( Z = \bigcup_{(m, Z') \in \cup X_{\EL}}Z' \big) \enspace .
\end{align*} \qedwhite
\end{definition}

We are interested in an inbound contract $K_\mathrm{in}$ matching the outbound contracts $K_1, \ldots, K_n$ of its predecessors. In words, this is the case if (i) for all integration concepts that are important to $K_\mathrm{in}$, all contracts $K_i$ either agree, or at least one of $K_\mathrm{in}$ or $K_i$ accepts any value (\emph{concept correctness}); and (ii) together, $K_1, \ldots, K_n$ supply all the message elements that $K_\mathrm{in}$ needs (\emph{data element correctness}).

Since pattern contracts can refer to arbitrary message elements, a formalization of an integration pattern can be quite precise.
On the other hand, unless care is taken, the formalization can easily become specific to a particular pattern composition.
In practice, it is often possible to restrict attention to a small number of important message elements (see \cref{ex:IPCG} below), which makes the formalization manageable.

Putting everything together, we formalize pattern compositions as integration pattern typed graphs with pattern characteristics and inbound and outbound pattern contracts for each pattern:
\begin{definition}
  \label{def:pattern_composition_new}
  An \emph{integration pattern contract graph} (IPCG) is a tuple
  \[
    (P, E, \type, \pchar, \iC, \oC)
  \]
  where $(P, E, \type)$ is an IPTG, $\pchar : P \to 2^{\PC}$ is a pattern characteristics assignment, and $\iC : \prod_{p \in P}({\CPT} \times 2^{\EL})^{|\bullet p|}$ and $\oC : \prod_{p \in P}({\CPT} \times 2^{\EL})^{|p \bullet|}$ are pattern contract assignments --- one for each incoming and outgoing edge of the pattern, respectively --- called the inbound and outbound contract assignment respectively.
  It is \emph{correct}, if the underlying IPTG $(P, E, \type)$ is correct, and  inbound contracts matches the outbound contracts of the patterns' predecessors, i.e.\ if for every $p \in P$
\[
  \type(p) = \text{start} \vee \match(\iC(p), \{ \oC(p')\ |\ p' \in \bullet p\}) \enspace .
\]

Two IPCGs are \emph{isomorphic} if there is a bijective function between their patterns that preserves edges, types, characteristics and contracts. \qedwhite
\end{definition}

\begin{example}
\label[example]{ex:IPCG}
  \Cref{fig:cpn_pattern_composition_summarized_v2a,fig:cpn_pattern_composition_summarized_v2b} show IPCGs representing an excerpt of the motivating example from the introduction.
\Cref{fig:cpn_pattern_composition_summarized_v2a} represents the IPCG of the original scenario with a focus on the contracts, and \cref{fig:cpn_pattern_composition_summarized_v2b} denotes an already improved composition showing the characteristics and giving an indication on the pattern latency.
\begin{figure*}[bt]
	\begin{center}$
		\begin{array}{cc}
		\subfigure[IPCG from the motivating example]{\label{fig:cpn_pattern_composition_summarized_v2a}\includegraphics[width=0.5\linewidth]{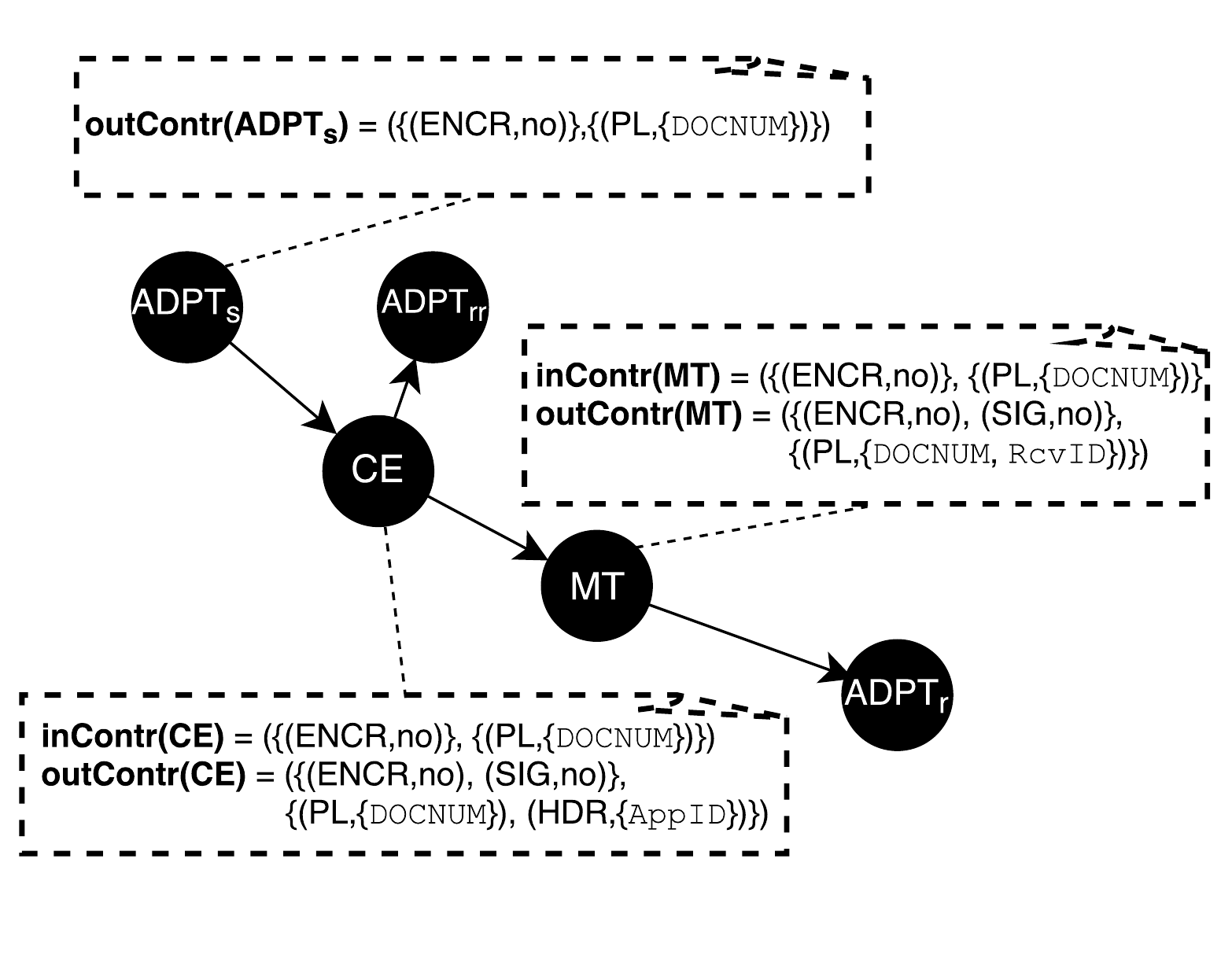}} &
\quad		\subfigure[IPCG after \enquote{sequence to parallel} optimization ]{\label{fig:cpn_pattern_composition_summarized_v2b}\includegraphics[width=0.5\linewidth]{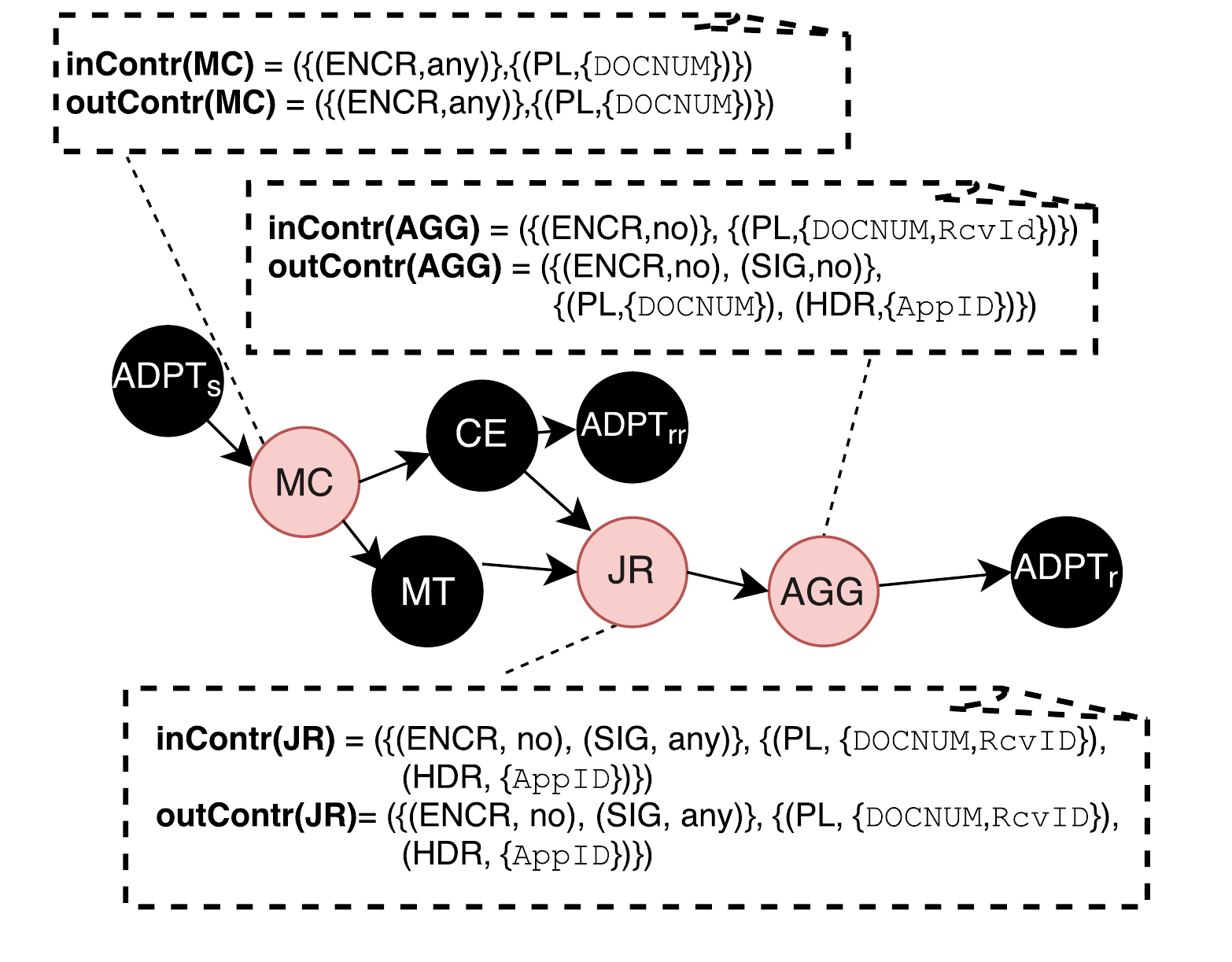}}
		\end{array}$
	\end{center}
	\vspace{-.3cm}
	\caption{An IPCG of an excerpt of the motivating example}
	\label{fig:summarized_example}
\end{figure*}
In \cref{fig:cpn_pattern_composition_summarized_v2a}, the input contract $\iC(CE)$ of the content enricher pattern $CE$ requires a non-encrypted message  and a payload element \texttt{DOCNUM}.
The content enricher makes a query to get an application ID \texttt{AppID} from an external system, and appends it to the message header.
Hence the output contract $\oC(CE)$ contains $(\text{HDR}, \{\texttt{AppID}\})$.
The content enricher then emits a message that is not encrypted or signed.
A subsequent message translator $MT$ requires the same message payload, but does not care about the appended header.
It adds another payload $\texttt{RcvID}$ to the message.
Comparing inbound and outbound pattern contracts for adjacent patterns, we see that this is a correct IPCG.

One improvement of this composition is depicted in \cref{fig:cpn_pattern_composition_summarized_v2b}, where the independent patterns $CE$ and $MT$ have been parallelized.
To achieve this, a read-only structural fork with channel cardinality $1$:$n$ in the form of a multicast $MC$ has been added.
The inbound and outbound contracts of $MC$ are adapted to fit into the composition.
After the concurrent execution of $CE$ and $MT$, a join router $JR$ brings the messages back together again and feeds the result into an aggregator $AGG$ that restores the format that $ADPT_r$ expects.
We see that the resulting IPCG is still correct, so this would be a sound optimization. \qedblack
\end{example}

\input{abstract_cost_model}

%% file: abstract_cost_model.tex
\subsection{Abstract Cost Model}
\label{sub:abstract_costs}

In order to decide if an optimization is an improvement or not, we want to associate abstract costs to integration patterns.
We do this on the pattern level, similar to the work on data integration operators~\cite{bohmsystem}.
The cost of the overall integration pattern can then be computed as the sum of the cost of its constituent patterns.
Costs are considered parametrized by the cardinality of data inputs $|d_{\mathrm{in}_i}|$ ($1 \leq i \leq n$, if the pattern has in-degree $n$), data outputs $|d_{\mathrm{out}_j}|$ ($1 \leq j \leq m$, if the pattern has out-degree $m$), and external resource data sets $|d_r|$.
The costs can also refer to the pattern characteristics.

\begin{definition}[Cost model]
    \label{def:pc:cost_model}
    A \emph{cost assignment} for an IPCG
    $G = (P, E, \type, \pchar,$ $\iC, \oC)$ is an function $\cost(p) : \mathbb{N}^n \times \mathbb{N}^k \times \mathbb{N}^r \to \mathbb{Q}$ for each $p \in P$, where $p$ has in-degree $n$, out-degree $k$ and $r$ external connections.
    The cost $\cost(G) : \mathbb{N}^{N} \times \mathbb{N}^{K} \times \mathbb{N}^R \to \mathbb{Q}$ of the IPCG pattern graph $G$, where $N$ is the sum of the in-degrees of its patterns, $K$ the sum of their out-degrees, and $R$ the sum of their external connections,  is defined to be the sum of the costs of its constituent patterns:
    \begin{multline*}
    \cost(G)(d_{\mathrm{in}}, d_{\mathrm{out}}, d_{r}) = \\ \sum_{p \in P}\cost(p)(|d_{\mathrm{in}}(p)|, |d_{\mathrm{out}}(p)|, |d_{r}(p)|) \enspace ,
    \end{multline*}
    where we suggestively have written $|d_{\mathrm{in}}(p)|$ for the projection from the tuple $d_{\mathrm{in}}$ corresponding to $p$, similarly for $|d_{\mathrm{out}}(p)|$ and $|d_{r}(p)|$. \qedblack
\end{definition}

\begin{table*}[bt]
    \centering
    \small
    \caption{Abstract costs of relevant patterns}
    \label{tab:op:abstract_costs}
    \begin{tabular}{lll}
        \hline
        Pattern $p$ & Abstract Cost $\cost(p)$ & Factors \\
        \hline
        {Content-based Router~\cite{hohpe2004enterprise}} & \parbox[t]{2.0cm}{$\frac{\sum_{i=0}^{n-1}|d_{\mathrm{in},i}|}{2}$} & \parbox[t]{10.0cm}{$n$=\#channel conditions, half of them evaluated on average}\\ 
        {Message Filter~\cite{hohpe2004enterprise}} & \parbox[t]{2.0cm}{$|d_{\mathrm{in}}|$} & \parbox[t]{10.0cm}{input data condition $|d_{\mathrm{in}}|$}\\
        \parbox[t]{1.5cm}{Aggregator~\cite{hohpe2004enterprise}} & {$2\times|d_{\mathrm{in}}|+\frac{|d_{\mathrm{in}}|+|d_r|}{avg(len(seq))}$} & \parbox[t]{10.0cm}{correlation, and completion conditions $|d_{\mathrm{in}}|$, aggregation function $\frac{|d_{\mathrm{in}}|+|d_r|}{avg(len(seq))}$ and length of a sequence length(seq) $>=$ 2, and (transacted) resource $d_r$}\\
        Claim Check~\cite{hohpe2004enterprise} & \parbox[t]{2.0cm}{$2\times|d_r|$} & \parbox[t]{10.0cm}{resource insert and get $|d_r|$}\\
        \parbox[t]{1.5cm}{Splitter~\cite{hohpe2004enterprise}} & \parbox[t]{2.0cm}{$|d_{\mathrm{out}}|$} & \parbox[t]{10.0cm}{output data condition $|d_{\mathrm{out}}|$}\\
        {Multicast, Join Router~\cite{Ritter201736}} & \parbox[t]{2.0cm}{$\sum_{i=0}^{n} \cost(\text{procunit}_i)$} & \parbox[t]{10.0cm}{costs of processing units $\cost(\text{procunit}_i$), \eg threading in software, for $n$ channels}\\
        \hline
        Content Filter~\cite{hohpe2004enterprise} & \parbox[t]{2.0cm}{$|d_{\mathrm{out}}|$} & \parbox[t]{10.0cm}{output data creation $|d_{\mathrm{out}}|$}\\
        \parbox[t]{1.5cm}{Mapping~\cite{hohpe2004enterprise}} & \parbox[t]{2.0cm}{$|d_{\mathrm{in}}| + |d_{\mathrm{out}}|$} & \parbox[t]{10.0cm}{output data creation $|d_{\mathrm{out}}|$ from input data $|d_{\mathrm{in}}|$}\\
        {Content Enricher~\cite{hohpe2004enterprise}} & \parbox[t]{2.0cm}{$|d_{\mathrm{in}}|$+$|d_r|$+$|d_{\mathrm{out}}|$} & \parbox[t]{10.0cm}{request message creation on $|d_{\mathrm{in}}|$, resource query $|d_r|$, response data enrich $|d_{\mathrm{out}}|$} \\
        \hline
        \parbox[t]{2cm}{External Call~\cite{Ritter201736}} & \parbox[t]{2.0cm}{$|d_{\mathrm{out}}| + |d_{\mathrm{in}}|$} & \parbox[t]{10.0cm}{request $|d_{\mathrm{out}}|$ and reply data $|d_{\mathrm{in}}|$}\\
        \parbox[t]{1.5cm}{Receive~\cite{hohpe2004enterprise}} & \parbox[t]{2.0cm}{$|d_{\mathrm{in}}|$} & \parbox[t]{10.0cm}{input data $|d_{\mathrm{in}}|$}\\
        \parbox[t]{1.5cm}{Send~\cite{hohpe2004enterprise}} & \parbox[t]{2.0cm}{$|d_{\mathrm{out}}|$} & \parbox[t]{10.0cm}{output data $|d_{\mathrm{out}}|$}\\
    \end{tabular}
\end{table*}
We have defined the abstract costs of the patterns discussed in this
work in \cref{tab:op:abstract_costs} --- these will be used in the subsequent evaluation.
We now explain the reasoning behind them.
Routing patterns such as content based routers, message filters and aggregators mostly operate on the input message, and thus have an abstract cost related to its element cardinality $|d_{\mathrm{in}}|$.
For example, the abstract cost of the CBR is $\cost(CBR)= \frac{\sum_{i=0}^{n-1}|d_{\mathrm{in},i}|}{2}$, since it evaluates on average $\frac{n-1}{2}$ routing conditions on the input message.
More complex routing patterns such as aggregators evaluate correlation and completion conditions,  as well as an aggregation function on the input message, and  also on sequences of messages of a certain length from an external resource.
Hence the cost of an aggregator is $\cost(AGG)= 2\times|d_{\mathrm{in}}|+\frac{|d_{\mathrm{in}}|+|d_r|}{avg(len(seq))}$, where $len(seq)$ denotes the length of a Message Sequence~\cite{hohpe2004enterprise} as for example used by an \Aggregator{} pattern.
In contrast, message transformation patterns like content filters and enrichers mainly construct an output message, hence their costs are determined by the output cardinality $|d_{\mathrm{out}}|$.
For example, content enrichers create a request message from the input message with cost $|d_{\mathrm{in}}|$, conducts an optional resource query $|d_r|$, and creates and enriches the response with cost $|d_{\mathrm{out}}|$.
Finally, the cost of message creation patterns such as  external calls, receivers, and senders arise from costs for transport, protocol handling, and format conversion, as well as decompression. Hence the cost depends on the element cardinalities of input and output messages $|d_{\mathrm{in}}|$, $|d_{\mathrm{out}}|$.

\begin{example}
    We return to the claimed improved composition in \cref{ex:IPCG}.
    The latency of the composition $G_1$ in \cref{fig:cpn_pattern_composition_summarized_v2a}, calculated from the constituent pattern latencies, is $\cost(G_1) = t_{CE}+t_{MT}$ with latency $t_{p}$ and pattern $p$.
    The latency improvement potential given by switching to the composition $G_2$ in \cref{fig:cpn_pattern_composition_summarized_v2b} is given by $\cost(G_2) = \max(t_{CE},t_{MT}) + t_{MC} + t_{JR} + t_{AGG}$.
    Obviously it is only beneficial to switch if $\cost(G_2) < \cost(G_1)$, and this condition depends on the concrete values involved.
    At the same time, the model complexity increases by three nodes and edges. \qedblack
\end{example}

%% file: petri_net_semantics.tex
\section{A Semantics Using Timed DB-nets}
\label{sec:semantics}

Integration pattern graphs model the structural composition of
integration patterns, but not their dynamics, i.e.\ how data actually
flow through the system. To model this, we use timed
db-nets~\cite{RITTER2019101439}, an extension of
db-nets~\cite{DBLP:journals/topnoc/MontaliR17} with an explicit notion
of time (addressing \refreq{time}). The formalism of db-nets in turn
is a refinement of colored Petri nets~\cite{jensen2013coloured} with
primitives for the net to query and update persistent data stores
(addressing \refreq{database}). Exceptions are built into the
framework in the form of rollbacks (addressing
\refreq{exceptions}).

We choose to work with timed db-nets rather than just colored Petri nets because they meet the mentioned requirements of integration processes that we have identified, and balances the dimensions of persistence, data logic and control layer of a Petri net.
Avoiding the heavier encoding of colored Petri nets, timed db-nets make the modeling more concise and tractable for the interpretation procedure defined in \cref{sec:interpretation}.

To make the definition compositional, we
have to extend the notion of timed db-nets to timed db-nets with
boundaries that can be reasoned about separately, and then plugged
together to form larger timed db-nets.

\subsection{Open Timed DB-nets}

In this section, we formally define the mathematical structure we use
to give a runtime semantics to pattern graphs. We first recall the
definition of timed db-nets, and then extend
them to open timed db-nets, in order to make the definition
compositional.

\subsubsection{Ordinary Timed DB-nets}

A timed db-net has three layers: a persistence layer describing the
underlying database of the net, a logic layer describing the queries
that can be made of the persistence layer, and a control layer
describing how tokens of data flow through the net, executing queries.
See Ritter et al.~\cite{RITTER2019101439} for motivation and a more gentle definition.

\begin{definition}[timed db-net~\cite{RITTER2019101439}]
  \label{def:eip-net}
  A \emph{timed db-net} is a tuple $(\mathfrak{D},\mathcal{P},\allowbreak\mathcal{L},\mathcal{N}, \tau)$, where:
  \begin{compactitem}[$\bullet$]
    \item $\mathfrak{D}$ is a type domain --- a finite set of data types, each of the form $D = (\Delta_D, \Gamma_D)$, where $\Delta_D$ is a value domain, and $\Gamma_D$ is a finite set of domain-specific predicate symbols.
    \item $\mathcal{P}$ is a $\mathfrak{D}$-typed \textbf{persistence layer}, i.e., a pair $(\mathcal{R}, E)$, where $\mathcal{R}$ is a $\mathfrak{D}$-typed database schema, and $E$ is a finite set of first-order FO($\mathfrak{D}$) constraints over $\mathcal{R}$.
    \item $\mathcal{L}$ is a $\mathfrak{D}$-typed \textbf{data logic layer} over $\mathcal{P}$, i.e., a pair $(Q, A)$, where $Q$ is a finite set of FO($\mathfrak{D}$) queries over $\mathcal{P}$, and $A$ is a finite set of actions over $\mathcal{P}$.
    \item $\mathcal{N}$ is a $\mathfrak{D}$-typed \textbf{control layer} over $\mathcal{L}$, i.e., a tuple $(P, T,F_{in},$ $F_{out}, F_{rb}, \allowbreak \texttt{color}, \allowbreak\texttt{query}, \texttt{guard}, \texttt{action})$, where:
      \begin{enumerate} 
      \item $P=P_c \uplus P_v$ is a finite set of places, partitioned into so-called control places $P_c$ and view places $P_v$,
      \item $T$ is a finite set of transitions,
      \item $F_{in}$ is an input flow from $P$ to $T$,
      \item $F_{out}$ and $F_{rb}$ are respectively output and roll-back flows from $T$ to $P_c$,
      \item \texttt{color} is a color assignment over $P$ (mapping $P$ to a Cartesian product of data types),
      \item \texttt{query} is a query assignment from $P_v$ to $Q$ (mapping the results of $Q$ as tokens of $P_v$),
      \item \texttt{guard} is a transition guard assignment over $T$ (mapping each transition to a formula over its input inscriptions), and
      \item \texttt{action} is an action assignment from $T$ to $A$ (mapping some transitions to actions triggering updates over the persistence layer).
      \end{enumerate}
    \item $\tau: T \to \mathbb{Q}^{\geq 0} \times (\mathbb{Q}^{\geq 0} \cup \{\infty\})$ is a timed transition guard, mapping each transition $t \in T$ to a pair of values $\tau(t) = (v_1, v_2)$, where $v_1$ is a non-negative rational number, and $v_2$ is a non-negative rational number with $v_1 \leq v_2$, or the special constant $\infty$. \qedwhite
    \end{compactitem}
\end{definition}

We adopt the following graphical conventions for drawing the control
layer of a timed db-net: places are depicted as round nodes --- view
places are labelled by a database icon
\tikz{\node[cylinder,aspect=0.5,draw,rotate=90] {};} with queries written in \textcolor{Green}{green} --- and transitions as rectangles. Rollback arcs are depicted with an ``x'':
\tikz{\node[rectangle,draw] (a) at (0,0) {}; \node[circle,draw] (b) at
  (1,0) {}; \draw[Rays-latex,red] (a) -- (b); }. Actions are written in blue, and guards are written in
square brackets next to the transition, and we adopt the following
conventions for a timed transition guard $\tau$ and a transition $t$:
\begin{inparaenum}[\it (i)]
	\item if $\tau(t) = (0, \infty)$, then no temporal label is shown for $t$ (this is often the default choice for $\tau(t)$);
	\item if $\tau(t)$ is of the form $(v, v)$, we attach label ``@$v$'' to $t$;
	\item if $\tau(t)$ is of the form $(v_1, v_2)$ with $v_1 \neq v_2$, we attach label ``$@\langle v_1,v_2 \rangle$'' to $t$.
\end{inparaenum}

\begin{example}
  \label{ex:tdbnet-aggregator}
  \begin{figure*}
    \centering
    \includegraphics[width=\textwidth]{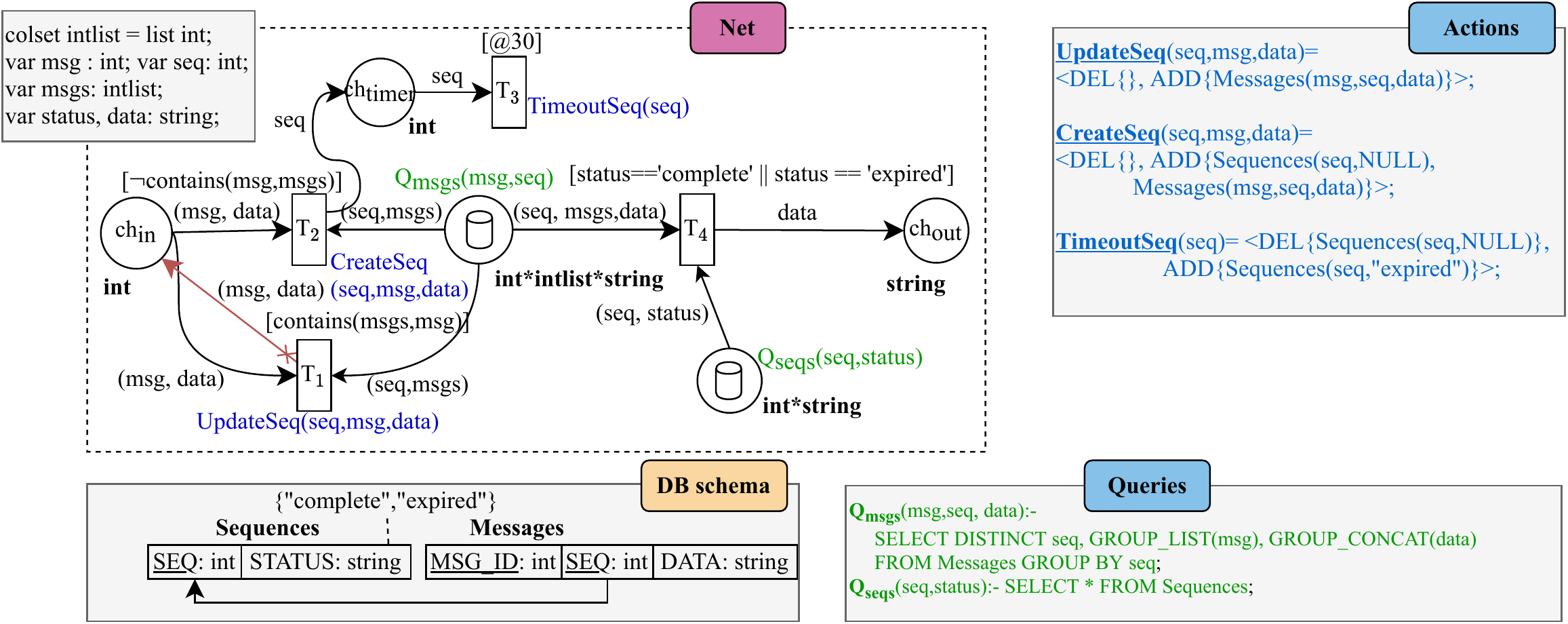}
    \caption{Timed db-net realization of an aggregator}
    \label{fig:example-tdbnet-aggregator}
  \end{figure*}
  \cref{fig:example-tdbnet-aggregator} shows a timed db-net
  realisation of an aggregator. The intention is that messages arrive
  at the place $ch_{in}$. The database is then queried using the
  $Q_{msgs}$ query via a view place, and if it already contains the
  message, it is updated via the UpdateSeq action at transition
  $T_1$. If it does not contain the message, the CreateSeq action is
  triggered at $T_2$ instead, and the sequence number gets passed to
  the $ch_{timer}$ place, whose output transition $T_3$ will be
  enabled after 30 seconds, triggering the TimeoutSeq action. This
  will enable transition $T_4$, with the effect that a token
  containing the data from the completed sequence, queried via
  $Q_{seqs}$ from the database, will move into $ch_{out}$.
\end{example}

\subsubsection{Open Timed DB-nets}
We now describe timed db-nets that are open, in the sense that they have \enquote{ports} for communicating with the outside world: the idea being that tokens can be received and sent on these ports, similar to in the existing literature on open Petri nets~\cite{openPetri2005,sobocinski2010representations,openPetri2020master}.

\begin{definition}[Open timed db-net]
  An \emph{open timed db-net} is a pair $A = (N_A, B_A)$, where $N_A = (\mathfrak{D},\mathcal{P},\mathcal{L},\mathcal{N}, \tau)$ is a timed db-net with control layer
\[
  \mathcal{N} = (P_c \uplus P_v, T, F_{in}, F_{out}, F_{rb}, \texttt{color}, \texttt{query}, \texttt{guard}, \texttt{action})
\]
and $B_A = (I_A, O_A) \in \List{P_c} \times \List{P_c}$ are lists of control places, called the input and output boundaries respectively, such that $F_{in}(o, t) = \emptyset$ for every $o \in O_A$, and $F_{out}(t, i) = F_{rb}(t, i) = \emptyset$ for every $i \in I_A$. The \emph{input (output) boundary configuration of $A$} is given by the corresponding list of colours of the input (output) boundary places of $A$, and we write
\[
  N_A : \texttt{color}(I_A) \to \texttt{color}(O_A)
\]
(where $\texttt{color}(X) = [\texttt{color}(x)\;|\;{x \in X}]$)
to indicate that $A = (N_A, (I_A, O_A))$ is an open timed db-net with the given boundary configurations.
\end{definition}

Note in particular that an open timed db-net with empty boundaries is by definition an ordinary timed db-net.

\begin{example}
  \begin{figure}
    \centering
    \includegraphics[width=\columnwidth]{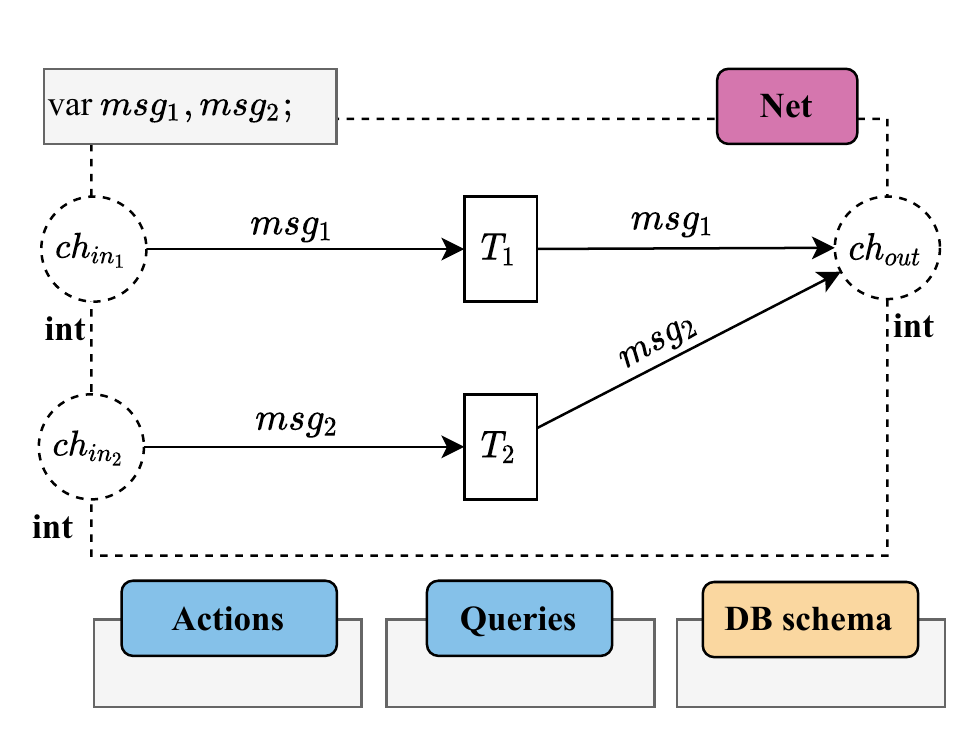}
    \caption{Timed db-net realization of a join router}
    \label{fig:example-tdbnet-join}
  \end{figure}
  \cref{fig:example-tdbnet-join} shows an open timed db-net
  realisation of a join router (joining messages containing integer
  data). The input boundary are the places $ch_{in_1}$ and $ch_{in_2}$
  and the output boundary the place $ch_{out}$. We draw the input
  boundary using dashed places on the left of the image, and the
  output similarly on the right (in general, a place can be part of
  both the input and the output boundary, but this will not occur in
  any nets constructed from pattern graphs).

  Similarly, the timed db-net realisation of an aggregator in
  \cref{fig:example-tdbnet-aggregator} from
  \cref{ex:tdbnet-aggregator} can be made into an open timed db-net by
  declaring the boundaries to be $ch_{in}$ and $ch_{out}$
  respectively.
\end{example}

\subsection{Execution Semantics for Open Timed DB-nets}
\label{sec:execution-semantics}
We define the execution semantics of a given open timed db-net as a
labelled transition system, where the states are snapshots of the
db-net and the labelled transitions are given by firings, as well as
transitions that can create and consume tokens at the input and output
boundary respectively.
A snapshot of an open timed db-net $\mathcal{B}$ is a
snapshot $(\mathcal{I}, m)$ of $\mathcal{B}$ considered as an ordinary
timed db-net (i.e.\ we forget about the boundaries), which in turn
consists of a compliant database instance $\mathcal{I}$ and a marking
$m$; see \cite{RITTER2019101439} for the precise definitions.

Given an open timed db-net $\mathcal{B}$ with boundaries $(I_\mathcal{B}, O_\mathcal{B})$, and a
$\mathcal{B}$-snapshot $s_0$ (the \emph{initial
  $\mathcal{B}$-snapshot}), we construct a labelled transition system
$\Gamma^{\mathcal{B}}_{s_0} = (S, s_0, \rightarrow)$ as follows:
$S$ is the infinite set of $\mathcal{B}$-snapshots, and
$\rightarrow \subseteq S \times (I_\mathcal{B} \cup T \cup O_B) \times S$ is defined by the following clauses (we write $s \overset{\ell}{\rightarrow} s'$ for $(s, \ell, s')  \in \rightarrow$):
\begin{itemize}
\item for a transition $t$ and $\mathcal{B}$-snapshots $s_1$, $s_2$, if there is a binding $\sigma$ such that $t$ fires in $s_1$ with binding $\sigma$ producing $s_2$ (see \cite{RITTER2019101439}), then $s_1 \overset{t}{\rightarrow} s_2$;
\item for an input boundary place $p_i$ and $\mathcal{B}$-snapshot $s = (\mathcal{I}, m)$, if $s' = (\mathcal{I}, m')$ where $m'(p_i) = m(p_i) + \{o\}$ for some $o \in \Delta_{\mathfrak{D}}$, and $m'(x) = m(x)$ for $x \neq p_i$, then $s \overset{p_i}{\rightarrow} s'$; and
\item for an output boundary place $p_o$ and $\mathcal{B}$-snapshot $s = (\mathcal{I}, m)$, if $s' = (\mathcal{I}, m')$ where $m'(p_i) = m(p_i) - \{o\}$ for some $o \in \Delta_{\mathfrak{D}}$, and $m'(x) = m(x)$ for $x \neq p_o$, then $s \overset{p_o}{\rightarrow} s'$.
\end{itemize}

That is, in addition to transitions in the LTS arising from
transitions in the db-net firing, we also have transitions labelled by
each boundary place that make one token appear or disappear at the
boundary, depending on if the place is an input or output. Note that for a closed timed db-net, the boundaries are
empty, and the LTS correspond exactly to the LTS of timed db-nets from \cite{RITTER2019101439}.

\subsection{Composition of Open Timed DB-nets}
\label{sub:composition_open_timed_dbnets}
It is straightforward to compose timed db-nets in parallel, i.e.\ in such a way that there is no interaction between the component nets. Given open timed db-nets
\begin{align*}
  \mathcal{A} &: [c_1, \ldots, c_n] \to [d_1, \ldots, d_m] \\
  \mathcal{B} &: [c'_1, \ldots, c'_{n'}] \to [d'_1, \ldots, d'_{m'}]
\end{align*}
with the same type domains, persistence layers and data logic layers\footnote{To compose timed db-nets with different underlying layers, we first rename any unintended clashing names, and then take the union of the layers and embed the nets into their now common layers.}, we define an open timed db-net
\[
  \mathcal{A} \otimes \mathcal{B}  : [c_1, \ldots, c_n, c'_1, \ldots, c'_{n'}] \to [d_1, \ldots, d_m, d'_1, \ldots, d'_{m'}]
\]
again with the same type domain, persistence layer data logic layer, but
whose places and transitions are the disjoint
union of the places and transitions in $\mathcal{A}$ and $\mathcal{B}$ respectively.
This gives a tensor product or parallel composition of nets, with unit $I : [] \to []$ the empty timed db-net (necessarily with empty boundary). Visually, we are stacking the control layers of $\mathcal{A}$ and $\mathcal{B}$ next to each other.

When the boundaries are compatible, i.e., when the input boundary
configuration is the same as the output boundary configuration, we can
also define a sequential composition of nets. This will be achieved by
\enquote{gluing} the two nets together along their common boundary, formally
expressed by quotienting the set of places in the construction of the composite net.

\begin{definition}[Sequential composition of open nets]
  Let $\mathcal{A} : \texttt{color}(I_A) \to \texttt{color}(O_A)$ and $\mathcal{B} : \texttt{color}(I_B) \to \texttt{color}(O_B)$ be open timed db-nets with the same type domains, persistence layers and data logic layers, and such that $\texttt{color}(O_A) = \texttt{color}(I_B)$.
Write $O_A = [o_1, \ldots, o_n]$ and $I_B = [i_1, \ldots, i_n]$ --- note that $O_A$ and $I_B$ must have the same length, since $\texttt{color}(O_A) = \texttt{color}(I_A)$.
  We define the composition
  \[
    \mathcal{A} \fatsemi \mathcal{B}  : \texttt{color}(I_A) \to \texttt{color}(O_B)
  \]
  to again have the same type domain, persistence layer and data logic layer as $\mathcal{A}$ and $\mathcal{B}$, and control layer
  \[
    \mathcal{N}_{\mathcal{A} \fatsemi \mathcal{B}} = (P, T, F_{in}, F_{out}, F_{rb}, \mathtt{color}, \mathtt{query}, \texttt{guard}, \mathtt{action})
  \]
  with
  \begin{align*}
    P &= (P_A \uplus P_B)/\sim \\
    \intertext{where $\sim$ is the equivalence relation generated by $\inn_{P_A}(o_k) \sim \inn_{PB}(i_k)$ for $0 < k \leq n$,}
    T &= T_A \uplus T_B \\
    F_{in}(x, y) &=
       \begin{cases}
         F^A_{in}(p, t) & \text{if $(x, y) = ([\inn_{P_A}(p)], \inn_{T_A}(t))$} \\
         F^B_{in}(p', t') & \text{if $(x, y) = ([\inn_{P_B}(p')], \inn_{T_B}(t'))$} \\
         \emptyset & \text{otherwise}
       \end{cases} \\
    F_{in}(x, y) &=
       \begin{cases}
         F^A_{out}(p, t) & \text{if $(x, y) = (\inn_{T_A}(t), [\inn_{P_A}(p)])$} \\
         F^B_{out}(p', t') & \text{if $(x, y) = (\inn_{T_B}(t'), [\inn_{P_B}(p')])$} \\
         \emptyset & \text{otherwise}
       \end{cases}  \\
    F_{rb}(x, y) &=
       \begin{cases}
         F^A_{rb}(t, p) & \text{if $(x, y) = (\inn_{T_A}(t), [\inn_{P_A}(p)])$} \\
         F^B_{rb}(t', p') & \text{if $(x, y) = (\inn_{T_B}(t'), [\inn_{P_B}(p')])$} \\
         \emptyset & \text{otherwise}
       \end{cases}
  \end{align*}
  \begin{align*}
    \mathtt{color} &= [\mathtt{color}_A, \mathtt{color}_B] \\
    \mathtt{query} &= [\mathtt{query}_A, \mathtt{query}_B] \\ 
    \mathtt{guard} &= [\mathtt{guard}_A, \mathtt{guard}_B] \\
    \mathtt{action}'' &= [\mathtt{action}, \mathtt{action}'] \\
    \tau'' &= [\tau, \tau']
  \end{align*}
\end{definition}

Note that the new colour assignment is well-defined on the quotient $(P_A \uplus P_B)/\sim$ since the two constituent nets have compatible boundaries, by assumption.

\begin{example}
  \begin{figure*}
    \centering
    \includegraphics[width=\textwidth]{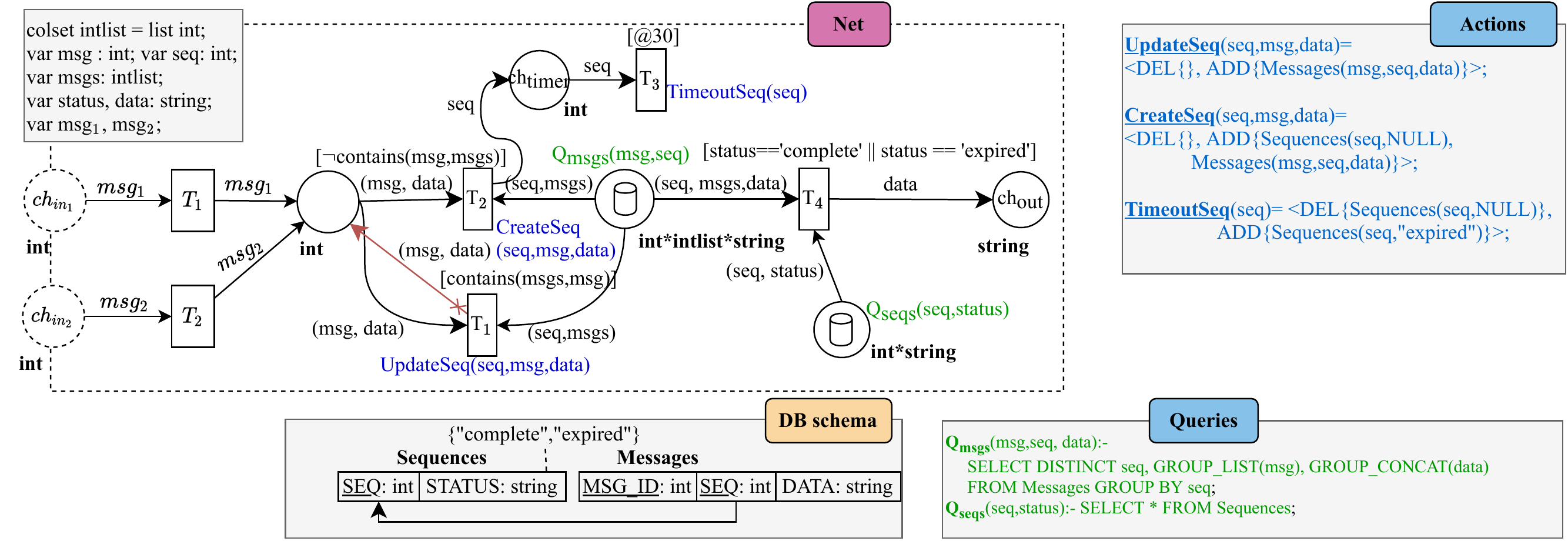}
    \caption{Sequential composition of timed db-net realizations of a join router and an aggregator}
    \label{fig:example-tdbnet-composed}
  \end{figure*}
In \cref{fig:example-tdbnet-composed}, we see the sequential composition \newline $\text{JoinRouter}~\fatsemi~\text{Aggregator}$ of the timed db-net realisation of a join router followed by an aggregator. The two nets have been glued together at their common boundary.
\end{example}

Composition of nets behaves as expected: it is associative, and there is an identity net which is a unit for composition.
Furthermore, the parallel and sequential compositions interact well, in the sense that we get the same result no matter if we first compose in parallel and then sequentially, or the other way around.
All in all, this means that nets with boundaries are the morphisms of a strict monoidal category~\cite{selinger2009graphical}:

\begin{lemma}
  \label{thm:nets-moncat}
  For any open timed db-nets $N$, $M$, $K$ with compatible boundaries, we have $N \fatsemi (M \fatsemi K) = (N \fatsemi M) \fatsemi K$, and for each boundary configuration $\vec{c} = [c_1, \ldots, c_n]$, there is an identity net $\mathsf{id}_{\vec{c}} : [c_1, \ldots, c_n] \to [c_1, \ldots, c_n]$ such that $\mathsf{id}_{\vec{c}} \fatsemi N = N$ and $\mathsf{id}_{\vec{c}} \fatsemi M = M$ for every $M$, $N$ with compatible boundaries.
  Furthermore, for every $N$, $M$, $K$, we have $N \otimes (M \otimes K) = (N \otimes M) \otimes K$, and for compatible nets $N \fatsemi (M \otimes K) = (N \fatsemi M) \otimes (N \fatsemi K)$.
\end{lemma}
\begin{proof}
  Associativity for both $\fatsemi$ and $\otimes$ is straightforward.
  The identity net for $[c_1, \ldots, c_n]$ is the net with exactly $n$ places $x_1, \ldots, x_n$ with $\texttt{color}(x_i) = c_i$, that are all both input and output boundaries. \qed
\end{proof}

In particular, the lemma implies that we can compose nets
sequentially and in parallel without worrying about how to bracket the compositions~\cite{selinger2009graphical}.

\subsection{CPN Tools Prototype}
\label{sub:testing}
We prototypically implemented our formalism so as to experiment with pattern compositions via simulation, following the idea described in \cite[Sect. 5]{RITTER2019101439}.
We have chosen CPN Tools v4.0.1\footnote{CPN Tools, visited 5/23: \url{https://cpntools.org/}} for modeling and simulation.
As compared to other PN tools like Renew v2.5\footnote{Renew, visited 5/2023: \url{http://www.renew.de/}}, CPN tools supports third-party extensions that can address the persistence and data logic layers of our formalism.
Moreover, CPN Tools handles sophisticated simulation tasks over models that use the deployed extensions.
To support db-nets, our extension\footnote{CPN Tools extension for timed db-net with boundaries and pattern models (\ie mainly $*boundary*.cpn$, $*fusion*.cpn$) is available for download, visited 5/2023: \url{https://github.com/dritter-hd/db-net-eip-patterns}.}---denoting an unpublished part of the first author's PhD thesis~\cite{ritter2019phd}---adds support for defining view places together with corresponding SQL queries as well as actions, and realizes the full execution semantics of db-nets using Java and a PostgreSQL database.

\section{Interpreting IPCGs as Open Timed DB-nets}
\label{sec:interpretation}
In this section we define the interpretation of integration pattern contract graphs as timed db-nets with boundaries.

\subsection{Interpretation of Single Patterns}
\label{sec:interpretation-atomic}

We assign an open timed db-net $\sem{p}$ for every node $p$ in a integration pattern contract graph.
Recall that an integration pattern contract graph has input and output contracts $\iC : \prod_{p \in P}({\CPT} \times 2^{\EL})^{|\bullet p|}$ and $\oC : \prod_{p \in P}({\CPT} \times 2^{\EL})^{|p \bullet|}$ respectively.
If the cardinality of $p$ is $k:m$, then the open timed db-net will be of the form
\[
  \sem{p} : \bigotimes_{i = 1}^k \iC_i(p)_{\EL}  \to \bigotimes_{j = 1}^m \oC_j(p)_{\EL}
\]
This incorporates the data elements of the input and output contracts into the boundary of the timed db-net, since these are essential for the dataflow of the net.
In \cref{sec:CPT-net-construction}, we will also incorporate the remaining concepts from the contracts such as signatures, encryption and encodings into the interpretation.

The shape of the timed db-net $\sem{p}$ depends on $\type(p)$ only, i.e., we give one interpretation for each pattern type:

\paragraph{Start and end pattern types}
We interpret a start pattern $p_{\text{start}}$ as the open timed db-net $\sem{p_{\text{start}}} : I \to \cout{p_{\text{start}}}{}$ shown in~\cref{fig:start}.
\begin{figure}[bt]
    \begin{center}$
        \begin{array}{cc}
        \subfigure[Start]{\label{fig:start}\includegraphics[width=0.5\columnwidth]{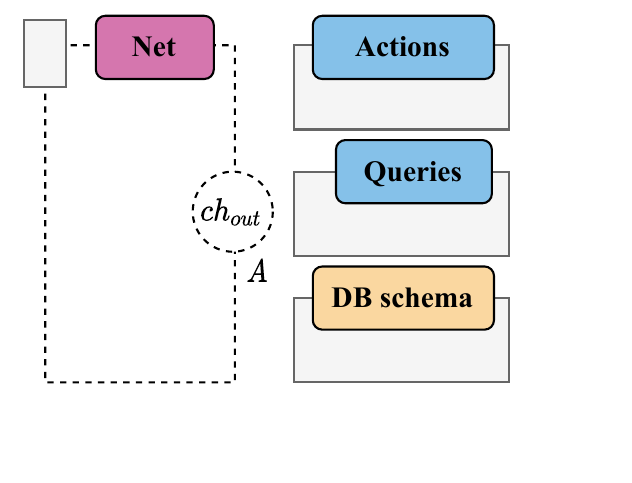}} &
        \subfigure[End]{\label{fig:end}\includegraphics[width=0.5\columnwidth]{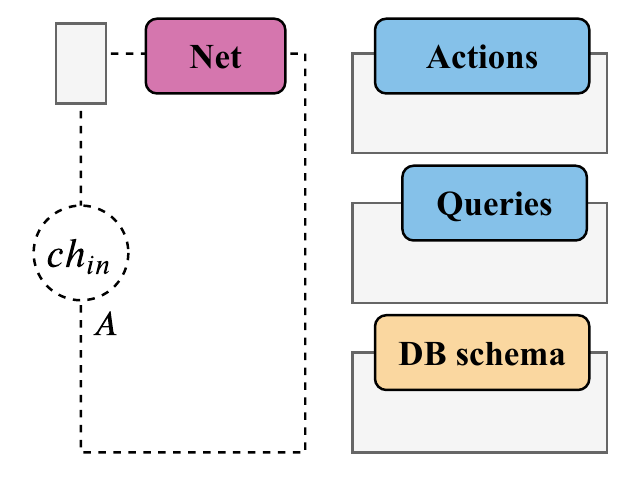}} 
        \end{array}$
    \end{center}
    \caption{Start and end patterns}
    \label{fig:}
\end{figure}
Similarily,~\cref{fig:end} shows the interpretation of an end pattern $p_{\text{end}}$ as an open timed db-net $\sem{p_{\text{end}}} : \cin{p_{\text{end}}}{} \to I$.

\paragraph{Non-conditional fork patterns}
We interpret a non-conditional fork pattern $p_{\text{fork}}$ with cardinality $1:n$ as the open timed db-net $\sem{p_{\text{fork}}} : \cin{p_{\text{fork}}} \to \bigotimes_{j = 1}^n \cout{p_{\text{fork}}}{j}$ shown in~\cref{fig:multicast}. 

\begin{figure}[bt]
	\centering
	\includegraphics[width=0.9\columnwidth]{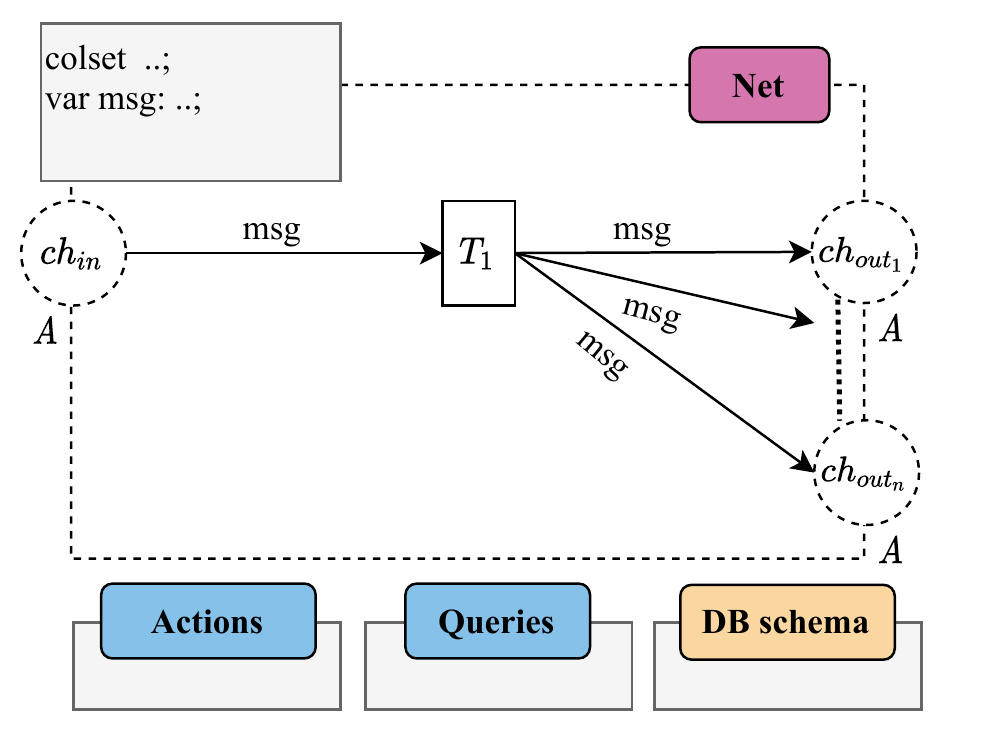}
	\caption{Interpretation of an unconditional fork pattern.}
	\label{fig:multicast}
\end{figure}

\paragraph{Non-conditional join patterns}
We interpret a non-conditional join pattern $p_{\text{join}}$ with cardinality $m:1$ as the open timed db-net $\sem{p_{\text{join}}} : \bigotimes_{j = 1}^m \cin{p_{\text{join}}}{j} \to \cout{p_{\text{join}}}{}$ shown in~\cref{fig:join}.

\begin{figure}[bt]
	\centering
	\includegraphics[width=0.9\columnwidth]{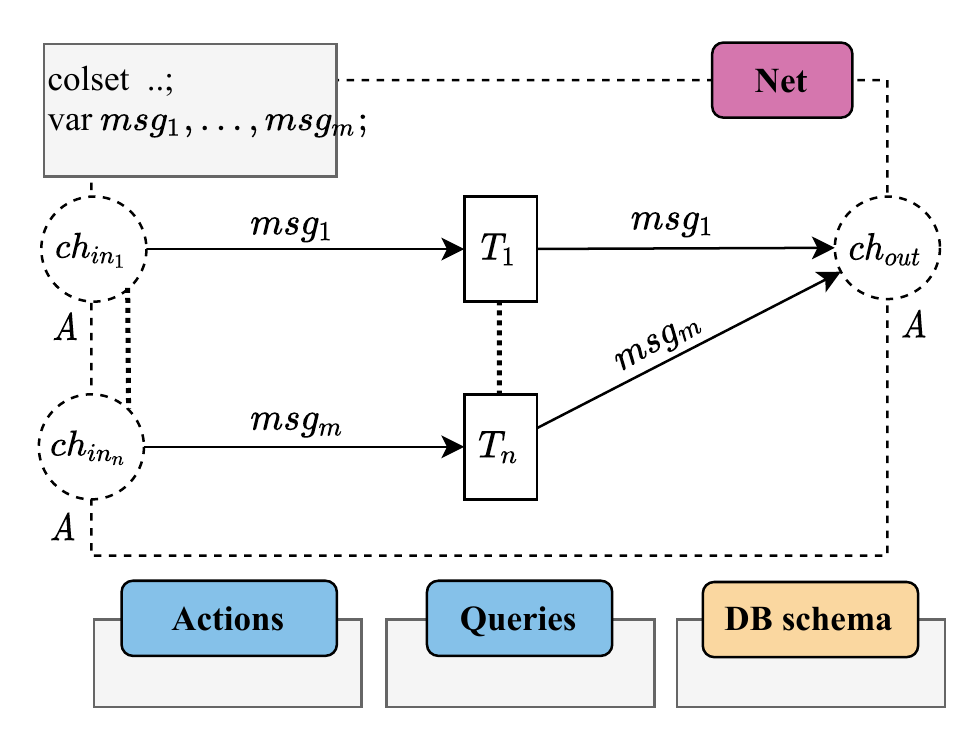}
	\caption{Interpretation of an unconditional join pattern}
	\label{fig:join}
\end{figure}

\paragraph{Conditional fork patterns}
We interpret a conditional fork pattern $p_{\text{cfork}}$ of cardinality $1:n$ with conditions $cond_1$, \ldots, $cond_{n-1}$ in its pattern characteristic assignment as the open timed db-net $\sem{p_{\text{cfork}}}$ $:$ $\cin{p_{\text{cfork}}}{}$ $\to$ $\bigotimes_{j = 1}^n \cout{p_{\text{cfork}}}{j}$ shown in~\cref{fig:content_based_router}.
Note that net is constructed so that the conditions are evaluated in order --- the transition corresponding to condition $k$ will only fire if condition $k$ is true, and conditions 1, \ldots, $k-1$ are false.
The last transition will fire if all conditions evaluate to false.

\begin{figure}[bt]
	\centering
	\includegraphics[width=1\columnwidth]{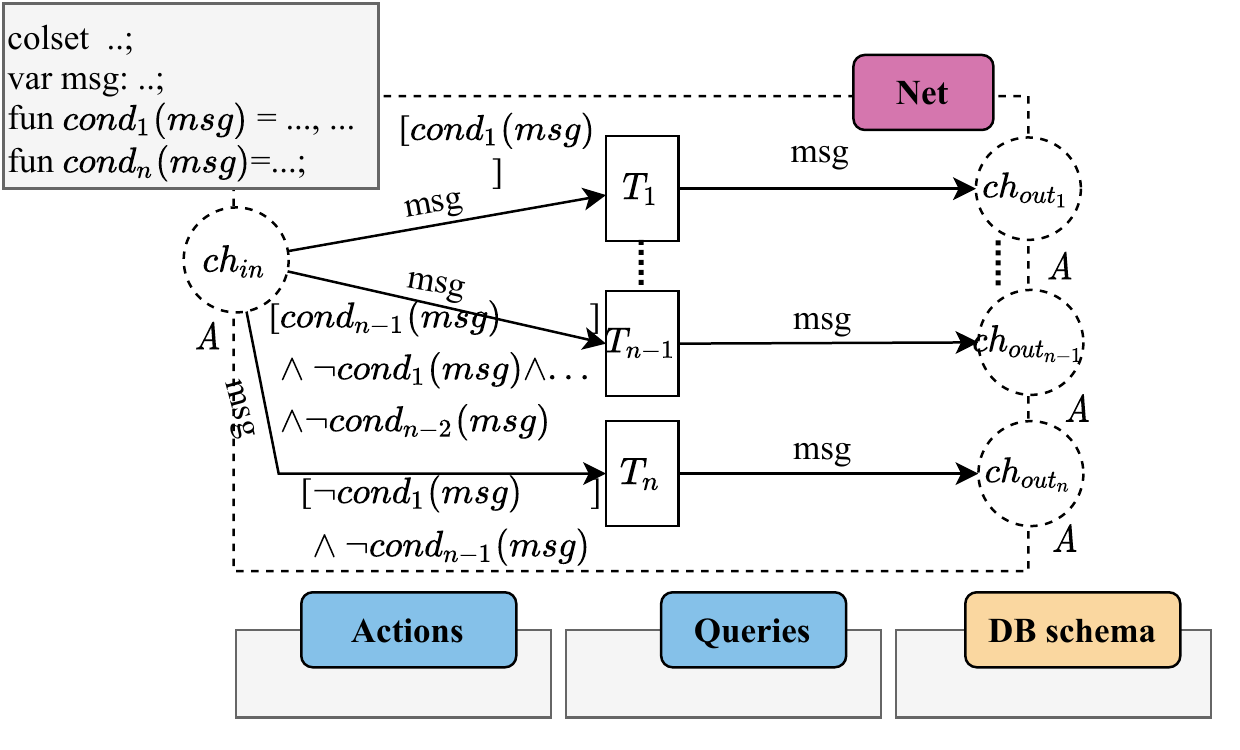}
        \caption{Interpretation of a conditional fork pattern}
	\label{fig:content_based_router}
\end{figure}

\paragraph{Message processor patterns}
We interpret a message processor pattern $p_{\text{mp}}$  with storage schema $S$, actions $A$, query $Q$, condition $cond$, time $\tau$ and program $f$ in its pattern characteristic assignment as the open timed db-net $\sem{p_{\text{mp}}} : \cin{p_{\text{mp}}}{} \to \cout{p_{\text{mp}}}{}$ shown in~\cref{fig:siso_w_storage}.
If the condition $cond$ and timing window $\tau$ are satisfied, the incoming message possibly gets enriched by data from the query $Q$ and action $A$ might be triggered, before the program $f$ transforms the data into possibly multiple messages, collected in a list. These get emitted one by one.

Of course, not all features need to be used by all message processor patterns (\eg no storage for control-time delayer).
\begin{figure*}[bt]
  \centering
  \includegraphics[width=0.7\linewidth]{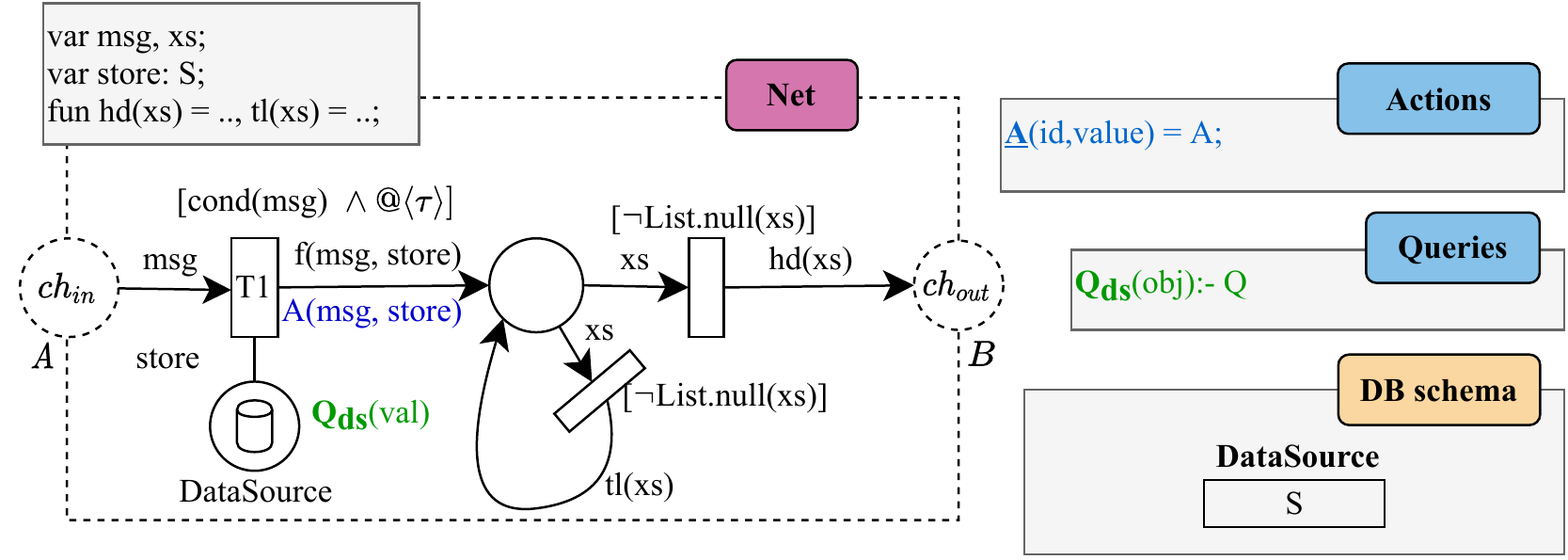}
  \caption{Interpretation of a message processor pattern.}
  \label{fig:siso_w_storage}
\end{figure*}

\paragraph{Merge patterns} 
We interpret a merge pattern $p_{\text{merge}}$  with aggregation function $f$ and timeout $\tau$ as the open timed db-net $\sem{p_{\text{merge}}} : \cin{p_{\text{merge}}}{} \to \cout{p_{\text{merge}}}{}$ shown in~\cref{fig:aggregator}, where \texttt{contains(msg, msgs)} is defined to be the function that checks if \texttt{msg} occurs in the list \texttt{msgs}.
Briefly, the net works as follows: the first message in a sequence makes transition $T1$ fire, which creates a new database record for the sequence, and starts a timer.
Each subsequent message from the same sequence gets stored in the database using transition $T2$, until $\tau$ seconds has elapsed, which will fire transition $T3$.
The action associated to $T3$ will make the condition for the Aggregate transition true, which will retrieve all messages $msgs$ and then put $f(msgs)$ in the output place of the net.

\begin{figure*}[bt]
	\centering
	\includegraphics[width=.8\linewidth]{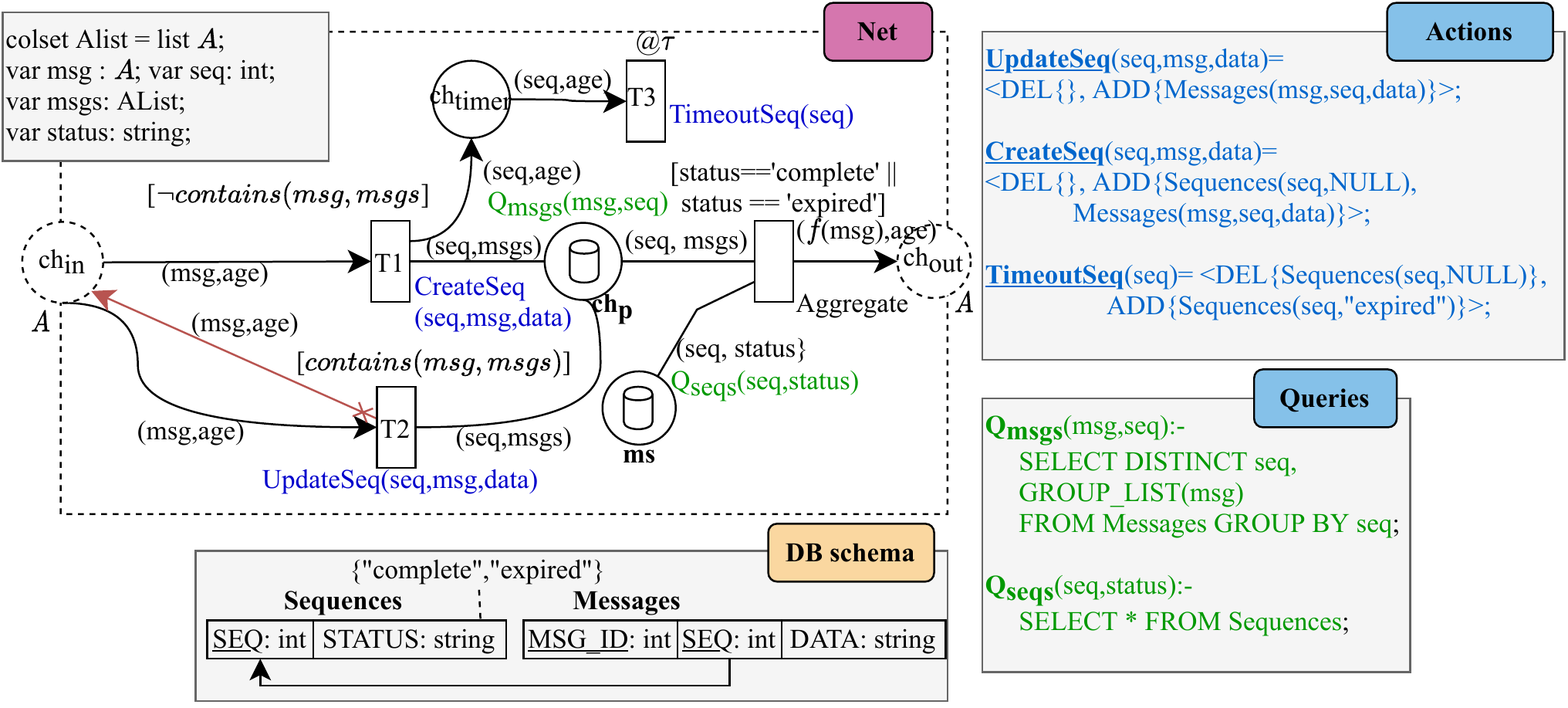}
	\caption{Interpretation of a merge pattern}
	\label{fig:aggregator}
\end{figure*}

\paragraph{External call patterns}
We interpret an external call pattern $p_{\text{call}}$ to the timed
db-net with boundaries
$\sem{p_{\text{call}}} : A \otimes B \to B \otimes A$ shown
in~\cref{fig:external_call}, where the boundary ports $ch_{sc1}$ and
$ch_{sc2}$ are meant to be plugged into the pattern representing the
external process called.
Token $x$ feeding into transition $T2$ ensures that the external process does not inject unasked for
messages into the pattern.

\begin{figure}[bt]
	\centering
	\includegraphics[width=\columnwidth]{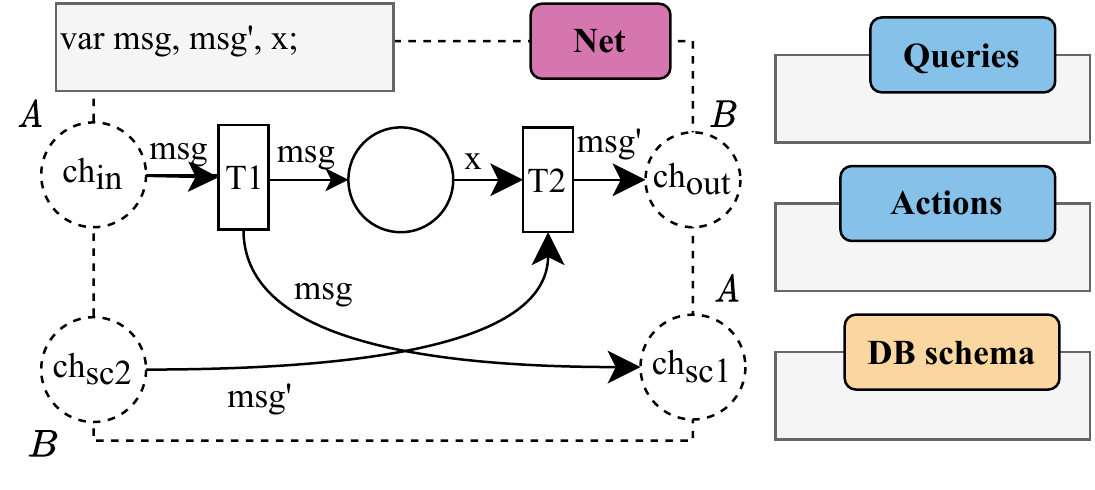}
	\caption{Interpretation of an external call pattern.}
	\label{fig:external_call}
\end{figure}

\subsection{Interpreting integration pattern contract graphs}
\label{sub:interpreting_ipcgs}
We now show how to interpret not just individual nodes from an integration pattern contract graph, but how to also take edges into account.
We first enrich the interpretations of single patterns with transitions and guards to enable and enforce the concepts from output and input contracts in \cref{sec:CPT-net-construction},
and then prove that composing the interpretation of individual patterns according to how they are connected in the graph gives rise to a well-formed timed db-net in \cref{sec:translation-edges}.

\subsubsection{Taking contract concepts into account}
\label{sec:CPT-net-construction}

Recall that a pattern contract also represents concepts, \ie properties of the exchanged data, such as if a pattern is able to process or produce signed, encrypted or encoded data.
A message can only be sent from one pattern to another if their contracts match, \ie if they agree on these properties.
To reflect this in the timed db-nets semantics, we enrich all colorsets to also keep track of this information: given a place $P$ with colorset $C$, we construct the colorset $C \times \{yes, no\}^3$, where the color $(x, b_{\text{sign}}, b_{\text{encr}}, b_{\text{enc}})$ is intended to mean that the data value $x$ is respectively signed, encrypted and encoded or not according to the yes/no values $b_{\text{sign}}$, $b_{\text{encr}}$, and $b_{\text{enc}}$.
To enforce the contracts, we also make sure that every token entering an input place $c_{in}$ is guarded according to the input contract by creating a new place $ch'_{in}$ and a new transition from $ch'_{in}$ to $ch_{in}$, which conditionally forwards tokens whose properties match the contract.
The new place $ch'_{in}$ replaces $ch_{in}$ as an input place.
Dually, for each output place $ch_{out}$ we create a new place $ch'_{out}$ and a new transition  from $ch_{out}$ to $ch'_{out}$ which ensures that all tokens satisfy the output contract, via a new transition for each combination of output contract values.
The new place $ch'_{out}$ replaces $ch_{out}$ as an output place.
Formally, the construction is as follows:

\begin{definition}
\label{def:pc:construction}
Let $\mathcal{X} : \otimes_{i<m} c_i \to \otimes_{i < n} c'_i$ be an open timed db-net, and $\vec{C} = IC_1, \ldots, IC_{m}, OC_1, \ldots, OC_{n}$ be a list of integration concepts with $IC_i, OC_j \in \CPT$. Define the open timed db-net
\[
  \cptize{\mathcal{X}}{\vec{C}} : \otimes_{i<m} (c_i \times \{yes,no\}^3) \to \otimes_{i < n} (c'_i  \times \{yes,no\}^3)
\]
  with the same type domains, persistence layers and data logic layers as $\mathcal{X}$, but with control layer
  \[
    \mathcal{N}' = (P', T', F'_{in}, F'_{out}, F'_{rb}, \mathtt{color}', \mathtt{query}, \texttt{guard}', \mathtt{action}')
  \]
  with
  \begin{align*}
    P' &= P \uplus \{ch'_{\text{in},1}, \dots, ch'_{\text{in},m}, ch'_{\text{out},1}, \dots, ch'_{\text{out},1}\} \\
    T' &= T \uplus T_{in} \uplus T_{out} \\
  \end{align*}
  where $T_{in} =  \{t_{\text{in},1}, \dots, t_{\text{in},m},\}$ and
  \[
    T_{\text{out}} = \{t_{\text{out},1, \vec{b}}, \dots, t_{\text{out},n, \vec{b}}\,|\, \vec{b} \in \{yes, no\}^3\}
  \]
  \begin{align*}
    F'_{in}(x) &=
       \begin{cases}
         F_{in}(p, t) & \text{if $x = (\inn_{P}(p), \inn_{T}(t))$} \\
         \{(y, y_{\text{sign}}, y_{\text{encr}}, y_{\text{enc}}) \} & \text{if $x = (ch_{in,i}, t_{in,i})$ } \\
         \{y \} & \text{if $x = (ch_{out,j}, t_{out,j,\vec{b}})$ } \\
         & \!\!\!\!\!\!\!\!\!\!\!\!\!\!\!\!\text{and $\vec{b} = (b_{\text{sign}}, b_{\text{encr}}, b_{\text{enc}})$} \\
         & \!\!\!\!\!\!\!\!\!\!\!\!\!\!\!\!\text{with $(OC_j)_{\CPT}(p) \in \{b_p, any\}$} \\
         \emptyset & \text{otherwise}
       \end{cases} \\
    F'_{out}(x) &=
       \begin{cases}
         F_{out}(p, t) & \text{if $x = (\inn_{P}(p), \inn_{T}(t))$} \\
         \{y \} & \text{if $x = (ch_{\text{in},i}, t_{\text{in},i})$ } \\
         \{(y, \vec{b})\} & \text{if $x = (ch_{\text{out}, j}, t_{\text{out},j, \vec{b}})$ } \\
         & \!\!\!\text{and $\vec{b} = (b_{\text{sign}}, b_{\text{encr}}, b_{\text{enc}})$} \\
         & \!\!\!\text{with $(OC_j)_{\CPT}(p) \in \{b_p, any\}$} \\
         \emptyset & \text{otherwise}
       \end{cases}
  \end{align*}
  \begin{align*}
    F'_{rb}(x) &=
       \begin{cases}
         F_{rb}(p, t) & \text{if $x = (\inn_{P}(p), \inn_{T}(t))$} \\
         \emptyset & \text{otherwise}
       \end{cases}
  \end{align*}
  \begin{align*}
    \mathtt{guard}'(\inn_{T}(t) &= \mathtt{guard}(t) \\
    \mathtt{guard}'(t_{\text{in},i}) &= \bigwedge_{\{ p\ |\ (IC_i)_{\CPT}(p) \neq any\}} y_{p} = (IC_i)_{\CPT}(p) \\
    \mathtt{guard}'(t_{\text{out},j}) &= \top \\
    \mathtt{action}' &= [\mathtt{action}, t_i \mapsto -, t'_j \mapsto -] \\
    \tau' &= [\tau, t_i \mapsto [0, \infty], t'_j \mapsto [0, \infty]]
  \end{align*}
\end{definition}

The pattern contract construction in~\cref{def:pc:construction}, can again be realized as template translation on an inter pattern level, as shown in
\cref{fig:cpt_input_output_many}.
On the input side, a token token $(y, b_{\text{sign}}, b_{\text{encr}}, \allowbreak{}b_{\text{enc}})$ only enables the transition $t_{\text{in},i}$ if $(b_{\text{sign}}, b_{\text{encr}}, b_{\text{enc}})$ fulfils the input contract $IC_i$, in which case the metadata is stripped and only the message $y$ passed to the actual pattern. On the output side, any of the boundary transitions $t_{\text{out},i,(b_{\text{sign}}, b_{\text{encr}}, b_{\text{enc}})}$ may fire and enrich the data $y$ with metadata $(b_{\text{sign}}, b_{\text{encr}}, b_{\text{enc}})$, ready to be passed to the next pattern.

\begin{figure*}[bt]
	\centering
	\includegraphics[width=0.8\linewidth]{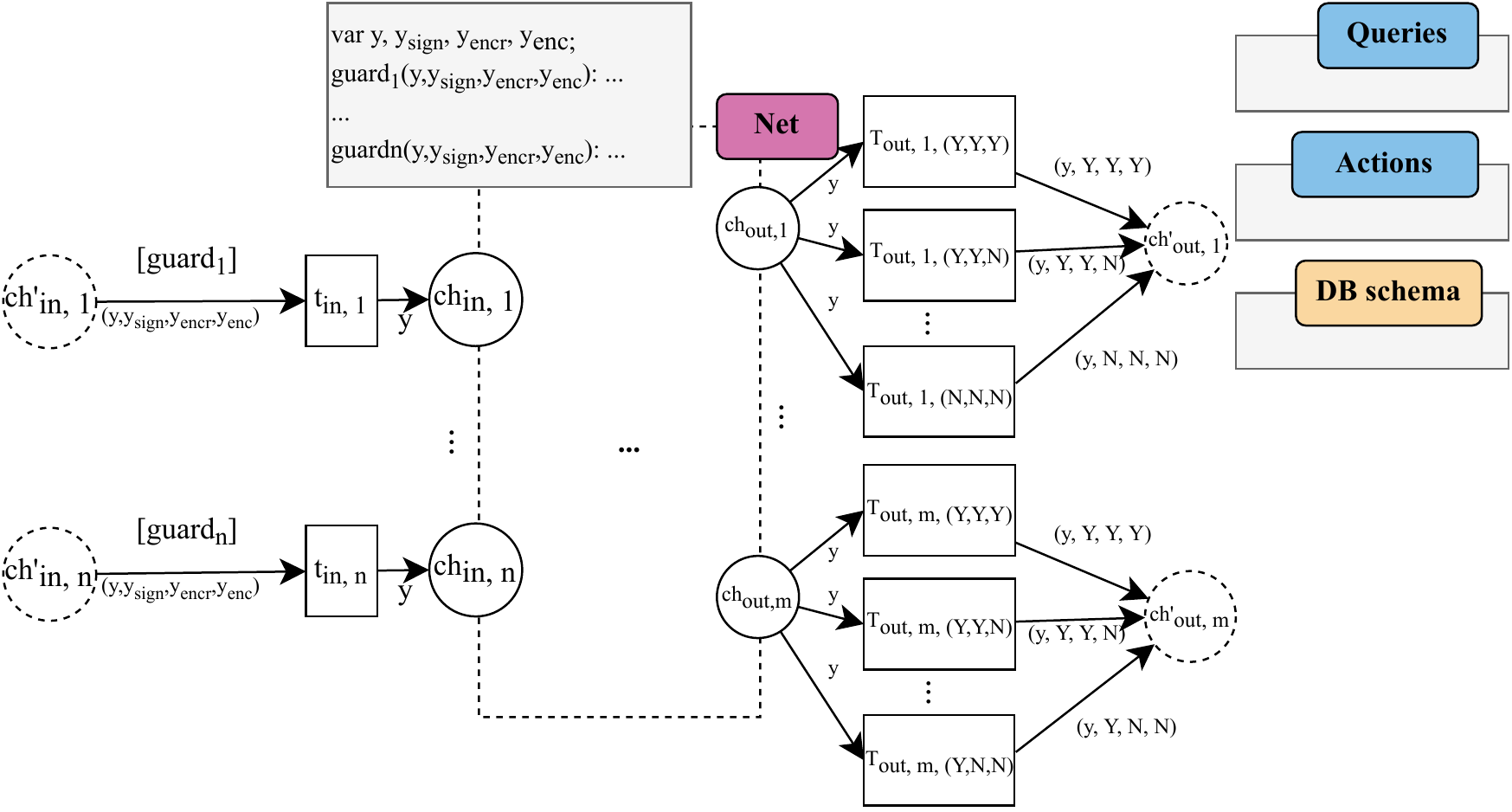}
	\caption{Boundary construction template.}
	\label{fig:cpt_input_output_many}
\end{figure*}

Let us consider two examples to gain an understanding of the construction.

\begin{example}
  \Cref{fig:message_translator_construction} shows the translation of a message translator pattern $MT$ with input contract $\{(ENC,no),$ $(ENCR,no),$ $(SIGN,any)\}$ and output contract $\{(ENC,no),$ $(ENCR,no),$ $\allowbreak(SIGN,\allowbreak{}any)\}$.
	The input transition $T'$ hence checks the guard $[x,no,any,\allowbreak{}no]$, and if it matches, the token is forwarded to the actual message translator.
	After the transformation, the resulting message $msg'$ is not encrypted, the signing is invalid, and not encoded, and thus emits $(x, no, no, no)$.
	\begin{figure*}[bt]
	    \centering
	    \includegraphics[width=.8\linewidth]{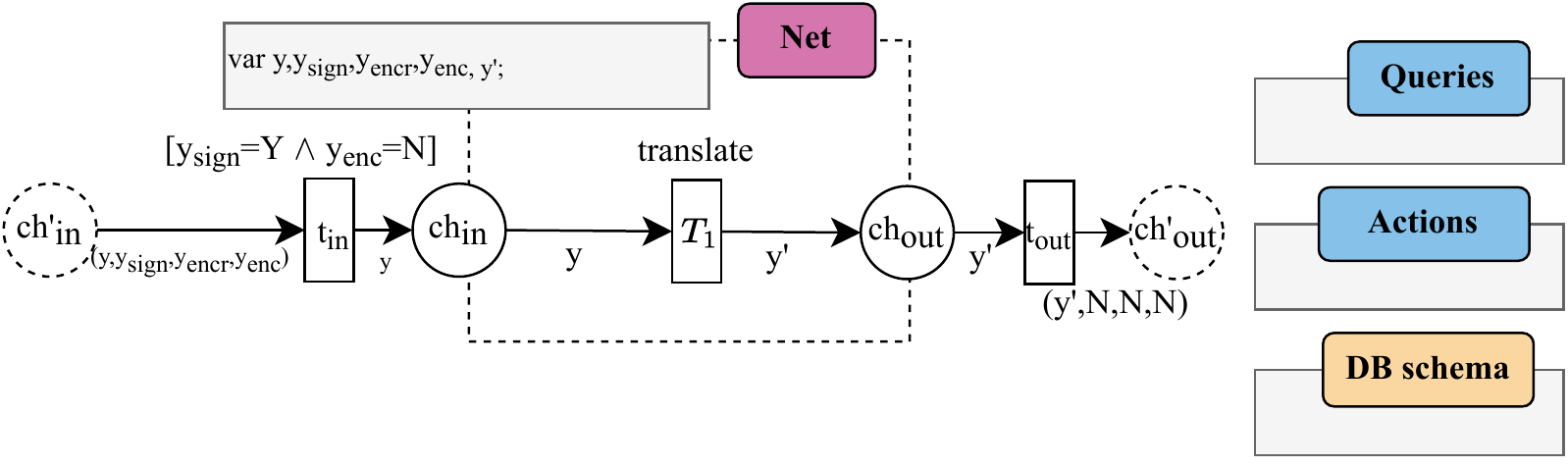}
	    \caption{Example: message translator construction}
	    \label{fig:message_translator_construction}
	\end{figure*}
\end{example}

\begin{example}
A join router structurally combines many incoming to one outgoing message channel without accessing the data.
Consequently, both input and output contracts have $any$ for all properties.
\Cref{fig:join_router_construction} shows the result of the boundary construction for the join router.
The input boundary does not enforce CPT constraints, and thus no guards are defined for the transitions.
The output boundary, however, supplies all 8 different permutations of \{yes,no\} for the three CPT properties.
\begin{figure*}[bt]
  \centering
  \includegraphics[width=.8\linewidth]{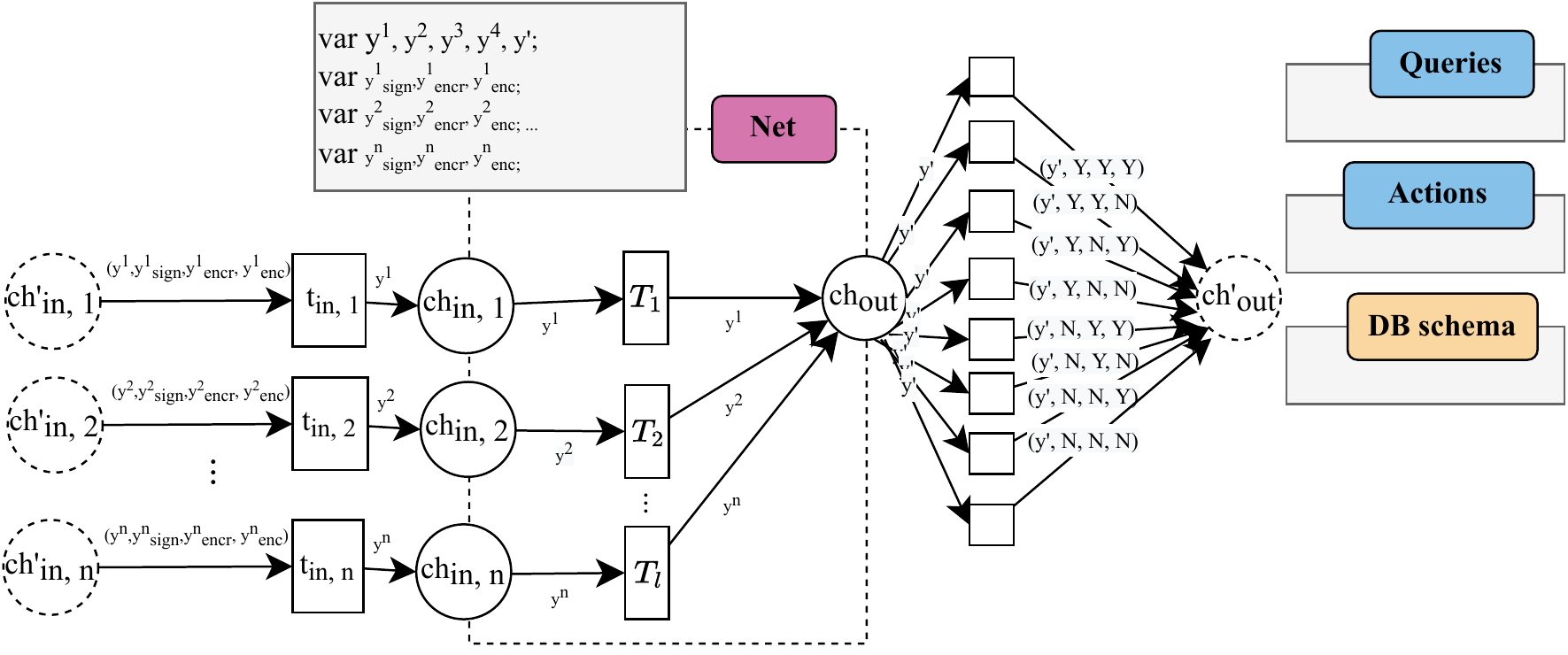}
  \caption{Join router construction}
  \label{fig:join_router_construction}
\end{figure*}
\end{example}

\subsubsection{Synchronising Pattern Compositions and correctness of the translation}
\label{sec:translation-edges}

We are now in a position to define the full translation of a correct integration pattern contract graph $G$.
For the translation to be well-defined, we need only data element correctness of the graph.
Concept correctness is used to show that in the nets in the image of the translation, tokens can always flow from the translation of the start node to the translation of the end node.

\begin{theorem}
  \label{thm:translation-correct}
  Let a correct integration pattern contract graph $G$ be given. For each node $p$, consider the timed db-net
  \[
    \cptize{\sem{p}}{\iC(p), \oC(p)} : \bigotimes_{i = 1}^k \cin{p}{i} \to \bigotimes_{j = 1}^m \cout{p}{j}
  \]
  Use the graphical language \cite{selinger2009graphical} enabled by \cref{thm:nets-moncat} to compose these nets according to the edges of the graph.
  The resulting timed db-net is then well-defined, and has the option to complete, \ie from each marking reachable from a marking with a token in some input place, it is possible to reach a marking with a token in an output place.
\end{theorem}
\begin{proof}
  Since the graph is assumed to be correct, all input contracts match the output contracts of the nets composed with it, which by the data element correctness means that the boundary configurations match, so that the result is well-defined.

  To see that the constructed net also has the option to complete, first note that the interpretations of basic patterns in \cref{sec:interpretation-atomic} do (in particular, one transition is always enabled in the translation of a conditional fork pattern in \cref{fig:content_based_router}, and the aggregate transition will always be enabled after the timeout in the translation of a merge pattern in \cref{fig:aggregator}).
  By the way the interpretation is defined, all that remains to show is that if $N$ and $N'$ have the option to complete, then so does $\cptize{N}{\vec{C}} \circ \cptize{N'}{\vec{C'}}$, if the contracts $\vec{C}$ and $\vec{C'}$ match.
  Assume a marking with a token in an input place of $N'$.
  Since $N'$ has the option to complete, a marking with a token in an output place of $N'$ is reachable, and since the contracts match, this token will satisfy the guard imposed by the $\cptize{N}{\vec{C}}$ construction.
  Hence a marking with a token in an input place of $N$ is reachable, and the statement follows, as $N$ has the option to complete.
  \qed
\end{proof}

\subsection{Discussion}

Solely giving an interpretation of pattern compositions as timed db-nets does not guarantee correctness.
However, \cref{thm:translation-correct} gives confidence in the translation itself, as it shows that the output of the translation is structurally well-behaved (\ie input cardinalities match output cardinalities), and also semantically well-behaved, in the sense that tokens flowing through the resulting timed db-net cannot get \enquote{stuck} (\ie having the option to complete).
Note that having a translation targeting timed db-nets means that one can now formulate formal conjectures and prove them, perhaps for classes of integration patterns.


%% file: optimization_realization_short.tex
\section{Optimization Strategy Realization}
\label{sec:realization}
In this section we formally define the optimizations from the different strategies identified in \cref{tab:optimization_strategies} in the form of a rule-based graph rewriting system (addressing \refreq{improvements}).
This gives a formal framework in which different optimizations can be compared.
We begin by describing the graph rewriting framework, and subsequently apply it to define the optimizations.


\subsection{Graph Rewriting}
\label{sec:graph-rewriting}
Graph rewriting provides a visual framework for transforming graphs in a rule-based fashion.
A graph rewriting rule is given by two embeddings of graphs $L \hookleftarrow K \hookrightarrow R$, where $L$ represents the left hand side of the rewrite rule, $R$ the right hand side, and $K$ their intersection (the parts of the graph that should be preserved by the rule).
A rewrite rule can be applied to a graph $G$ after a match of $L$ in $G$ has been given as an embedding $L \hookrightarrow G$; this replaces the match of $L$ in $G$ by $R$.
The application of a rule is potentially non-deterministic: several distinct matches can be possible~\cite{Ehrig:2006:FAG:1121741}.
Visually, we represent a rewrite rule by a left hand side and a right hand side graph colored green and red: green parts are shared and represent $K$, while the red parts are to be deleted in the left hand side, and inserted in the right hand side respectively.
For instance, the following rewrite rule moves the node $P_1$ past a fork by making a copy in each branch, changing its label from $c$ to $c'$ in the process:
\begin{center}
\includegraphics[scale=0.6]{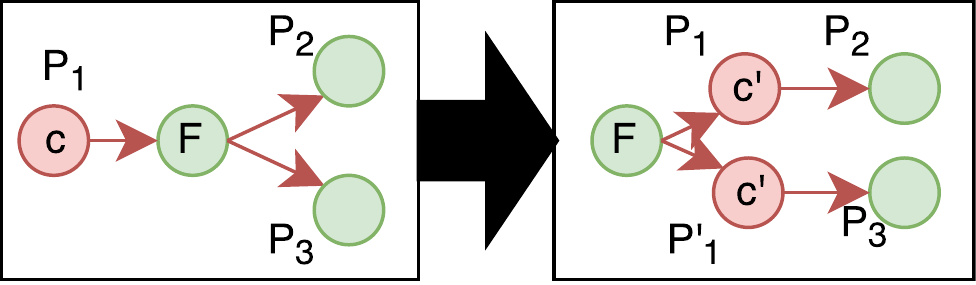}
\end{center}
Formally, the rewritten graph is constructed using a double-pushout (DPO)~\cite{ehrig1973graph} from category theory.
We use DPO rewriting since rule applications are side-effect free (\eg no \enquote{dangling} edges) and local (\ie all graph changes are described by the rules).
We additionally use Habel and Plump's relabeling DPO extension~\cite{habel2002relabelling} to facilitate the relabeling of nodes in partially labeled graphs.
In \cref{fig:summarized_example}, we showed contracts and characteristics in dashed boxes, but in the rules that follow, we will represent them as (schematic) labels inside the nodes for space reasons.

In addition, we also consider rewrite rules parameterized by graphs, where we draw the parameter graph as a cloud (see \eg \cref{fig:redundant_control_flow_v2} for an example).
A cloud represents any graph, sometimes with some side-conditions that are stated together with the rule.
When looking for a match in a given graph $G$, it is of course sufficient to instantiate clouds with subgraphs of $G$ --- this way, we can reduce the infinite number of rules that a parameterized rewrite rule represents to a finite number.
Parameterized rewrite rules can formally be represented using substitution of hypergraphs~\cite{plump1994hypergraph} or by !-boxes in open graphs~\cite{bangBoxes}.
Since we describe optimization strategies as graph rewrite rules, we can be flexible with when and in what order we apply the strategies.
We apply the rules repeatedly until a fixed point is reached, \ie when no further changes are possible, making the process idempotent.
Each rule application preserves IPCG correctness in the sense of \cref{def:pattern_composition_new}, because input contracts do not get more specific, and output contracts remain the same.
Methodologically, the rules are specified by pre-conditions, change primitives, post-conditions and an optimization effect, where the pre- and post-conditions are implicit in the applicability and result of the rewriting rule.

\subsection{OS-1: Process Simplification}
We first consider the process simplification strategies from \cref{sec:stategies} OS-1 to OS-3 that mainly strive to reduce the model complexity and latency.

\subsubsection{Redundant sub-process} 
\label{opt:redundant-subprocess}
This optimization removes redundant copies of the same sub-process within a process.

\begin{figure}[bt]
    \subfigure[Redundant sub-process]{\label{fig:redundant_control_flow_v2}\includegraphics[width=.7\columnwidth]{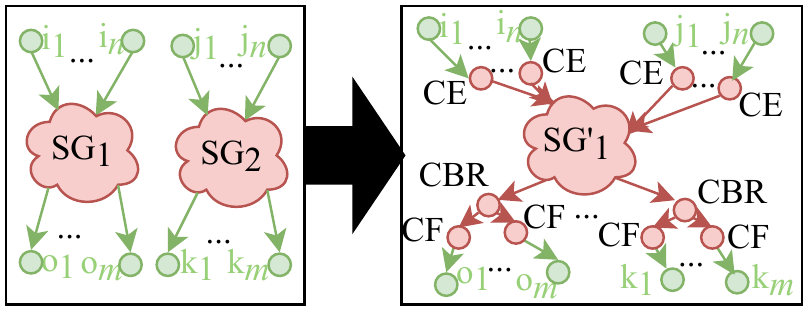}} \\
    \hfill
    \subfigure[Combine sibling patterns]{\label{fig:combine_sibling_nodes_v2_b}\includegraphics[width=.7\columnwidth]{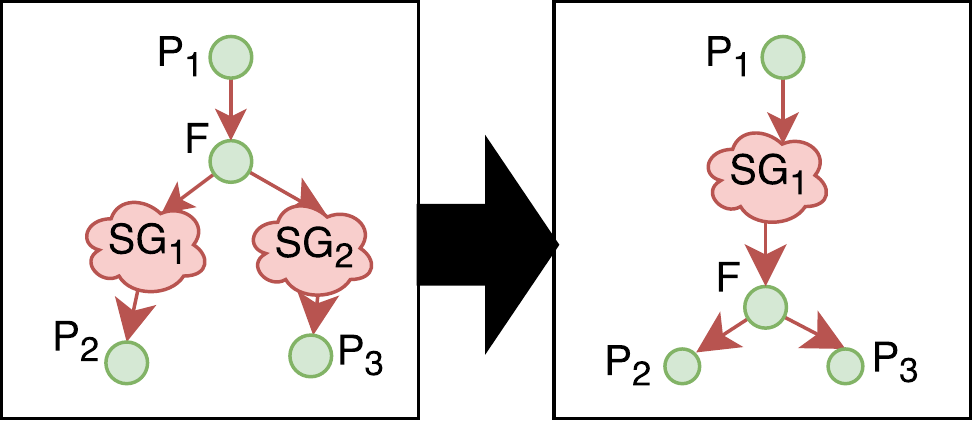}} 
    \vspace{-.3cm}
    \caption{Rules for redundant sub-process and combine sibling patterns.}
    \label{fig:opt:redundant}
\end{figure}

\labeltitle{Change primitives:} The rewriting is given by the rule in \cref{fig:redundant_control_flow_v2}, where $SG1$ and $SG2$ are isomorphic pattern graphs with in-degree $n$ and out-degree $m$.
The Content Enricher (CE) node is a message processor pattern from~\cref{fig:siso_w_storage} with a pattern characteristic $(PRG, (\prg{addCtxt}, [0, \infty)))$ for an enrichment program $\prg{addCtxt}$ which is used to add content to the message (does it come from the left or right subgraph?).
Similarly, the Content Filter (CF) is a message processor, with a pattern characteristic $(PRG, (\prg{removeCtxt},$ $[0, \infty)))$  for an enrichment program $\prg{removeCtxt}$ which is used to remove the added content from the message again.
Moreover, the Content-based Router (CBR) node is a conditional fork pattern from~\cref{fig:content_based_router} with a pattern characteristic $(CND, \{\cnd{fromLeft?}\})$ for a condition $\cnd{fromLeft?}$ which is used to route messages depending on their added context.
In the right hand side of the rule, the $CE$ nodes add the context of the predecessor node to the message in the form of a content enricher pattern, and the $CBR$ nodes are content-based routers that route the message to the correct recipient based on the context introduced by $CE$.
The graph $SG'_1$ is the same as $SG_1$, but with the context introduced by $CE$ copied along everywhere.
This context is stripped off the message by a content filter $CF$.

\labeltitle{Effect:} The optimization is beneficial for model complexity when the isomorphic subgraphs contain more than $n + m$ nodes, where $n$ is the in-degree and $m$ the out-degree of the isomorphic subgraphs.
The latency reduction is by the factor of subgraphs minus the latency introduced by the additional $n$ $CE$ nodes, $m$ $CBR$ nodes and $2m$ $CF$ nodes.

\subsubsection{Combine sibling patterns}
\label{opt:combine-siblings}
Sibling patterns have the same parent node in the pattern graph (\eg they follow a non-conditional forking pattern) with channel cardinality of $1$:$1$.
Combining them means that only one copy of a message is traveling through the graph instead of two --- for this transformation to be correct in general, the siblings also need to be side-effect free, \ie no external calls.

\labeltitle{Change primitives:} The rule is given in \cref{fig:combine_sibling_nodes_v2_b}, where $SG_1$ and $SG_2$ are isomorphic side-effect free pattern graphs, and $F$ is a fork. 

\labeltitle{Effect:} The model complexity and latency are reduced by the model complexity and latency of $SG_2$.

\subsection{OS-2: Data Reduction} \label{sub:data_reduction}
Now, we consider data reduction optimization strategies, which mainly target improvements of the message throughput (incl. reducing element cardinalities).
These optimizations require that pattern input and output contracts are regularly updated with snapshots of element data sets $EL_{\mathrm{in}}$ and $EL_{\mathrm{out}}$ from live systems, \eg from experimental measurements through benchmarks~\cite{ritter2016benchmarking}.

\subsubsection{Early-Filter}
\label{opt:early-filter}
A filter pattern can be moved to or inserted prior to some of its successors to reduce the data to be processed.
The following types of filters have to be differentiated:
\begin{itemize}
    \itemsep0em
    \item A \emph{message filter} removes messages with invalid or incomplete content.
    It can be used 
    to prevent exceptional situations, and thus improves 
    stability. 
    \item A \emph{content filter} removes elements from messages, thus reduces the amount of data passed to subsequent patterns.
\end{itemize}
Both patterns are message processors in the sense of~\cref{fig:siso_w_storage}.
The content filter assigns a filter function $(prg_1, ([0,\infty)) \rightarrow f(msg,value)$ to remove data from the message (\ie without temporal information), and  the message filter assigns a filter condition $\{(cond_1)\} \rightarrow g(msg)$.
\begin{figure}[bt]
    \subfigure[Early Filter]{\label{fig:early_filter_v2_insert}\includegraphics[width=.7\columnwidth]{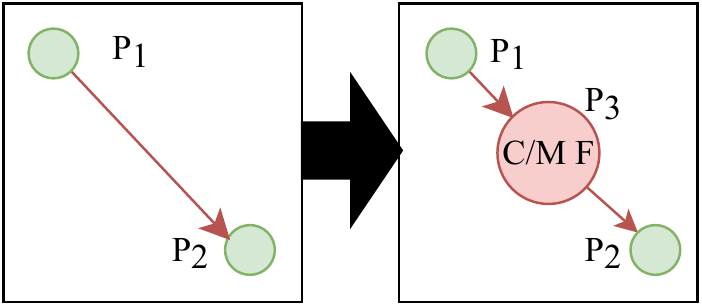}}\\
    \hfill
    \subfigure[Early Mapping]{\label{fig:early_mapping_v2}\includegraphics[width=.7\columnwidth]{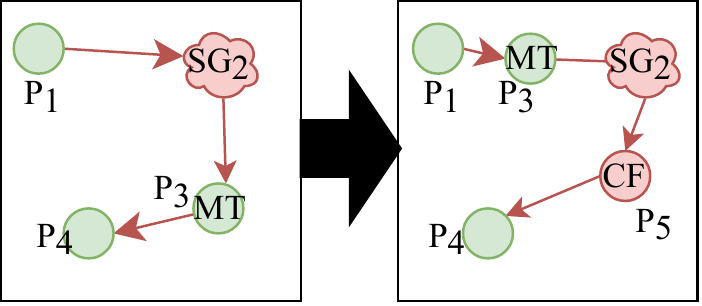}}
    \vspace{-.3cm}
    \caption{Rules for early-filter and early-mapping.}
    \label{fig:opt:early_filter}
\end{figure}

\labeltitle{Change primitives:} The rule is given in \cref{fig:early_filter_v2_insert}, where $P_3$ is either a content or message filter matching the output contracts of $P_1$ and the input contract of $P_2$, removing the data not used by $P_2$.

\labeltitle{Effect:} Message throughput increases by the ratio of the number of reduced 
elements that are processed per second, unless limited by the throughput of the additional pattern.

\subsubsection{Early-Mapping}
\label{opt:early-mapping}
A mapping that reduces the number of elements in a message can increase the message throughput.

\labeltitle{Change primitives:} The rule is given in \cref{fig:early_mapping_v2}, where $P_3$ is an element reducing message mapping compatible with both $SG_2$, $P_4$, and $P_1$, $SG_2$, and where $P_4$ does not modify the elements mentioned in the output contract of $P_3$.
Furthermore $P_5$ is a content filter, which ensures that the input contract of $P_4$ is satisfied.
The Message Translator (MT) node is a message processor pattern from~\cref{fig:siso_w_storage} with a pattern characteristic $(PRG, (prg, [0, \infty)))$ for some program $prg$ which is used to transform the message.

\labeltitle{Effect:} The message throughput for the subgraph subsequent to the mapping increases by the ratio of the number of unnecessary data elements processed.

\subsubsection{Early-Aggregation}
\label{opt:early-aggregation}
A micro-batch processing region is a subgraph which contains patterns that are able to process multiple messages combined to a multi-message~\cite{DBLP:conf/bncod/Ritter17} or one message with multiple segments with an increased message throughput.
The optimal number of aggregations is determined by the highest batch-size for the throughput ratio of the pattern with the lowest throughput, if latency is not considered.

\begin{figure}[bt]
    \subfigure[Early-Aggregation]{\label{fig:early_aggregate}\includegraphics[width=0.7\columnwidth]{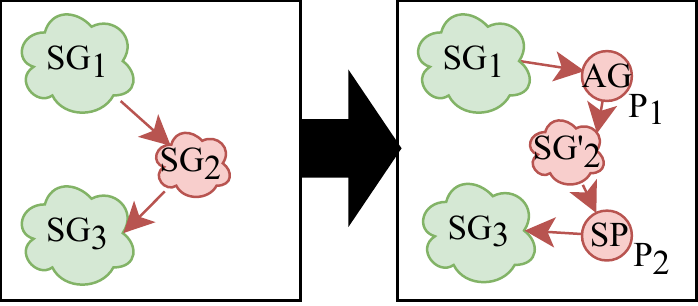}} \\ 
    \hfill
    \subfigure[Early-Claim Check]{\label{fig:early_claim_v2}\includegraphics[width=0.7\columnwidth]{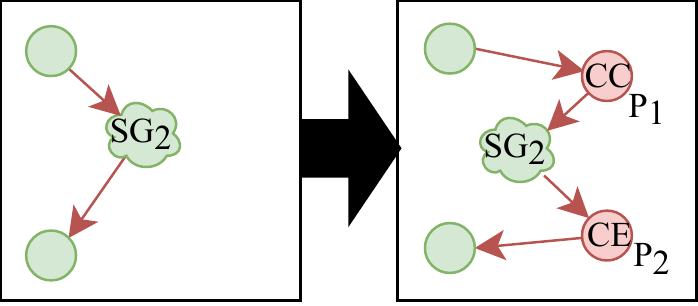}}
    \caption{Rules for early-aggregation and early-claim check}
    \label{fig:op:early_aggregate_claim}
\end{figure}

\labeltitle{Change primitives:} The rule is given in \cref{fig:early_aggregate}, where $SG_2$ is a micro-batch processing region, $P_1$ an aggregator, $P_2$ a splitter which separates the batch entries to distinct messages to reverse the aggregation, and $SG'_2$ finally is $SG_2$ modified to process micro-batched messages.
The Aggregator (AG) node is a merge pattern from~\cref{fig:aggregator} with a pattern characteristic $\{(CND,\{cnd_{cr}, cnd_{cc}\}),$ $(PRG, prg_{agg},$ $(v_1,v_2)\}$ for some correlation condition $cnd_{cr}$, completion condition $cnd_{cc}$, aggregation function $prg_{agg}$, and timeout interval $(v_1,v_2)$. The Splitter (SP) node is a message processor from~\cref{fig:siso_w_storage} with a pattern characteristic $(PRG, (prg, [0, \infty)))$ for some split function $prg$ which is used to split the message into several ones.

\labeltitle{Effect:} The message throughput is the minimal pattern throughput of all patterns in the micro-batch processing region.
If the region is followed by patterns with less throughput, only the overall latency might be improved.

\subsubsection{Early Claim Check}
\label{opt:early-claim-check}
If a subgraph does not contain a pattern with message access, the message payload can be stored intermediately persistently or transiently (depending on the quality of service level) and not moved through the subgraph.
For instance, this applies to subgraphs consisting of data independent control-flow logic only, or those that operate entirely on the message header (\eg header routing).

\labeltitle{Change primitives:} The rule is given in \cref{fig:early_claim_v2}, where $SG_2$ is a message access-free subgraph, $P_1$ a claim check that stores the message payload and adds a claim to the message properties (and possibly routing information to the message header), and $P_2$  a content enricher that adds the original payload to the message.
The Claim Check (CC) node is a message processor from~\cref{fig:siso_w_storage} with a pattern characteristic $(PRG, (\_, [0, \infty)))$, which stores the message for later retrieval.

\labeltitle{Effect:} The main memory consumption and CPU load decreases, which could increase the message throughput of $SG_2$, if the claim check and content enricher pattern throughput is greater than or equal to the improved throughput of each of the patterns in the subgraph.

\subsubsection{Early-Split}
\label{opt:early-split}
Messages with many segments can be reduced to several messages with fewer segments, and thereby reducing the processing 
per message.
A segment is an iterable part of a message, such as a list of elements.
When such a message grows bigger, the message throughput of a set of adjacent patterns might decrease, compared to the expected performance for a single segment;
a phenomenon called \emph{segment bottleneck sub-sequence}.
Algorithmically, such bottlenecks can be found using max flow-min cut techniques based on workload statistics. 
The splitter (SP) node is a message processor from~\cref{fig:siso_w_storage} with a pattern characteristic $(PRG, (prg, [0, \infty)))$, for some split program $prg$.

\begin{figure}[bt]
    \subfigure[Early Split]{\label{fig:early_split_v2}\includegraphics[width=0.7\columnwidth]{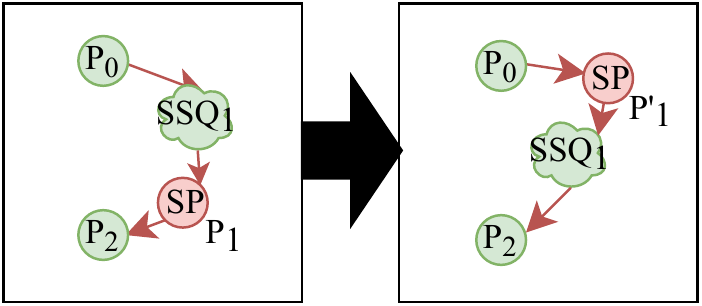}}\\
    \hfill
    \subfigure[Early Split (inserted)]{\label{fig:early_split_v2_insert}\includegraphics[width=0.7\columnwidth]{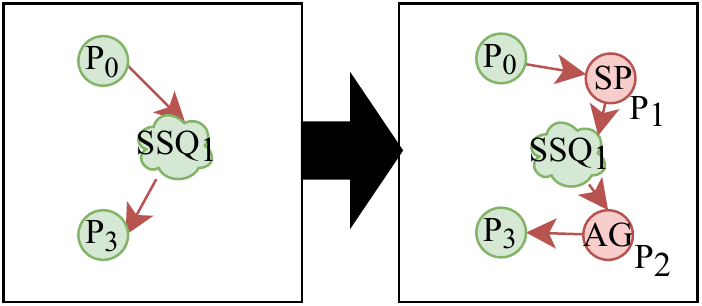}}
    \caption{Rules for early split.}
    \label{fig:op:early_split}
\end{figure}

\labeltitle{Change primitives:} The rule is given in \cref{fig:op:early_split}, where $SSQ_1$ is a segment bottleneck sub-sequence.
If $SSQ_1$ already has an adjacent splitter, \cref{fig:early_split_v2} applies, otherwise \cref{fig:early_split_v2_insert}.
In the latter case, $SP$ is a splitter and $P_2$ is  an aggregator that re-builds the required segments for the successor in $SG_2$.
For an already existing splitter $P_1$ in \cref{fig:early_split_v2}, the split condition has to be adjusted to the elements required by the input contract of the subsequent pattern in $SSQ_1$.
In both cases we assume that the patterns in $SSQ_1$ deal with single- and multi-segment messages; otherwise all patterns have to be adjusted. 

\labeltitle{Effect:} The message throughput increases by the ratio of reduced number of message segments per message, if the throughput of the moved / added splitter (and aggregator) $\geq$ throughput of each of the patterns in the segment bottleneck sub-sequence after the segment reduction.

\subsection{OS-3: Parallelization}
Parallelization optimization strategies increase message throughput,
and again require experimentally measured message throughput statistics, \eg from benchmarks~\cite{ritter2016benchmarking}.

\subsubsection{Sequence to parallel}
\label{opt:seq-to-parallel}
A bottleneck sub-sequence with channel cardinality 1:1 can also be handled by distributing its input and replicating its logic.
The parallelization factor is the average message throughput of the predecessor and successor of the sequence divided by two, which denotes the improvement potential of the bottleneck sub-sequence.
The goal is to not overachieve the mean of predecessor and successor throughput with the improvement to avoid iterative re-optimization.
Hence the optimization is only executed, if the parallel sub-sequence reaches lower throughput than their minimum.

\begin{figure}[bt]
    \subfigure[Sequence to parallel]{\label{fig:sequence_to_parallel_v2}\includegraphics[width=0.7\columnwidth]{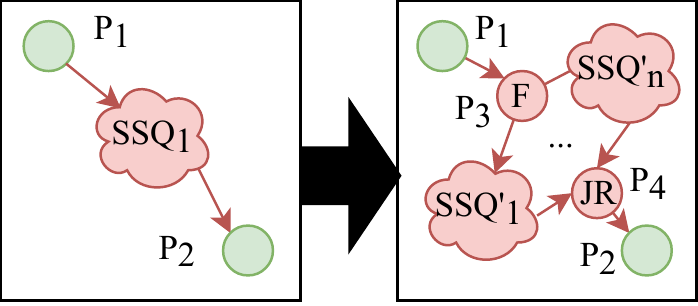}} \\
    \hfill
    \subfigure[Merge parallel]{\label{fig:op:merge_parallel_v2}\includegraphics[width=0.7\columnwidth]{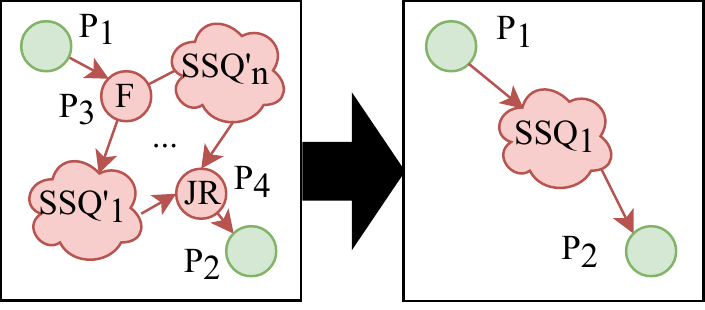}}
    \vspace{-.3cm}
    \caption{Rules for sequence to parallel variants.}
    \label{fig:op:parallel}
\end{figure}

\labeltitle{Change primitives:} The rule is given in \cref{fig:sequence_to_parallel_v2}, where $SSQ_1$ is a bottleneck sub-sequence, $P_2$ a fork node, $P_3$ a join router, and each $SSQ'_k$ is a copy of $SSQ_1$, for $1 \leq k \leq n$.
The parallelization factor $n$ is a parameter of the rule.

\labeltitle{Effect:} The message throughput improvement rate depends on the parallelization factor $n$, and the message throughput of the balancing fork and join router on the runtime.
For a measured throughput $t$ of the bottleneck sub-sequences, the throughput can be improved to $n \times t \leq$ average of the sums of the predecessor and successor throughput, which is limited by the upper boundary of the balancing fork or join router.

\subsubsection{Merge parallel}
\label{opt:merge-parallel}
The balancing fork and join router realizations can limit the throughput in some runtime systems, so that a parallelization decreases the throughput,
\eg when a fork or a join has smaller throughput than a pattern in the following sub-sequence.

\labeltitle{Change primitives:} The rule is given in \cref{fig:op:merge_parallel_v2}, where $P_3$ and $P_4$ limit the message throughput of each of the $n$ sub-sequence copies $SSQ'_1$, \ldots, $SSQ'_n$ of $SSQ_1$.

\labeltitle{Effect:} The model complexity is reduced by $(n-1)k - 2$, where each $SSQ'_i$ contains $k$ nodes.
The message throughput might improve, since the transformation lifts the limiting upper boundary of a badly performing balancing fork or join router implementations to the lowest pattern throughput in the bottleneck sub-sequence.

\subsubsection{Heterogeneous Parallelization}
\label{opt:hetero-parallel}
A heterogeneous parallelization consists of parallel sub-sequences that are not isomorphic.
In general, two subsequent patterns $P_i$ and $P_j$ can be parallelized, if the predecessor pattern of $P_i$ fulfills the input contract of $P_j$, $P_i$ behaves read-only with respect to the data element set of $P_j$, and the combined outbound contracts of $P_i$ and $P_j$ fulfill the input contract of the successor pattern of $P_j$.
\begin{figure}[bt]
    \centering
    \includegraphics[width=0.7\columnwidth]{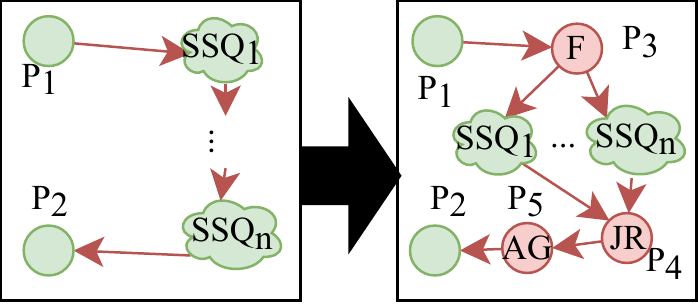}
    \caption{Heterogeneous sequence to parallel.}
    \label{fig:heterogeneous_sequence_to_parallel_v2}
\end{figure}

\labeltitle{Change primitives:} The rule is given in \cref{fig:heterogeneous_sequence_to_parallel_v2}, where the sequential sub-sequence parts $SSQ_1$, .., $SSQ_n$ are side-effect free and can be parallelized, $P_3$ is a parallel fork, $P_4$ is a join router, and $P_5$ is an aggregator that waits for messages from all sub-sequence branches before emitting a combined message that fulfills the input contract of $P_2$.

\labeltitle{Effect:}
Synchronization latency can be improved, but the model complexity increases by 3.
The latency improves from the sum of the sequential pattern latencies to the maximal latency of all sub-sequence parts plus the fork, join, and aggregator latencies.

\subsection{OS-4: Pattern Placement}
All of the data reduction optimizations discussed in \cref{sub:data_reduction} can be applied in OS-4, \ie \enquote{Pushdown to Endpoint}, by extending the placement to the message endpoints, with contracts similar to our definition.
However, due to our focus on the integration processes, we will not further elaborate on it in this work.
 
\subsection{OS-5: Reduce Interaction}
Optimization strategies that reduce interactions target a more resilient behavior of an integration process.

\subsubsection{Ignore Failing Endpoints}
\label{opt:ignore-failing-endpoints}
When endpoints fail, different exceptional situations have to be handled on the caller side.
This can come with long timeouts, which can block the caller and increase latency.
Knowing that an endpoint is unreliable can speed up processing, by immediately falling back to an alternative.

\begin{figure}[bt]
  \subfigure[Ignore Failing Endpoint]{\label{fig:ignore_failing_endpoints_v2}\includegraphics[width=0.7\columnwidth]{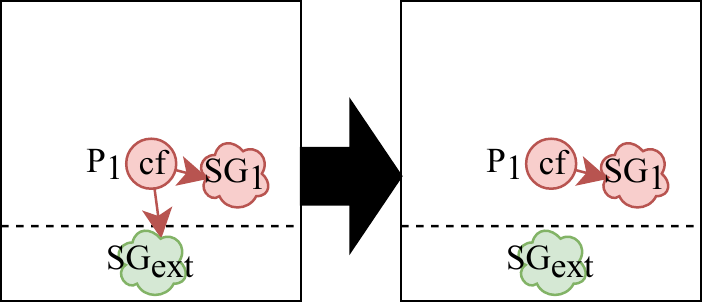}}\\
  \hfil7
  \subfigure[Try Failing Endpoint]{\label{fig:ignore_failing_endpoints_v2_b}\includegraphics[width=0.7\columnwidth]{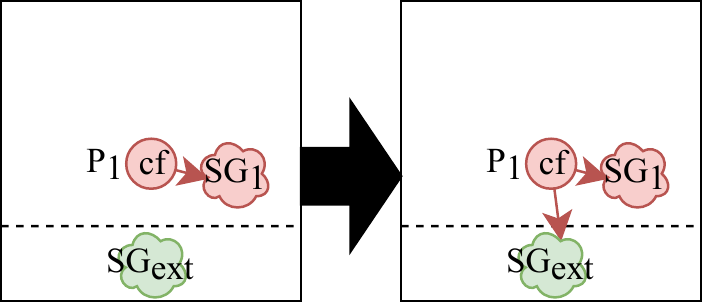}}
    \caption{Rules for ignore failing endpoints.}
    \label{fig:ignore_failing}
\end{figure}

\labeltitle{Change primitives:} The rule is given in \cref{fig:ignore_failing_endpoints_v2}, where $SG_{ext}$ is a failing endpoint, $SG_1$ and $SG_2$ subgraphs, and $P_1$ is a service call or message send pattern  with configuration $cf$.
This specifies the collected number of subsequently failed delivery attempts to the endpoint or a configurable time interval.
If one of these thresholds is reached, the process stops calling $SG_{ext}$ and does not continue with the usual processing in $SG_1$, however, invokes an alternative processing or exception handling in $SG_2$.

\labeltitle{Effect:} Besides improved latency (\ie average time to response from endpoint in case of failure), the integration process behaves more stable due to immediate alternative processing.
To not exclude the remote endpoint forever, the rule in \cref{fig:ignore_failing_endpoints_v2_b} is scheduled for execution after a period of time to try whether the endpoint is still failing.
If not, the configuration is updated to $cf'$ to avoid the execution of \cref{fig:ignore_failing_endpoints_v2}.
The retry time is adjusted depending on experienced values (\eg endpoint is down every two hours for ten minutes).

\subsubsection{Reduce Requests}
\label{opt:reduce-requests}
A \emph{message limited} endpoint, \ie an endpoint that is not able to handle a high rate of requests, can get unresponsive or fail.
To avoid this, the caller can notice this (\eg by TCP back-pressure) and react by reducing the number or frequency of requests.
This can be done be employing a throttling or even sampling pattern~\cite{ritter2016exception}, which removes messages.
An aggregator can also help to combine messages to multi-messages~\cite{DBLP:conf/bncod/Ritter17}.

\begin{figure}[bt]
  \subfigure[Reduce Requests]{\label{fig:reduce_requests_v2}\includegraphics[width=0.8\columnwidth]{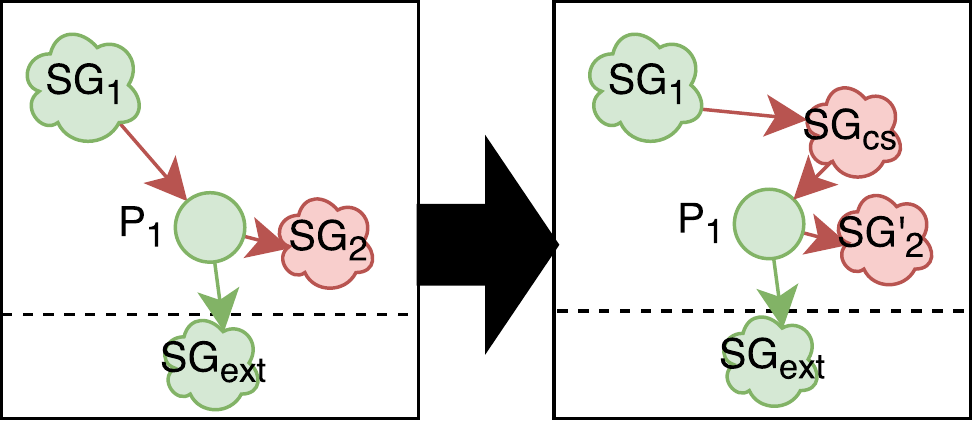}}
  \caption{Rules for reduce requests.}
  \label{fig:reduce}
\end{figure}

\labeltitle{Change primitives:} The rewriting is given by the rule in \cref{fig:reduce_requests_v2}, where $P_1$ is a service call or message send pattern, $SG_{ext}$  a message limited external endpoint, $SG_2$ a subgraph with $SG'_2$ a re-configured copy of $SG_2$ (\eg for vectorized message processing~\cite{DBLP:conf/bncod/Ritter17}), and $SG_{cs}$ a subgraph that reduce the pace, or number of messages sent.

\labeltitle{Effect:} Latency and message throughput might improve, but this optimization mainly targets stability of communication.
This is improved by configuring the caller to a message rate that the receiver can handle.

\subsection{Optimization Correctness}
\label{sub:op:correctness}
We now show that the optimizations do not change the input-output behaviour of the pattern graphs in the timed db-nets semantics, \ie if we have a rewrite rule $G \Rightarrow G'$ (cf.~\cref{sec:graph-rewriting}), then the constructed \tdbnet with boundaries $\sem{G}$ has the same observable behaviour as that of $\sem{G'}$ (addressing \refreq{correctness}). 
More formally, we mean that the transition systems of the original and the rewritten graphs are bisimilar in a certain sense,
as defined in~\cref{def:op:bisimilarity}.
At a high level, this means that $G$ can simulate $G'$ with respect to input-output behaviour, and vice versa.
Recall that we associate a labelled transition system to each \tdbnet in \cref{sec:execution-semantics}.

\begin{definition}[Functional bisimulation] 
	\label{def:op:bisimilarity}
    Let $B$ and $B'$ be \tdbnets with equal boundaries, and let $\Gamma^\mathcal{B}_{s_0} = \langle S, s_0, \rightarrow \rangle$ and $\Gamma^\mathcal{B'}_{s_0} = \langle S', s'_0, \rightarrow' \rangle$ be their associated labelled transition systems. We say that a $B$-snapshot $(I, m)$ is \emph{functionally equivalent} to a $B'$-snapshot $(I', m')$, $(I, m) \approx (I', m')$,  if $I = I'$, and $m$ and $m'$ agree on output places except for age variables. 
    Further we say that $\Gamma^\mathcal{B}_{s_0}$ is \emph{functionally bisimilar} to $\Gamma^\mathcal{B'}_{s_0}$, $\Gamma^\mathcal{B}_{s_0} \sim \Gamma^\mathcal{B'}_{s_0}$, if whenever $s_0 \rightarrow^* (I, m)$ then there is $(I', m')$ such that $s'_0 \rightarrow'^* (I', m')$, $(I, m) \approx (I', m')$, and $\Gamma^\mathcal{B}_{(I,m)} \sim \Gamma^\mathcal{B'}_{(I',m')}$, and similarly whenever $s'_0 \rightarrow^* (I', m')$ then there is $(I, m)$ such that $s_0 \rightarrow^* (I, m)$, $(I', m') \approx (I, m)$, and $\Gamma^\mathcal{B'}_{(I',m')} \sim \Gamma^\mathcal{B}_{(I,m)}$. \qedblack
\end{definition}

The notion of functional bisimulation captures the notion of having the same output behaviour, in the sense that transition system can reach a certain configuration of the output places if and only if the other one can. Note that this definition allows bisimilar transition systems to assign different token ages --- what matters is not the exact age value, but that the corresponding transitions are always possible.
Let us discuss an explanatory example of bisimulation. 

\begin{example}
	\Cref{fig:op:bisimilarity_example} shows the interpretation of a simple IPCG as a \tdbnet before and after applying the rewrite rule for the combining sibling patterns from~\cref{fig:combine_sibling_nodes_v2_b} (for simplicity without boundaries).
	\begin{figure}[bt]
		\begin{center}$
			\begin{array}{c}
			\subfigure[Before applying the rewrite rule]{\label{fig:op:bisimilarity_example_1}\includegraphics[width=1\columnwidth]{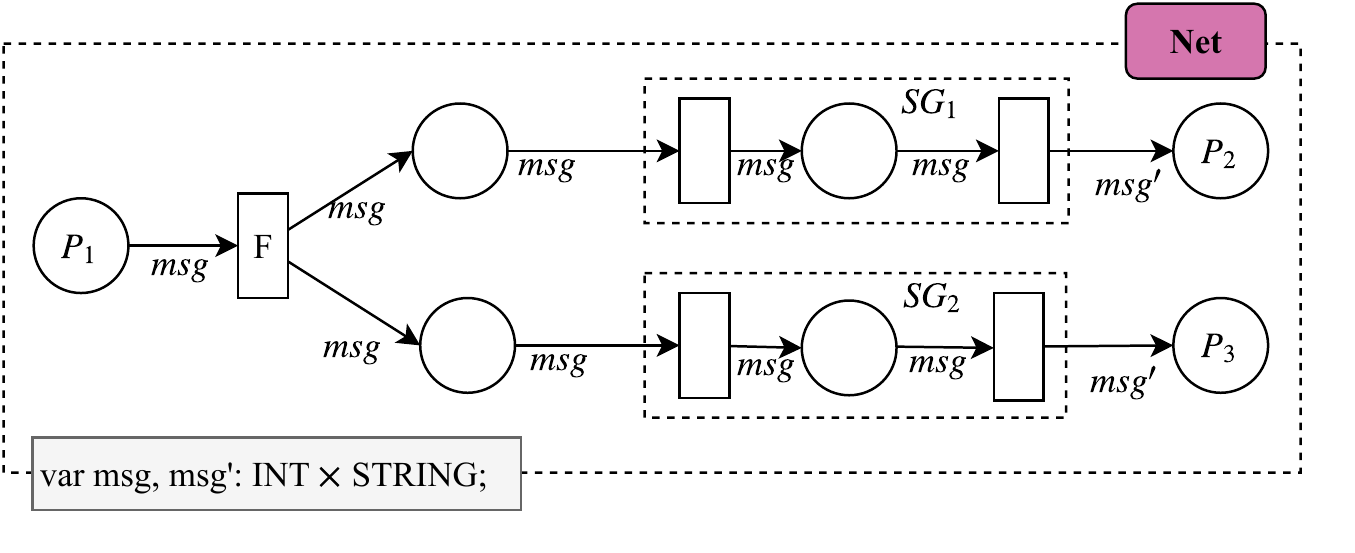}} \\
			\subfigure[After applying the rewrite rule]{\label{fig:op:bisimilarity_example_2}\includegraphics[width=1\columnwidth]{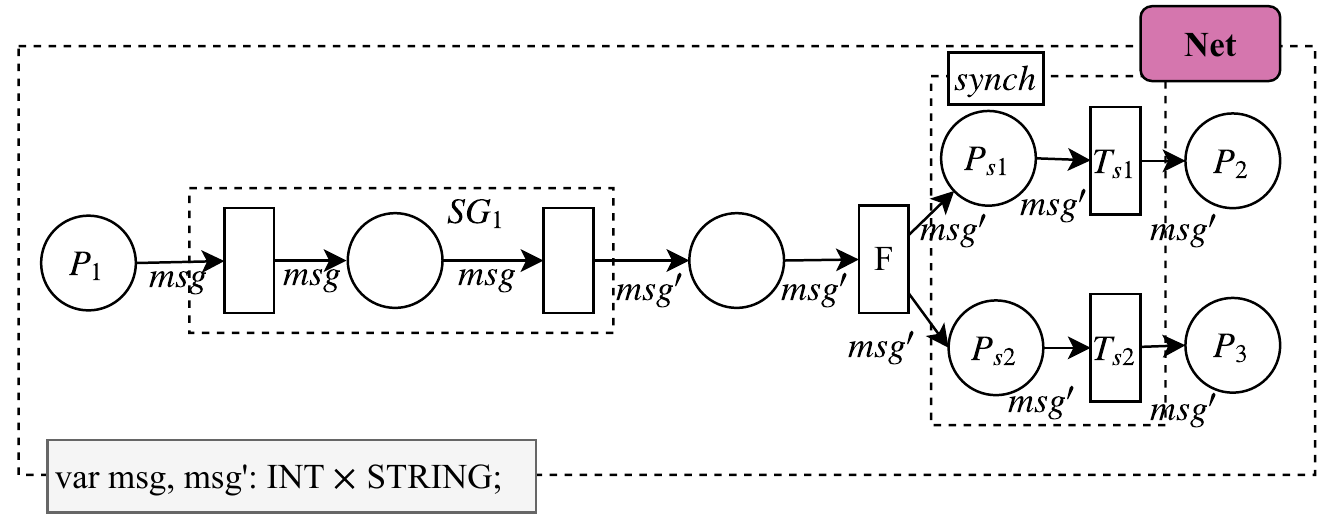}}        
			\end{array}$
		\end{center}
		\caption{Timed db-net translation of IPCGs before and after applying the \enquote{combine sibling patterns} rewrite rule}
		\label{fig:op:bisimilarity_example}
	\end{figure}
	The improvement of the optimization is to move $SG_1$ (isomorphic to $SG_2$) in front of the forking pattern $F$ and leave out $SG_2$, which reduces the modeling complexity on the right hand side (cf.~\cref{fig:op:bisimilarity_example_2}).
	The synchronization subnet \emph{synch} is required to show bisimilarity between the original and the resulting net, since tokens might be moved independently in $SG_1$ and $SG_2$ before applying the optimisation.
	The subnet (essentially transitions $T_{s1},T_{s2}$) compensates for that to ensure that places $P_2$ and $P_3$ can be reached independently as well.
	The \tdbnet $B$ representing \cref{fig:op:bisimilarity_example_1} and $B'$ \cref{fig:op:bisimilarity_example_2} are bisimilar $\Gamma^\mathcal{B}_{(I,m)} \sim \Gamma^\mathcal{B'}_{(I,m')}$ for any database instance $I$, and any markings $m$ and $m'$ with $m(P_1) = m'(P_1)$ and $m(p) = \emptyset = m'(p)$ for all other $p$. \qedblack
\end{example}

We will need the following basic lemma to show the correctness of our optimizations, \ie, that the right and left hand sides of the respective optimization rules are bisimilar.

\begin{lemma}
    \label{thm:simulation-congruence}
    The relation $\sim$ is an congruence relation with respect to composition of timed db-nets with boundaries, \ie it is reflexive,
    symmetric and transitive, and if 
    if $\Gamma^{\mathcal{B}_1}_{s_0} \sim \Gamma^{\mathcal{B}_1'}_{s_0}$ and $\Gamma^{\mathcal{B}_2}_{s_0} \sim \Gamma^{\mathcal{B}_2'}_{s_0}$ for all $s_0$ on the shared boundary of $B_1$ and $B_2$, then $\Gamma^{\mathcal{B}_1 \circ \mathcal{B}_2}_{s_0} \sim \Gamma^{\mathcal{B}'_1 \circ \mathcal{B}'_2}_{s_0}$. \qedwhite
\end{lemma}

The following lemma means that it makes sense to ask the question if the optimized version of an IPCG is bisimilar to the original IPCG or not.

\begin{lemma}
	\label{thm:simulation-boundary}
	Let $G$ and $G'$ be IPCGs. 
	For each optimisation rewrite rule $G \Rightarrow G'$, $\sem{G}$ and $\sem{G'}$ have the same boundary.
	\qedwhite
\end{lemma}

We are now ready to state and prove our correctness theorem.

\begin{theorem}[Change Correctness]
  \label{thm:change-correctness}
	Let $G$ and $G'$ be IPCGs such that $G \Rightarrow G'$ is an optimization rule.
	For every initial snapshot $s_0$ of both $\sem{G}$ and $\sem{G'}$, with tokens in input places only, we have $\Gamma^{\sem{G}}_{s_0} \sim \Gamma^{\sem{G'}}_{s_0}$.
\end{theorem}

\begin{proof}
  We verify the statement for each optimization $G \Rightarrow G'$. 
  By~\cref{thm:simulation-congruence}, it is enough to show that the parts of the interpretation of the graphs which are actually modified by the rewrite are bisimilar.

\labelsubtitle{Redundant Sub-Processes (\cref{opt:redundant-subprocess})}
Each move on the left hand side of the optimization rule in~\cref{fig:redundant_control_flow_v2} (on~\cpageref{fig:redundant_control_flow_v2}) either moves tokens into a cloud, out of a cloud, or inside a cloud.
In the first two cases, this can be simulated by the right hand side by moving the token through the CE or CBR and CF respectively followed by a move into or out of the cloud, while in the latter case the corresponding token can be moved in $SG'_1$ up to the isomorphism between $SG'_1$ and the cloud on the left.

Similarly, a move on the right hand side into or out of the cloud can easily be simulated on the left hand side. Suppose a transition fires in $SG'_1$. Since all guards in $SG'_1$ have been modified to require all messages to come from the same enriched context, the corresponding transition can either be fired in $SG_1$ or $SG_2$.

\labelsubtitle{Combining Sibling Patterns (\cref{opt:combine-siblings})}
Suppose the left hand side of~\cref{fig:combine_sibling_nodes_v2_b} (on~\cpageref{fig:combine_sibling_nodes_v2_b}) takes a finite number of steps and ends up with $m(P_2)$ tokens in $P_2$ and $m(P_3)$ tokens in $P_3$.
There are three possibilities: (i) there are tokens of the same color in both $P_2$ and $P_3$; or (ii) there is a token in $P_2$ with no matching token in $P_3$; or (iii) there is a token in $P_3$ with no matching token in $P_2$.
For the first case, the right hand side can simulate the situation by emulating the steps of the token ending up in $P_2$, and forking it in the end.
For the second case, the right hand side can simulate the situation by emulating the steps of the token ending up in $P_2$, then forking it, but not moving one copy of the token across the boundary layer in the interpretation of the fork pattern.
The third case is similar, using that $SG_2$ is isomorphic to $SG_1$.

The right hand side can easily be simulated by copying all moves in $SG_1$ into simultaneous moves in $SG_1$ and the isomorphic $SG_2$.

\labelsubtitle{Early-Filter (\cref{opt:early-filter})}
By construction, the filter removes the data not used by $P_2$, so if the left hand side of~\cref{fig:early_filter_v2_insert} (on~\cpageref{fig:early_filter_v2_insert}) moves a token to $P_2$, then the same token can be moved to $P_2$ on the right hand side and vice versa.

\labelsubtitle{Early-Mapping (\cref{opt:early-mapping})}
Suppose the left hand side of~\cref{fig:early_mapping_v2} (on~\cpageref{fig:early_mapping_v2}) moves a token to $P_4$.
The same transitions can then move the corresponding token to $P_4$ on the right hand side, with the same payload, by construction.
Similarly, the right hand side can be simulated by the left hand side.

\labelsubtitle{Early-Aggregation (\cref{opt:early-aggregation})}
The interpretation of the subgraph $SG_2$ is equivalent to the interpretation of $P_1$ followed by $SG'_2$ followed by $P_3$, by construction in~\cref{fig:early_aggregate} (on~\cpageref{fig:early_aggregate}), hence the left hand side and the right hand side are equivalent.

\labelsubtitle{Early Claim Check (\cref{opt:early-claim-check})}
Since the claim check CC + CE in~\cref{fig:early_claim_v2} (on~\cpageref{fig:early_claim_v2}) simply stores the data and then adds it back to the message in the CE step, both sides can simulate each other.

\labelsubtitle{Early-Split (\cref{opt:early-split})}
By assumption, $P_1$ followed by $SSQ_1$ ($P_1$ followed by $SSQ_1$ followed by $P_2$ for the inserted early split in~\cref{fig:early_split_v2} (on~\cpageref{fig:early_split_v2})) is equivalent to $SSQ_1$ followed by $P_1$, from which the claim immediately follows.

\labelsubtitle{Sequence to Parallel (\cref{opt:seq-to-parallel}), Merge Parallel (\cref{opt:merge-parallel})}
The left hand side of~\cref{fig:sequence_to_parallel_v2} (on~\cpageref{fig:sequence_to_parallel_v2}) can be simulated by the right hand side by copying each move in $SSQ_1$ by a move each in $SSQ'_1$ to $SSQ'_n$.
If the right hand side moves a token to an output place, it must move a token through some $SSQ'_i$, and the same moves can move a token through $SSQ_1$ in the left hand side.

The same reasoning applies to the Merge Parallel transformation in \cref{fig:op:merge_parallel_v2}, but in reverse.

\labelsubtitle{Heterogeneous Sequence to Parallel (\cref{opt:hetero-parallel})}
By assumption, the sub-sequences $SSQ_1$ to $SSQ_n$ are side-effect free.
The right hand side of~\cref{fig:heterogeneous_sequence_to_parallel_v2} (on~\cpageref{fig:heterogeneous_sequence_to_parallel_v2}) can simulate the left hand side as follows: if the left hand side moves a token to an output place, it must move it through all of $SSQ_1$ to $SSQ_n$.
The right hand side can make the same moves in the same order.
For the other direction, the left hand side can reorder the moves of the right hand side to first do all moves in $SSQ_1$, then in $SSQ_2$ and so on.
This is still a valid sequence of steps because the subseqences can be parallelized.

\labelsubtitle{Ignore, try failing endpoints (\cref{opt:ignore-failing-endpoints})}
Suppose the left hand side of~\cref{fig:ignore_failing_endpoints_v2} takes a finite amount of steps to move a token to an output place in $SG_1$, however, the transition to $SG_{ext}$ does not produce a result due to an exceptional situation (\ie no change of the marking in $cf$).
Correspondingly, the right hand side moves the token, however, without the failing, and thus read-only transition to $SG_{ext}$, which ensures the equality of the resulting tokens on either side.
Under the same restriction that the no exception context is returned from $SG_{ext}$, the right hand side can simulate the left hand side accordingly.

The situation for try failing endpoints in \cref{fig:ignore_failing_endpoints_v2_b} is the same in reverse.

\labelsubtitle{Reduce requests (\cref{opt:reduce-requests})}
Since the only difference between the left hand side and the right hand side is the slow-down due to the insertion of the pattern $CS$, and simulation does not take the age of messages into account, the left hand side can obviously simulate the right hand side and vice versa.
\qedwhite

\end{proof}

\subsection{Discussion}
\Cref{thm:change-correctness} makes precise in which sense the proposed optimization rules are correct, in the sense that they preserve the \enquote{meaning} of integration patterns as defined in our translation to timed db-nets.
We emphasise that this is a property proven by analysing the proposed optimization rules, and not a fact that we expect to be decidable for arbitrary rewrite rules.
By interpreting integration patterns as timed db-nets, we however get access to a framework where it makes sense to prove individual optimizations correct in the above sense, since there is a formal definition of the behaviour of (the interpretation of) the integration pattern.


%% file: evaluation.tex
\section{Evaluation}
\label{sec:evaluation}
In this section, (a) we evaluate the impact of optimization strategies (\ie OS-1--3) from \cref{sec:stategies} that are most relevant to this work, and (b) we study ReCO for real-world integration processes.

\input{optimizations_quantitative_and_case_study}

\input{timed_db_net_w_boundaries_case_study}

\subsection{Discussion}
The evaluation on the optimization strategies on IPCGs (cf. (a)) resulted into several interesting conclusions, \ie emphasizing on the importance of a pattern composition and optimization formalization, which are relevant even for experienced integration experts (conclusions 1–2), with interesting choices (conclusions
3–4, 6), implementation details (conclusions 5, 10) and trade-offs (conclusions 7–9).
The contract graphs provide a rich composition context, which might help the user when composing patterns with built-in structural correctness guarantees.

The second major aspect of this work --- besides process optimizations --- concerns the responsible composition of processes out of integration patterns (cf. (b)) that can be automated (cf. conclusion 10).
The evaluation of two case studies --- following ReCO --- resulted into further interesting conclusions, \ie the suitability of our approach for pattern compositions (cf. conclusions (11,13)), model complexity considerations (cf. conclusions (12,14)) and desirable extensions like automatic translation (cf. conclusion (15)).
However, while IPCGs based on \tdbnets with boundaries denote the first comprehensive definition of application integration scenarios with built-in functional correctness and compositional correctness validation and verification, it might not give an appealing modeling language for (non-technical) users (cf. conclusions (12,14)).
We envision a novel modeling language and tool support that facilitates a translation from that language to IPCGs (cf. conclusion (15)), which we consider as future work.
Based on such a language infrastructure, more advanced compositional aspects like modeling guidelines on the different layers (\ie language, intermediate IPCG, and simulation \tdbnet with boundaries) could be studied.

%% file: optimizations_quantitative_and_case_study.tex
\subsection{Optimization Strategies}
For (a) we quantitatively analyze the effect of optimizations on two catalogs of integration processes regarding improvements of model complexity, throughput and latency.
The catalogs have a two year difference to be able to study whether improvements were found by integration experts within that time span by themselves.

Then, we revisit our motivating example from \cref{sec:ReCO} and study a more complex integration process regarding applicable optimization strategies.

\subsubsection{Quantitative Analysis}
\label{sub:analysis}
We applied the optimization strategies OS-1--3 to $627$ integration scenarios from the 2017 standard content of the SAP CPI (called ds17), and compared with $275$ scenarios from 2015 (called ds15).
Our goal is to show the applicability of our approach to real-world integration scenarios, as well as the scope and trade-offs of the optimization strategies.
The comparison with a previous content version features a practical study on content evolution.
To analyze the difference between different scenario domains, we grouped the scenarios into the following categories~\cite{Ritter201736}:
On-Premise to Cloud (OP2C),
Cloud to Cloud (C2C), and
Business to Business (B2B).
Since hybrid integration scenarios such as OP2C target the extension or synchronization of business data objects, they are usually less complex.
In contrast native cloud application scenarios such as C2C or B2B mediate between several endpoints, and thus involve more complex integration logic~\cite{Ritter201736}.
The process catalog also contained a small number of simple Device to Cloud scenarios; none of them could be improved by our approach.

\labeltitle{Setup: Construction and analysis of IPCGs} For the analysis, we constructed an IPCG for each integration scenario following the workflow sketched in \cref{fig:evaluation_pipeline}.
\begin{figure}[bt]
	\centering
	\includegraphics[width=1\columnwidth]{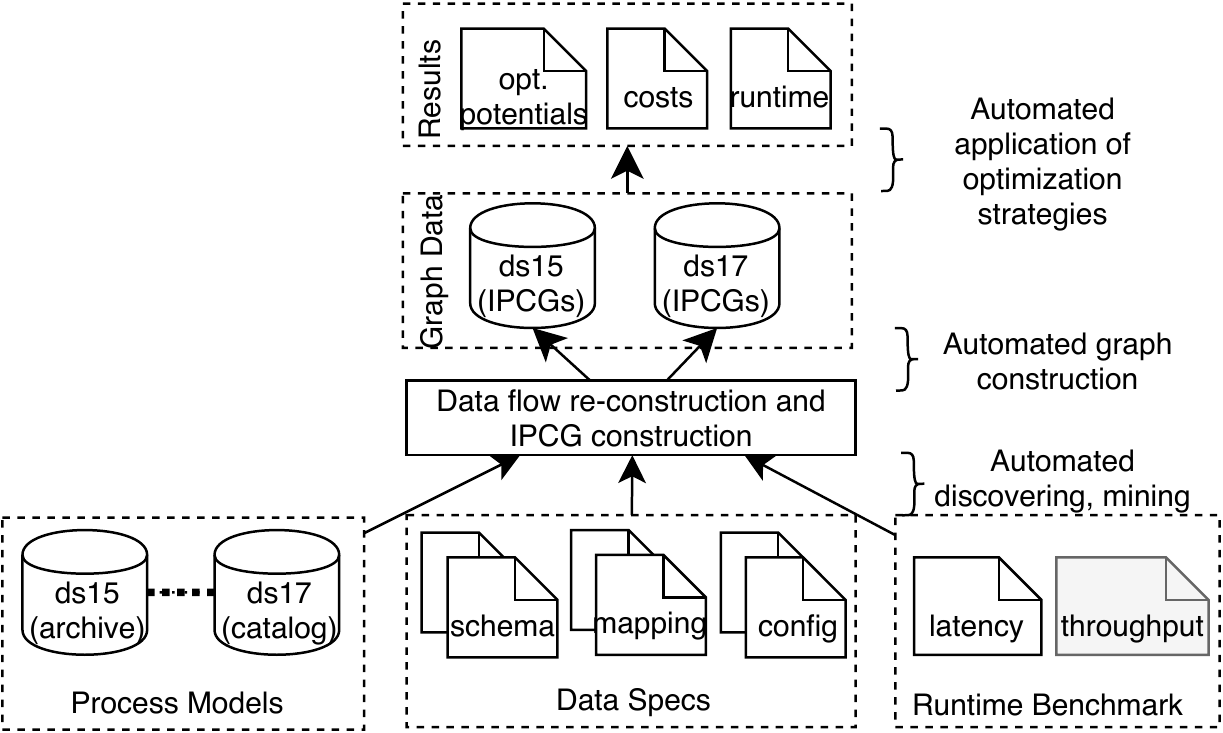}
	\vspace{-0.3cm}
	\caption{Pattern composition evaluation pipeline.}
	\label{fig:evaluation_pipeline}
\end{figure}
Notably, the integration scenarios are stored as process models in a BPMN-like notation~\cite{ritter2016exception}. 
The process models reference data specifications such as schemas (\eg XSD, WSDL), mapping programs, selectors (\eg XPath) and configuration files.
For every pattern used in the process models, runtime statistics are available from benchmarks~\cite{ritter2016benchmarking}.
The data specifications are picked up from the 2015 content archive and from the current 2017 content catalog, while the runtime benchmarks are collected using the open-source integration system \emph{Apache Camel}~\cite{Ibsen:2010:CA:1965487}\footnote{All measurements were conducted on a HP Z600 workstation, equipped with two Intel X5650 processors clocked at 2.67GHz with a $12$ cores, 24GB of main memory, running a 64-bit Windows 7 SP1 and a JDK version $1.7.0$, with $2$GB heap space.} as used in SAP CPI.
The mapping and schema information is automatically mined and added to the patterns as contracts, and the rest of the collected data as pattern characteristics.
For each integration scenario and each optimization strategy, we determine if the strategy applies, and if so, if the cost is improved.
This analysis runs in about two minutes in total for all 902 scenarios on our workstation.

We now discuss the improvements for the different kinds of optimization strategies identified in \cref{sec:stategies}.

\labeltitle{Improved Model Complexity: Process Simplification (OS-1).}
The relevant metric for the process simplification strategies from OS-1 is the average reduction in model complexity, 
shown in \cref{fig:model_complexity}.

\labelsubtitle{Results}
Although all scenarios were implemented by integration experts, who are familiar with the modeling notation and the underlying runtime semantics, there is still a small amount of patterns per scenario that could be removed without changing the execution semantics.
On average, the content reduction for the content from 2015 and 2017 was $1.47$ and $2.72$ patterns/IPCG, respectively, with significantly higher numbers in the OP2C domain.

\begin{figure}[bt]
	\centering
	\includegraphics[width=0.9\columnwidth]{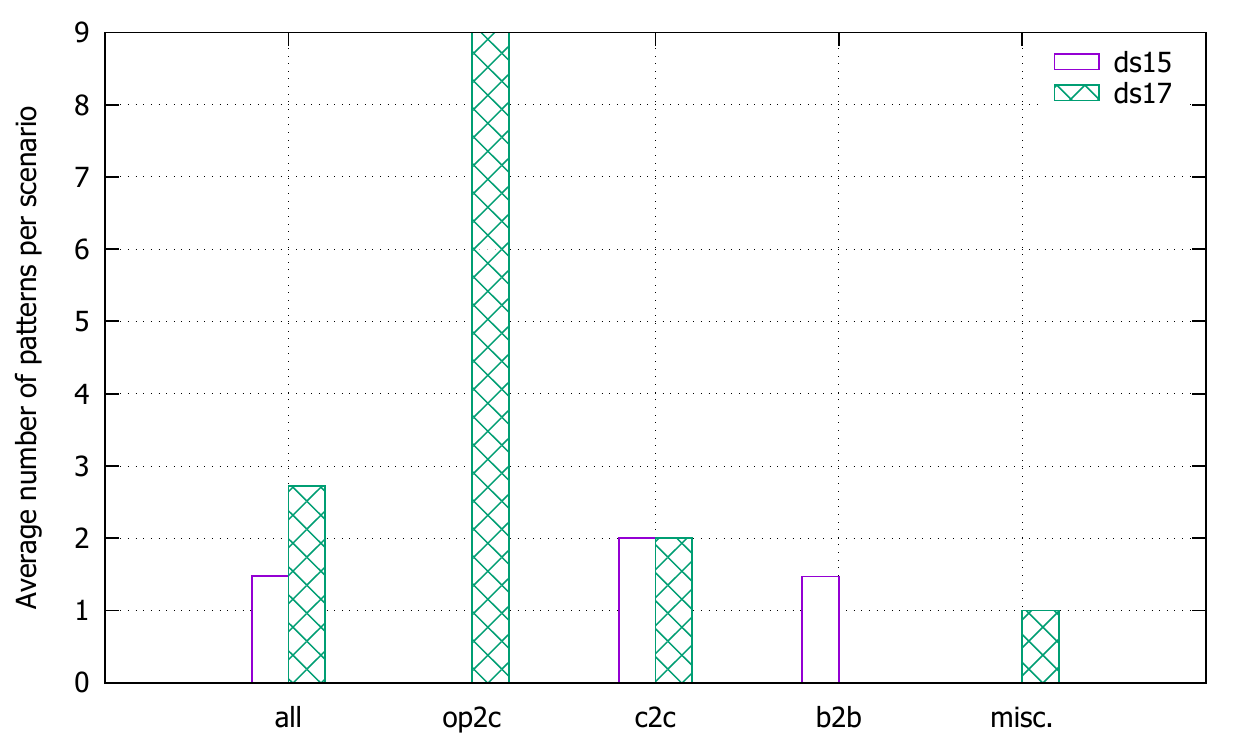}
	\vspace{-0.3cm}
	\caption{Pattern reduction per scenario.}
	\label{fig:model_complexity}
\end{figure}

\begin{figure*}[bt]
	\begin{center}$
		\begin{array}{cc}
		\subfigure[Used vs. unused data elements]{\label{fig:fields}\includegraphics[width=0.5\linewidth]{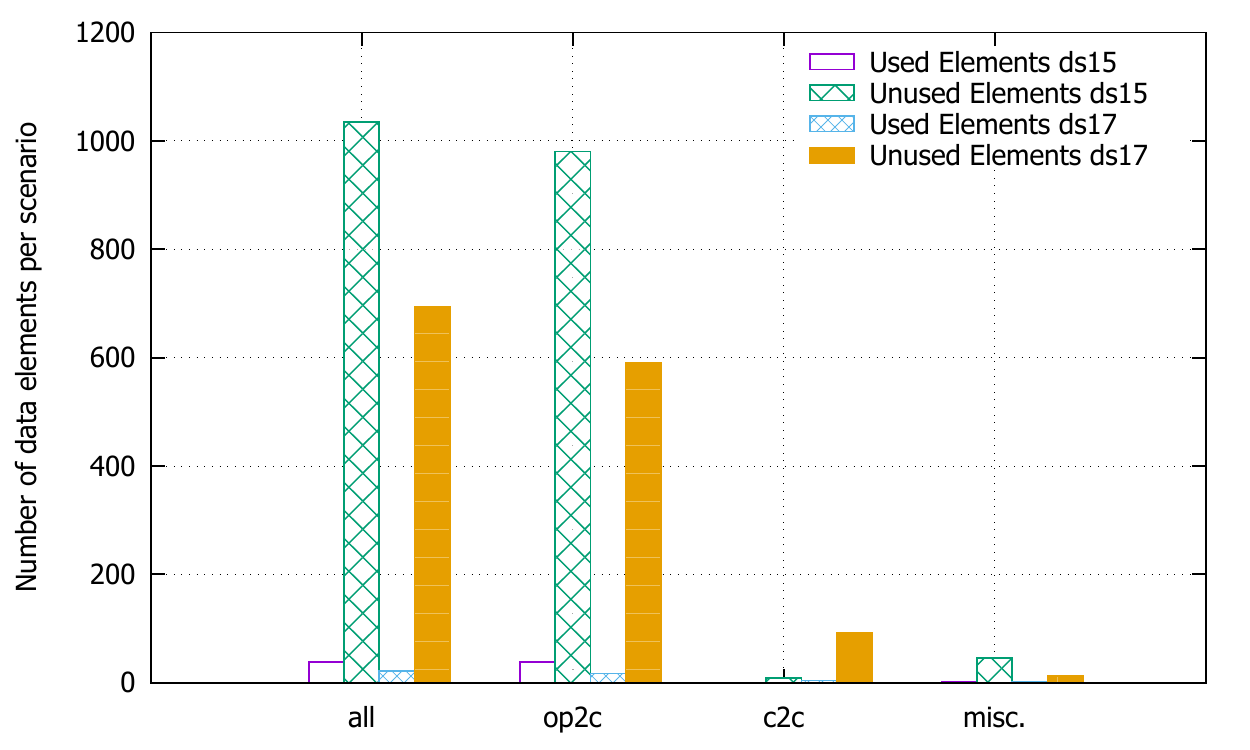}}
		\subfigure[Savings in abstract costs on unused data elements]{\label{fig:field_savings}\includegraphics[width=0.5\linewidth]{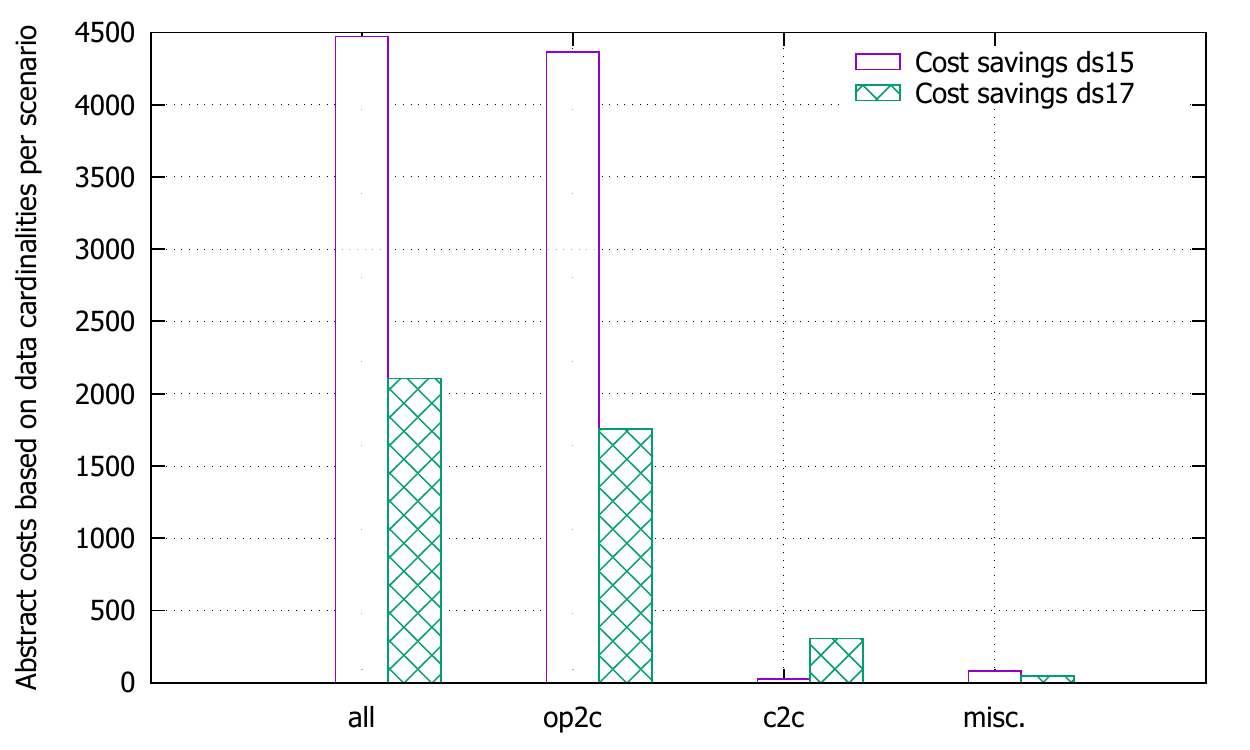}}
		\end{array}$
	\end{center}
	\vspace{-0.3cm}
	\caption{Unused elements in integration scenarios.}
	\label{fig:}
\end{figure*}

\labelsubtitle{Conclusions} (1) Even simple process simplifications are not always obvious to integration experts in scenarios represented in a control-flow-centric notation (\eg current SAP CPI does not use BPMN Data Objects to visualize the data flow); and (2) the need for process simplification does not seem to diminish as experts gain more experience.

\labeltitle{Improved Bandwidth: Data Reduction (OS-2).} Data reduction impacts the overall bandwidth and message throughput~\cite{DBLP:conf/debs/RitterDMR17}.
To evaluate data reduction strategies from OS-2, we leverage the data element information attached to the IPCG contracts and characteristics, and follow their usages along edges in the graph, similar to \enquote{ray tracing} algorithms~\cite{glassner1989introduction}.
We collect the data elements that are used or not used, where possible --- we do not have sufficient design time data to do this for user defined functions or some of the message construction patterns, such as request-reply.
Based on the resulting data element usages, we calculate two metrics: the comparison of used vs. unused elements in \cref{fig:fields}, and the savings in abstract costs on unused data elements in \cref{fig:field_savings}.

\labelsubtitle{Results}
There is a large amount of unused data elements per scenario for the OP2C scenarios; these are mainly web service communication and message mappings, for which most of the data flow can be reconstructed.
This is because the predominantly used EDI and SOA interfaces (\eg SAP IDOC, SOAP) for interoperable communication with on-premise applications define a large set of data structures and elements, which are not required by the cloud applications, and vice versa.
In contrast, C2C scenarios are usually more complex, and mostly use user defined functions to transform data, which means that only a limited analysis of the data element usage is possible.

When calculating the abstract costs for the scenarios with unused fields, there is an immense cost reduction potential for the OP2C scenarios as shown in \cref{fig:field_savings}.
This is achieved  by adding a content filter to the beginning of the scenario, which removes unused fields.
This results in a cost increase $|d_{in}|= \#$unused elements for the content filter, but reduces the cost of each subsequent pattern up to the point were the elements are used.

\labelsubtitle{Conclusions} (3) Data flows can best be reconstructed when design time data based on interoperability standards is available; and (4) a high number of unused data elements per scenario indicates where bandwidth reductions are possible.

\labeltitle{Improved Latency: Parallelization (OS-3).}  For the sequence-to-parallel optimization strategies from OS-3, the relevant metric is the processing latency of the integration scenario.
Because of the uncertainty in determining whether a parallelization optimization would be beneficial, we first report on the classification of parallelization candidates in \cref{fig:op1_genereal}.
We then report both the improvements according to our cost model in \cref{fig:op1_genereal_costs}, as well as the actual measured latency in \cref{fig:op1_scenario_costs}.

\begin{figure*}[bt]
	\begin{center}$
		\begin{array}{cc}
		\subfigure[Parallelization scenario candidates]{\label{fig:op1_genereal}\includegraphics[width=0.33\linewidth]{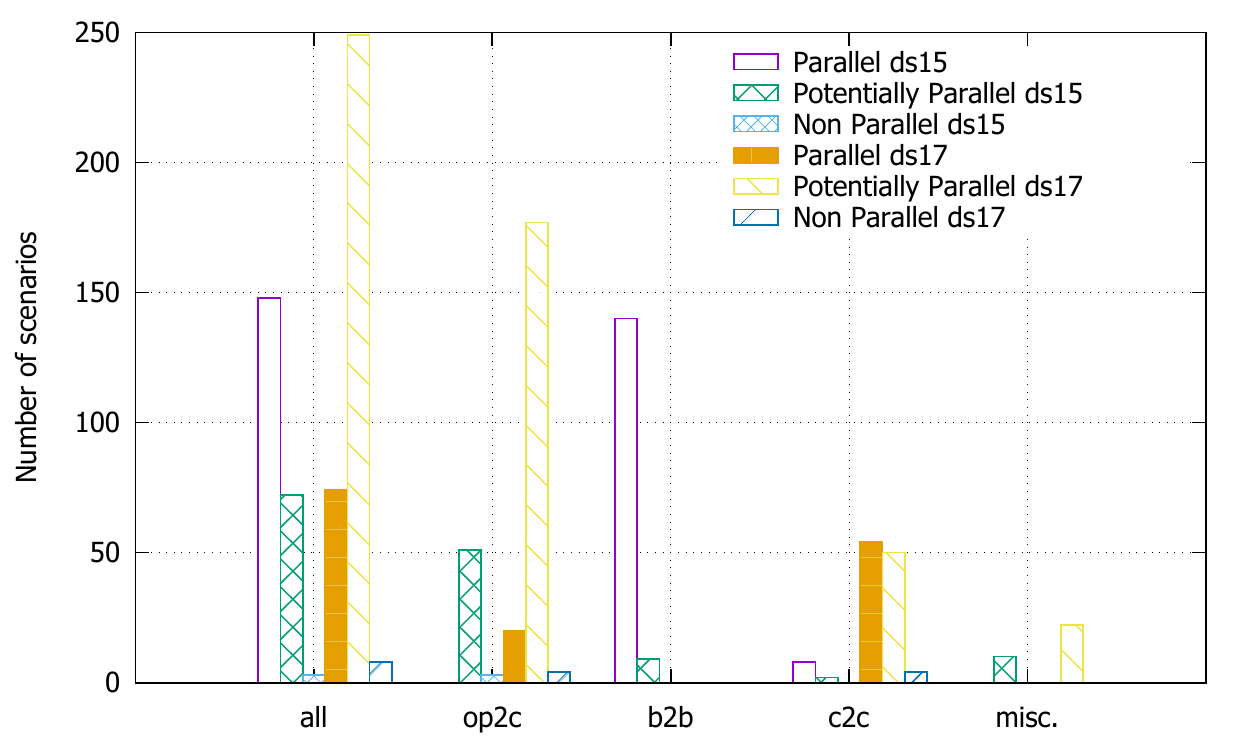}}
		\subfigure[Optimized scenarios based on costs]{\label{fig:op1_genereal_costs}\includegraphics[width=0.33\linewidth]{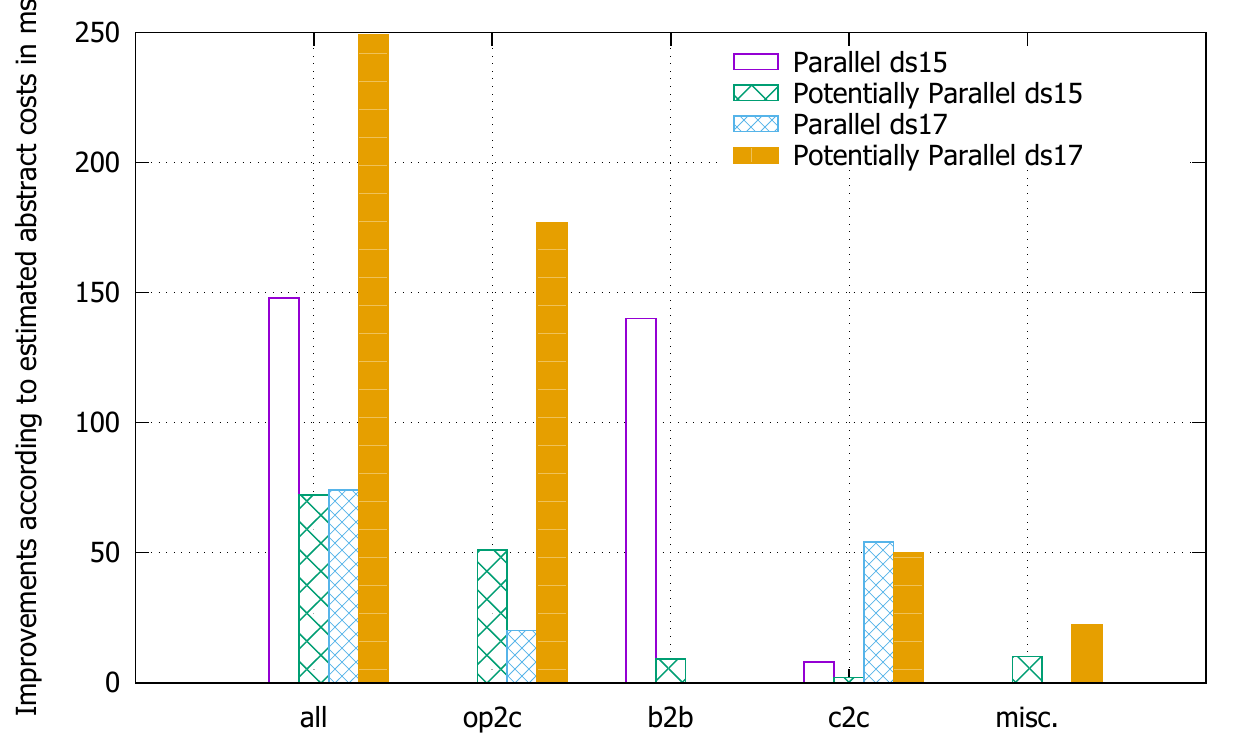}}
		\subfigure[Actual latency improvements]{\label{fig:op1_scenario_costs}\includegraphics[width=0.33\linewidth]{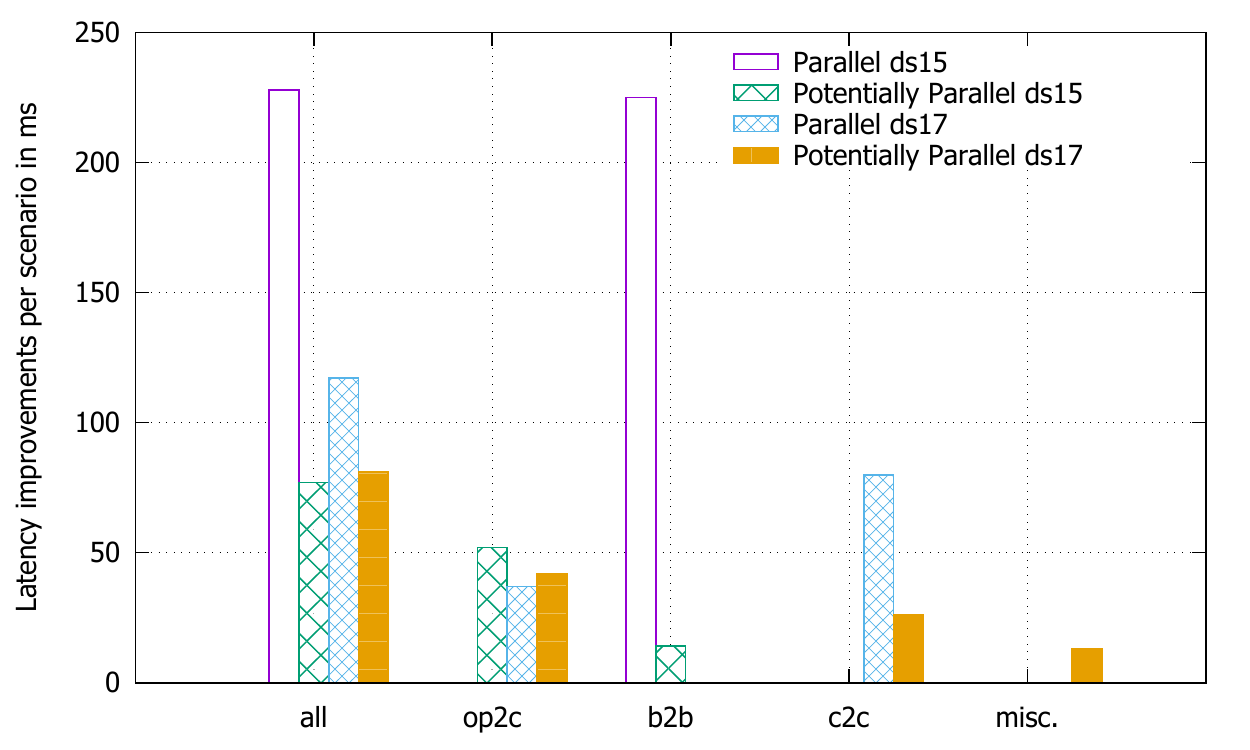}}
		\end{array}$
	\end{center}
	\vspace{-0.3cm}
	\caption{OS-3 \enquote{Sequence to parallel} optimization candidates on (a) integration flows, (b) optimization selection based on abstract cost model, and (3) actual latency improvements.}
	\label{fig:parallel_stages}
\end{figure*}

\labelsubtitle{Results}
Based on the data element level, we classify scenario candidates as \texttt{parallel}, definitely \texttt{non parallel}, or \texttt{potentially parallel} in \cref{fig:op1_genereal}.
The uncertainty is due to sparse information.
From the 2015 catalog, 81\% of the scenarios are classed as \texttt{parallel}, or \texttt{potentially parallel}, while the number for the 2017 catalog is 53\%.
In both cases, the OP2C and B2B scenarios show the most improvement potential.
\Cref{fig:op1_genereal_costs} shows the selection based on our cost model, which supports the pre-selection of all of these optimization candidates.
The actual, average improvements per impacted scenario are shown in~\cref{fig:op1_scenario_costs}.
The average improvements of up to $230$ milliseconds per scenario must be understood in the context of the average runtime per scenario, which is 1.79 seconds.
We make two observations: (a) the cost of the additional fork and join constructs in Java are high compared to those implemented in hardware~\cite{DBLP:conf/debs/RitterDMR17}, and the improvements could thus be even better, and (b) the length of the parallelized pattern sequence is usually short: on average 2.3 patterns in our scenario catalog.

\labelsubtitle{Conclusions} (5) The parallelization requires low cost fork and join implementations; and (6) better runtime improvements might be achieved for scenarios with longer parallelizable pattern sequences.

\subsubsection{Case Studies}
\label{sub:cases}

We apply, analyze and discuss the proposed optimization strategies in the context of two case studies: the Replicate Material on-premise to cloud scenario from \cref{fig:pn:c4c} in \cref{sec:ReCO}, as well as an SAP eDocument invoicing cloud to cloud scenario.
These scenarios are part of the SAP CPI standard, and thus several users (\ie SAP's customers) benefit immediately from improvements.
For instance, we additionally implemented a content monitor pattern~\cite{Ritter201736} that allowed analysis of the SAP CPI content.
This showed the Material Replicate scenario was used by $546$ distinct customers in $710$ integration processes copied from the standard into their workspace --- each one of these users is affected by the improvement.

\labeltitle{Replicate Material (revisited)}
Recall from \cref{sec:ReCO} that the Replicate Material scenario is concerned with enriching and translating messages coming from a CRM before passing them on to a Cloud for Customer service, as in \cref{fig:pn:c4c}.
As already discussed, the content enricher and the message translator can be parallelized according to the sequence to parallel optimization from OS-3.
The original and resulting IPCGs are shown in \cref{fig:cpn_pattern_composition_summarized_v2a,fig:cpn_pattern_composition_summarized_v2b}.
No throughput optimizations apply.

\labelsubtitle{Latency improvements}
The application of this optimization can be considered, if the latency of the resulting parallelized process is smaller than the latency of the original process, i.e.\ if
\begin{align*}
  \cost(MC) & {} + \max(\cost(CE),\cost(MT)) \\ & {} + \cost(JR) + \cost(AGG) \\ &\qquad\qquad < \cost(CE)+\cost(MT)
\end{align*}
Subtracting $\max(\cost(CE),\cost(MT))$ from both sides of the inequality, we are left with
\begin{align*}
  \cost(MC) + {} &  \cost(JR) + \cost(AGG) \\ & \qquad\qquad < \min(\cost(CE),\cost(MT))
\end{align*}
If we assume that the content enricher does not need to make an external call, its abstract cost becomes
\[
\cost(CE)(|d_{in}|, |d_r|) = |d_{in}|,
\]
and plugging in experimental values from a pattern benchmark~\cite{ritter2016benchmarking}, we arrive at the inequality (with latency costs in seconds)
\[
0.01 + 0.002 + 0.005 \not< \min(0.005,0.27)
\]
which tells us that the optimization is not beneficial in this case --- the additional overhead is larger than the saving.
However, if the content enricher does use a remote call, $\cost(CE)(|d_{in}|, |d_r|) = |d_{in}| + |d_r|$, and the experimental values now say $\cost(CE) = 0.021$.
Hence the optimization is worthwhile, as
\[
  0.01 + 0.002 + 0.005 < \min(0.021,0.27) \enspace .
\]

\labelsubtitle{Model Complexity}
Following S\'anchez-Gonz\'alez et al.~\cite{sanchez2010prediction}, we measure the model complexity by node count.
In this case, the optimization increases the complexity by 3.

\labelsubtitle{Conclusions}
(7) Pattern characteristics are important when deciding if an optimization should be applied (\eg local vs.\ remote enrichment); and (8) there are conflicts between different objectives, as illustrated by the trade-off between latency reduction and model complexity increase.

\labeltitle{eDocuments: Italy Invoicing.}
The Italian government accepts electronic invoices from companies,
as long as they follow regulations --- they have to be correctly formatted, signed, and not be sent in  duplicate.
Furthermore, these regulations are subject to change.
This can lead to an ad-hoc integration process such as in \cref{fig:invoice_italy} (simplified).
Briefly, the companies' \emph{Fattura Electronica} is used to generate a \emph{factorapa} document with special header fields (\eg \emph{Paese}, \emph{IdCodice}), then the message is signed and sent to the authorities, if it has not been sent previously.
The multiple authorities respond with standard \emph{Coglienza}, \emph{Risposta} acknowledgments, that are transformed to a \emph{SendInvoiceResponse}.
\begin{figure*}[tbh]
    \centering
    \includegraphics[width=1\linewidth]{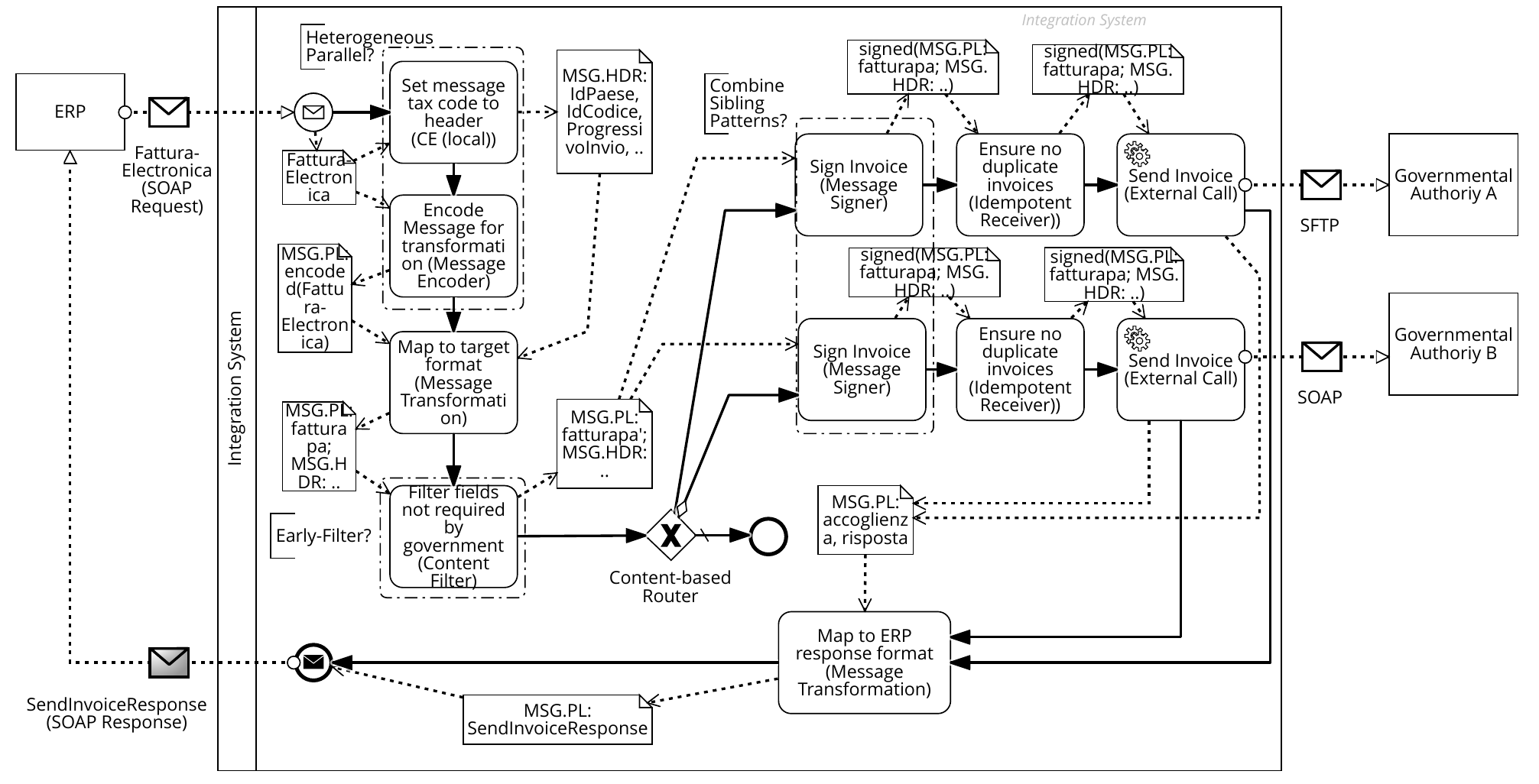}
    \vspace{-.3cm}
    \caption{Country-specific invoicing (potential improvements as BPMN Group)}
    \label{fig:invoice_italy}
\end{figure*}
We transformed the BPMN model to an IPCG, tried to apply optimizations, and created a BPMN model again from the optimized IPCG.

\labelsubtitle{Model Complexity} Our heuristics for deciding in which order to try to apply different strategies are \enquote{simplification before parallelization} and \enquote{structure before data}, since this seems to enable the largest number of optimizations.
Hence we first try to apply OS-1 strategies: the \emph{combine siblings} rule matches the sibling Message Signers, since the preceding content-based router is a fork. (The signer is also side-effect free, so applying this rule will not lead to observably different behavior.)
\begin{figure*}[tbh]
    \centering
    \includegraphics[width=1\linewidth]{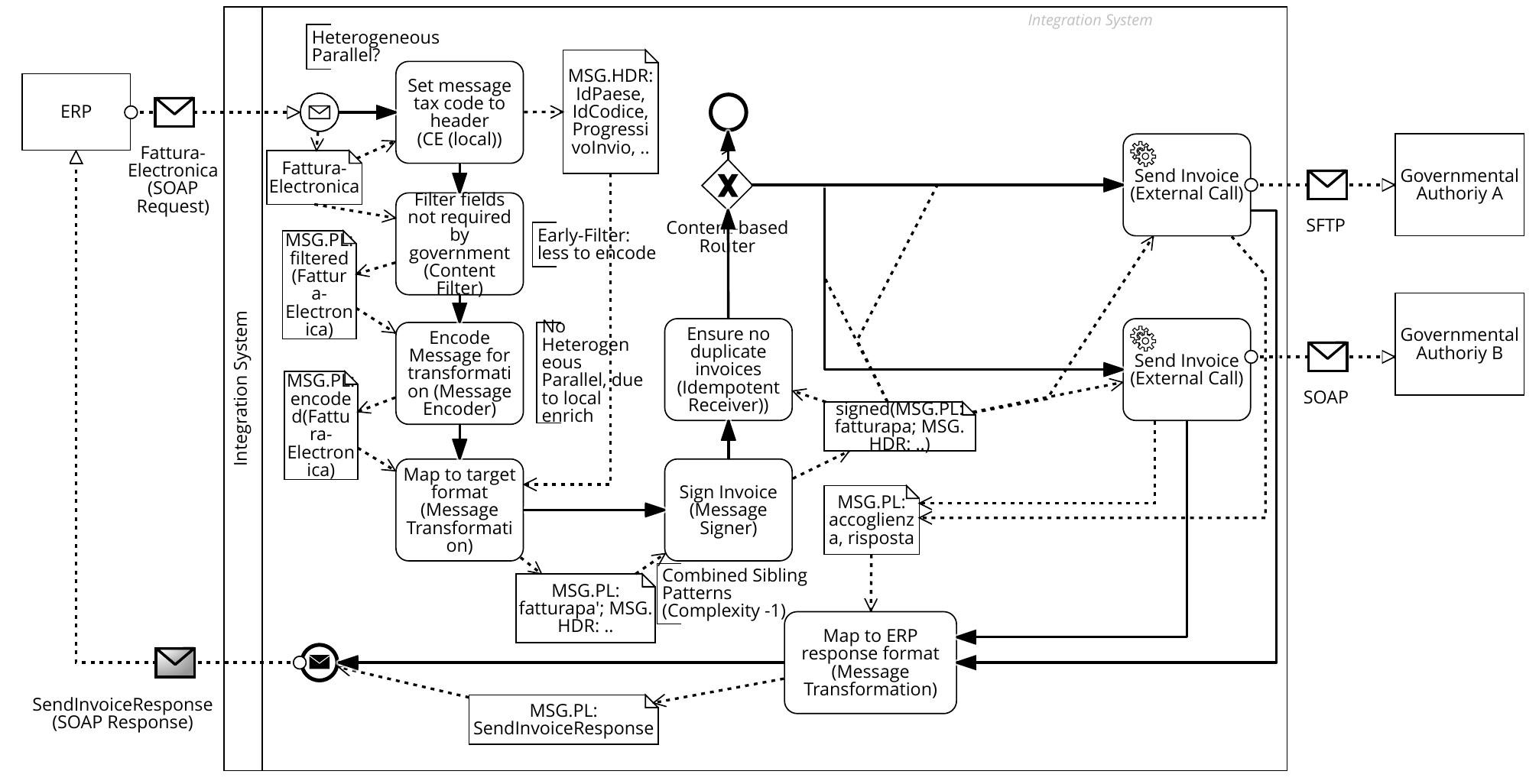}
    \vspace{-.3cm}
    \caption{Invoice processing from \cref{fig:invoice_italy} after application of strategies OS-1--3.}
    \label{fig:invoice_simplified_parallel_data}
\end{figure*}
\labelsubtitle{Latency Improvements} Next we try OS-3 strategies.
Although \emph{heterogeneous parallelization} matches for the CE and the Message Encoder, it is not applied since
\begin{align*}
  \cost(MC) + \cost(JR) + \cost(AGG) \\
    \qquad\qquad\qquad \not< \min(\cost(CE), \cost(ME)),
\end{align*}
\ie the overhead is too high, due to the low-latency, local CE.
Finally, the \emph{early-filter} strategy from OS-2 is applied for the Content Filter, inserting it between the Content Enricher and the Message Encoder.
No further strategies can be applied.
The resulting integration process translated back from IPCTG to BPMN is shown in \cref{fig:invoice_simplified_parallel_data}.

\labelsubtitle{Conclusions} (9) The application order OS-1, OS-3, OS-2 seems most beneficial (\enquote{simplification before parallelization}, \enquote{structure before data});
(10) an automatic translation from IPCGs to concepts like BPMN could be beneficial for connecting with existing solutions.

%% file: timed_db_net_w_boundaries_case_study.tex
\subsection{Case Studies: Responsible Pattern Composition}
\label{sub:pc:evaluation}
For (b), we evaluate the translation in two case studies of real-world integration scenarios: the replicate material scenario from~\cref{fig:pn:c4c}, and a predictive machine maintenance scenario.
The former is an example of hybrid integration, and the latter of \iotshort~device integration.

For each of the scenarios, we give an integration pattern contract graph with matching contracts, 
translate it to a \tdbnet with boundaries, and show how its execution can be simulated. 
The scenarios are both taken from the SAP Cloud Platform Integration solution catalog of reference integration scenarios, and are frequently used by customers~\cite{sap-hci-content}.
For the simulation we use the CPN Tools \tdbnet prototype from~\cref{sub:testing} with the extension for hierarchical PN composition.
In CPN Tools hierarchies, the patterns can be represented as sub-groups and pages with explicit in- and out-port type definitions~\cite{jensen2007coloured}, which we use as part of the boundaries defined in~\cref{sec:semantics}.
Thereby the synchronization is checked based on the CPN color sets of the port types.
The other boundary checks are performed during the simulation according to the constructed boundaries (see construction mechanism in~\cref{def:pc:construction} in \cref{sub:interpreting_ipcgs}). 

\subsubsection{Hybrid Integration: Replicate Material}
An IPCG representing an integration process for the replication of material from an enterprise resource planning or customer relationship management system to a cloud system was given in \cref{fig:cpn_pattern_composition_summarized_v2a} in \cref{sub:ipcgs}.
We now add slightly more data in the form of the pattern characteristics, which provides sufficient information for the translation to \tdbnets with boundaries.
\Cref{fig:pc:cpn_pattern_composition_summarized_v3a} depicts the enriched IPCG.
\begin{figure*}[bt]
	\centering
	\includegraphics[width=.7\linewidth]{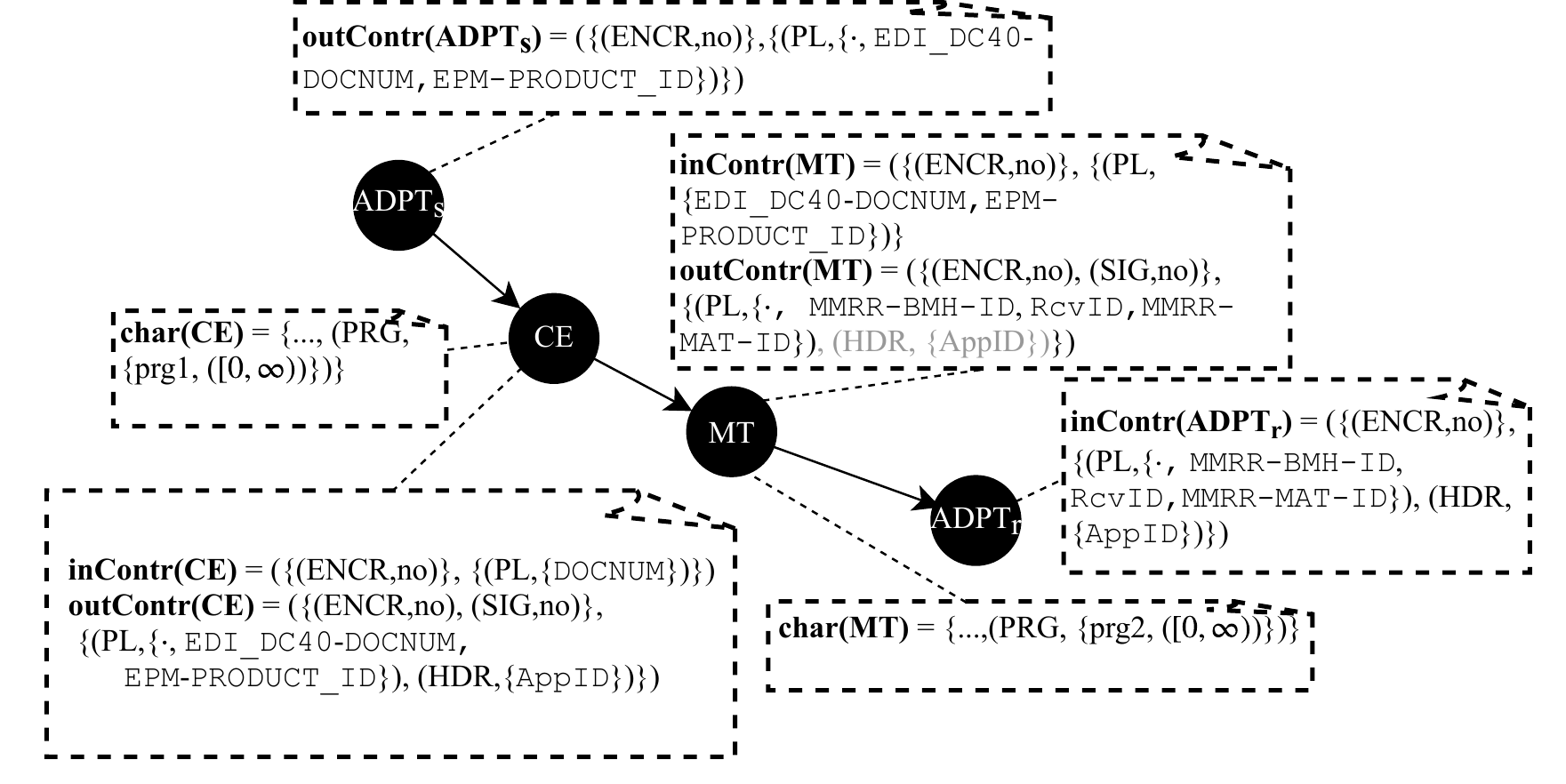}
	\caption{Complete integration pattern contract graph of the replicate material scenario}
	\label{fig:pc:cpn_pattern_composition_summarized_v3a}
\end{figure*}
The adapters are actually message processors, however, for simplicity they are represented as start and end pattern types, $ADPT_s$ denoting \emph{erp} and $ADPT_r$ representing \emph{cod}.
The characteristics of the $CE$ node includes the tuple $(PRG, (prg1, [0, \infty)))$, with enrichment function $prg1$ which assigns the \emph{DOCNUM} payload to the new header field \emph{AppID}.
Similarly, the characteristics of the $MT$ nodes includes a tuple $(PRG, (prg2, \_))$ with mapping program $prg2$,
which maps the \emph{EDI\_DC40-DOCNUM} payload to the \emph{MMRR-BMH-ID} field (the Basic Message Header ID of the Material Mass Replication Request structure), and the \emph{EPM-PRODUCT\_ID} payload to the \emph{MMRR-MAT-ID} field (the Material ID of the Material Mass Replication Request structure).

\paragraph{Translation to a Timed DB-Nets\index{Timed db-net} with Boundaries}
First we translate each single pattern from \cref{fig:pc:cpn_pattern_composition_summarized_v3a} according to the construction in \cref{sec:interpretation-atomic}.
The integrati on adapter nodes $ADPT_s$ and $ADPT_r$ are translated as the start and end patterns in \cref{fig:start} and \cref{fig:end}, respectively.
The content enricher $CE$ node and message translator $MT$ node are message processors without storage, and hence translated as in \cref{fig:siso_w_storage} with $<f>_{CE} = prg1$ and $<f>_{MT} = prg2$ (no database values are required).
Since no database table updates are needed for either translation, the database update function parameter $<g>$ can be chosen to be the identity function in both cases.

In the second step, we refine the \tdbnet with boundaries to also take contract concepts into account by the construction in \cref{def:pc:construction}.
The resulting net is shown in~\cref{fig:pc:material_replicate_w_boundaries}.
\begin{figure*}[bt]
	\centering
	\includegraphics[width=1.0\linewidth]{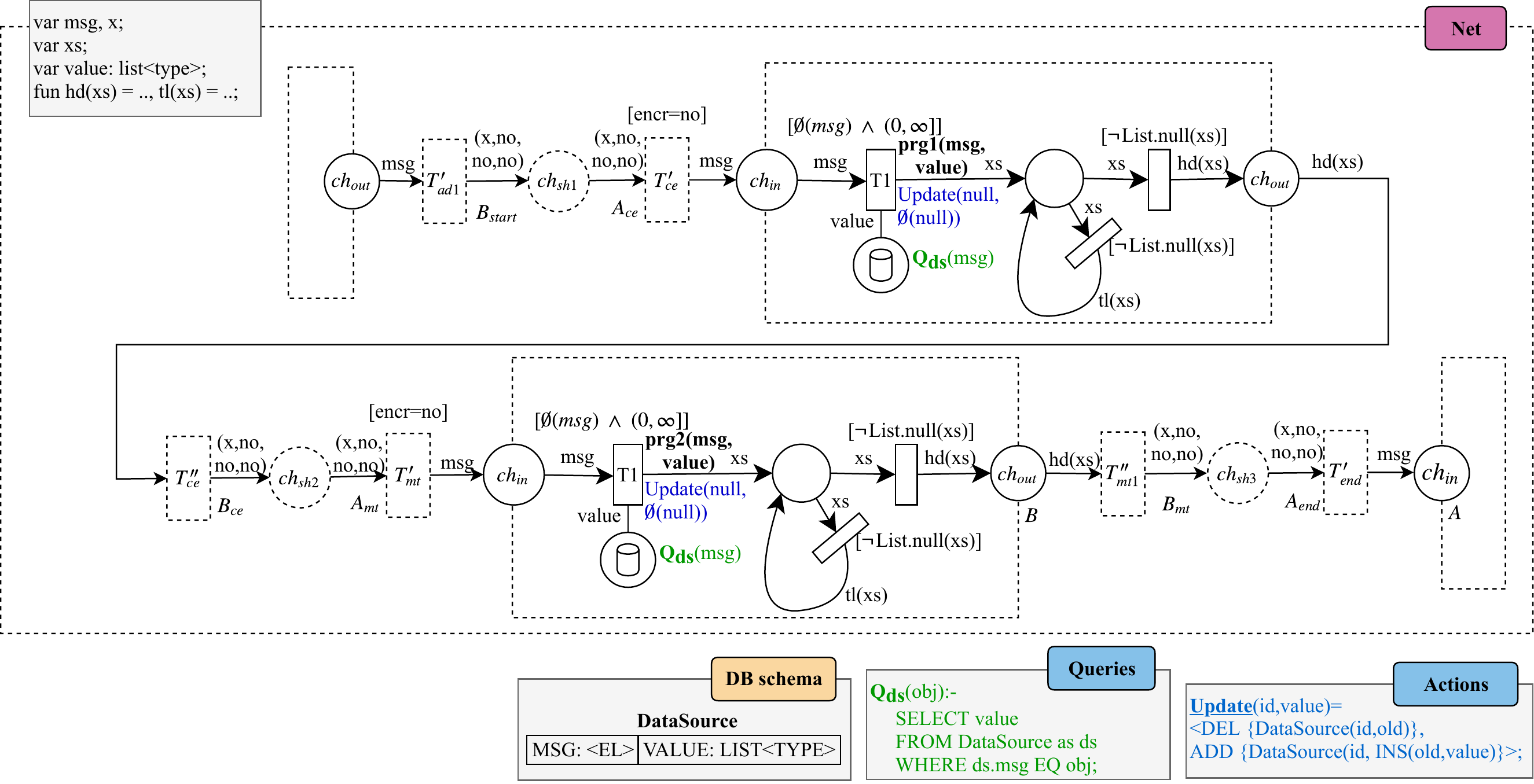}
	\caption{Material replicate scenario as a \tdbnets with boundaries}
	\label{fig:pc:material_replicate_w_boundaries}
\end{figure*}
This ensures the correctness of the types of data exchanged between patterns, and follows directly from the correctness of the corresponding IPCG.
Other contract properties such as encryption \emph{encr}, encodings \emph{enc}, and signatures \emph{sign} are checked through transition guards.

\paragraph{Simulation} We test the composition construction of the material replicate scenario in~\cref{fig:pc:material_replicate_w_boundaries} through simulation in the form of a hierarchical \tdbnet model, shown in~\cref{fig:pc:material_replicate_w_boundaries_simulation}.
\begin{figure*}[bt]
	\centering
	\includegraphics[width=1.0\linewidth]{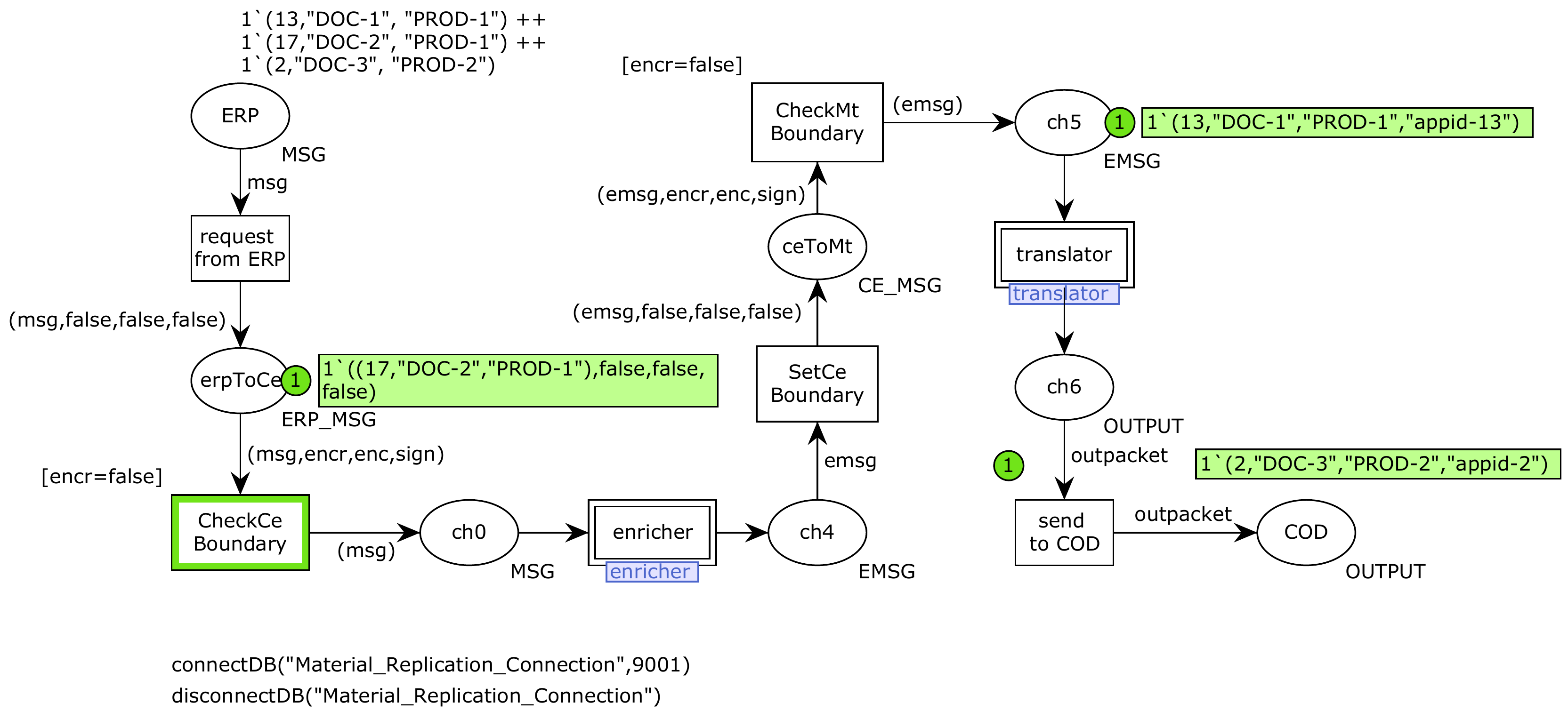}
    \vspace{-0.7cm}
	\caption{Material replicate scenario simulation}
	\label{fig:pc:material_replicate_w_boundaries_simulation}
\end{figure*}
Thereby, the \Ce and \mt patterns are represented by CPN Tool Subpage elements that are annotated with subpage tags \emph{enricher}, \emph{translator}, respectively.

On arrival of the request \emph{msg} from the ERP system, the boundary configuration is appended to the message in place \emph{erpToCe}.
In the replicate material scenario the data is received unencrypted, uncoded and unsigned, leading to a boundary (msg,no,no,no) (which is encoded as \emph{(msg,false,false,false)} in our prototype).
The extended message \emph{erp\_msg} is then moved to the boundary place \emph{ch0} by transition \emph{CheckCeBoundary}, if the \emph{[encr=false]} guard holds, and thus ensures the correctness of the data exchange between patterns.
Subsequently only the actual message without the boundary data is moved to place \emph{ch0}, that is linked to the input place \emph{ch0} of the enricher, as in \cref{fig:pc:material_replicate_w_boundaries_simulation}. 
We recall, that the \emph{in} port type ensures that the synchronization on the CPN color set level are correct.
After the \emph{enricher} processing, the \emph{out} port type ensures the correctness of the synchronization on the CPN color set level and the resulting message \emph{emsg} is moved to the linked output place \emph{ch4}.
The constructed outbound boundary, represented by transition \emph{SetCeBoundary} sets the boundary properties of the enricher to (msg,false,false,false) for the following pattern.
On the input boundary side of the \emph{translator}, transition \emph{CheckMtBoundary} evaluates its guard, before moving the message without the boundary data to the boundary place \emph{ch5}, which proceeds similar to the enricher.

Note that our boundary construction mechanism from \cref{def:pc:construction} generated the input boundary, \eg denoted by place \emph{erpToCe} and transition \emph{CheckCeBoundary}, as well as the output boundary, \eg transition \emph{SetCeBoundary} and place \emph{ceToMt}, including the transition guards, colorsets, variables, and port type configurations, for the validation by simulation. 

\paragraph{Discussion}
Notably, constructing an IPCG requires less technical knowledge such as particularities of \tdbnets but still enables correct pattern compositions on an abstract level.
While the $CPT$ part of the pattern contracts (\eg encrypted, signed) could be derived and translated automatically from a scenario in a yet to be defined modeling language, many aspects like their elements $EL$ as well as the configuration of the characteristics by enrichment and mapping programs requires a technical understanding of IPCGs and the underlying scenarios.
As such IPCGs can be considered a suitable intermediate representation of pattern compositions.
The user might still prefer a more appealing graphical modeling language on top of IPCGs.
The simulation capabilities of the constructed \tdbnet with boundaries allow for the experimental validation of real-world pattern compositions.
However, the complexity of the construction highlights the importance of an automation of the construction.

\labelsubtitle{Conclusions} (11) IPCG and \tdbnet with boundaries can be shown correct with respect to composition and execution semantics; 
(12) \tdbnets with boundaries are even more complex than \tdbnets; (13) IPCGs are more comprehensible than \tdbnets, and expressive enough for current integration scenarios.

\subsubsection{Internet of Things: Predictive Maintenance and Service (PDMS)}
The IPCG representing the predictive maintenance create notification scenario that connects machines with enterprise resource planning (ERP) and PDMS systems is given in~\cref{fig:pc:cpn_pattern_composition_pdms}.
We add all pattern characteristics and data, which provides sufficient information for the translation to \tdbnets with boundaries.
\Cref{fig:pc:cpn_pattern_composition_pdms} depicts the corresponding IPCG.
The characteristics of the 
$CE_1$ node includes an enrichment function $prg1$ that adds further information about the machine in the form of the \emph{FeatureType} to the message that contains machine \emph{ID} and \emph{UpperThresholdWarningValue}.
This data is leveraged by the $UDF_1$ \emph{predict} node, which uses a prediction function $prg2$ about the need for maintenance and adds the result into the \emph{MaintenanceRequestById} field.
Before the data is forwarded to the ERP system (simplified by an $End$), the single machine predictions are combined into one message by the $AGG_1$ node with correlation $cnd_{cr}$ and completion $cnd_{cc}$ conditions as well as the aggregation function $prg_{3}$ and completion timeout $(v_1,v_2)$ as pattern characteristics $\{(\{cnd_{cr}, cnd_{cc}\}), (PRG, prg_{4}, (v_1,v_2)) \}$.
\begin{figure*}[bt]
	\centering
	\includegraphics[width=.8\linewidth]{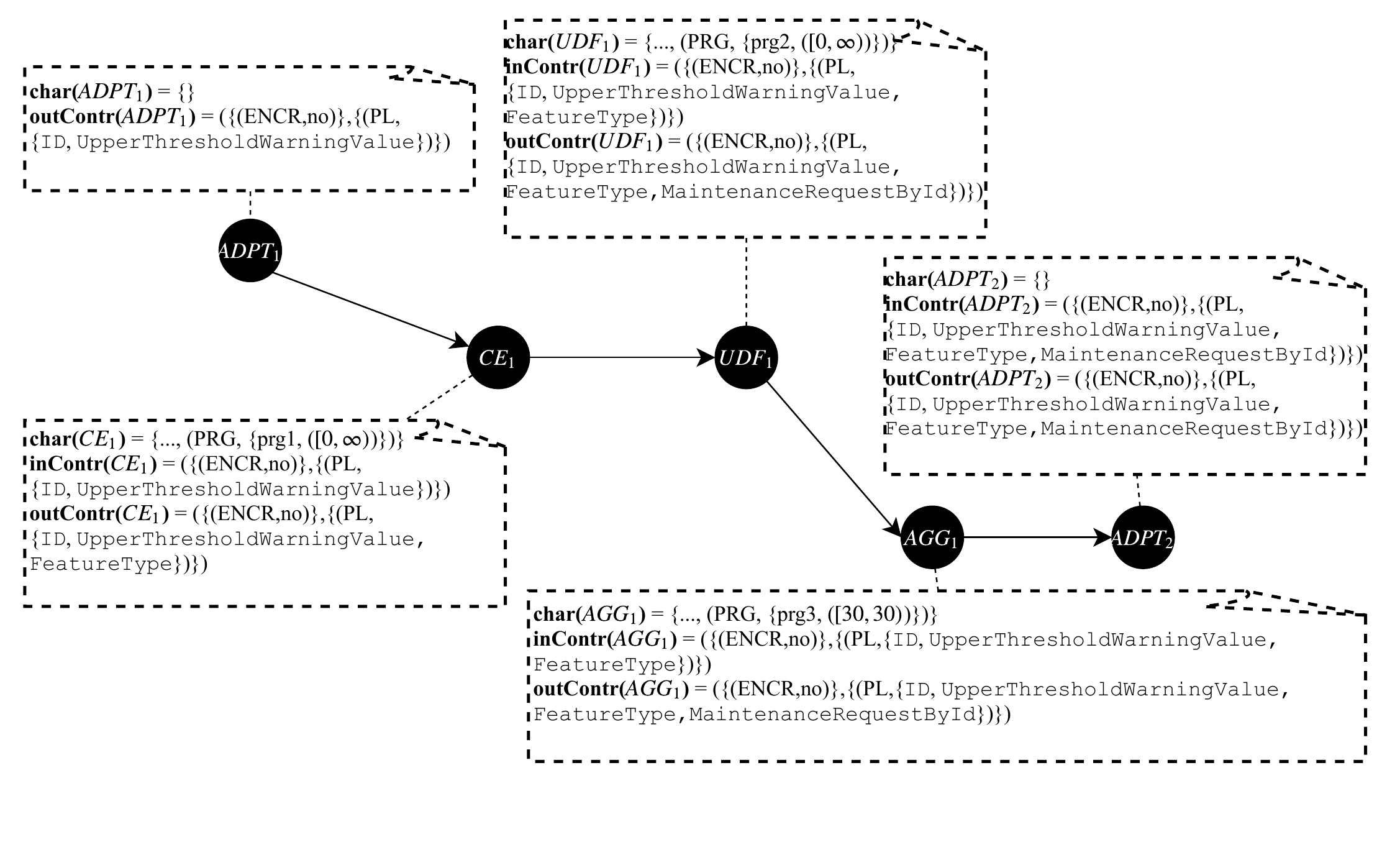} 
	\vspace{-1.1cm}
	\caption{Integration pattern contract graph of the predictive maintenance scenario}
	\label{fig:pc:cpn_pattern_composition_pdms}
\end{figure*}
\paragraph{Translation to Timed DB-Nets\index{Timed db-net} with Boundaries}
Again, we translate each single pattern from~\cref{fig:pc:cpn_pattern_composition_pdms} according to the construction in~\cref{sec:interpretation-atomic}.
The $Start$ and $End$ nodes are translated as the start and end pattern in \cref{fig:start} and \cref{fig:end} respectively.
The \Ce $CE_1$ and user-defined function $UDF_1$ nodes are message processors, and hence translated as in \cref{fig:siso_w_storage} with $<f>_{CE_1} = prg1$ and $<f>_{UDF_1} = prg2$. 
Since no table updates are needed for either translation, the database update function parameter $<g>$ can be chosen to be the identity function in all cases.
The \Aggregator $AGG_1$ node is a merge pattern type, and thus translated as in~\cref{fig:aggregator} with $(v_1,v_2) \rightarrow [\tau_1,\tau_2]$, $prg_{agg} \rightarrow f(msgs)$.
Moreover, the correlation condition $cnd_{cr} \rightarrow g(msg,msgs)$ and the completion condition $cnd_{cc} \rightarrow complCount$.

In the second step, we refine the \tdbnet with boundaries to also take contract concepts into account by the construction in~\cref{def:pc:construction}.
The resulting net is shown in~\cref{fig:pc:pdms_w_boundaries}.
\begin{figure*}[bt]
	\centering
	\includegraphics[width=.9\linewidth]{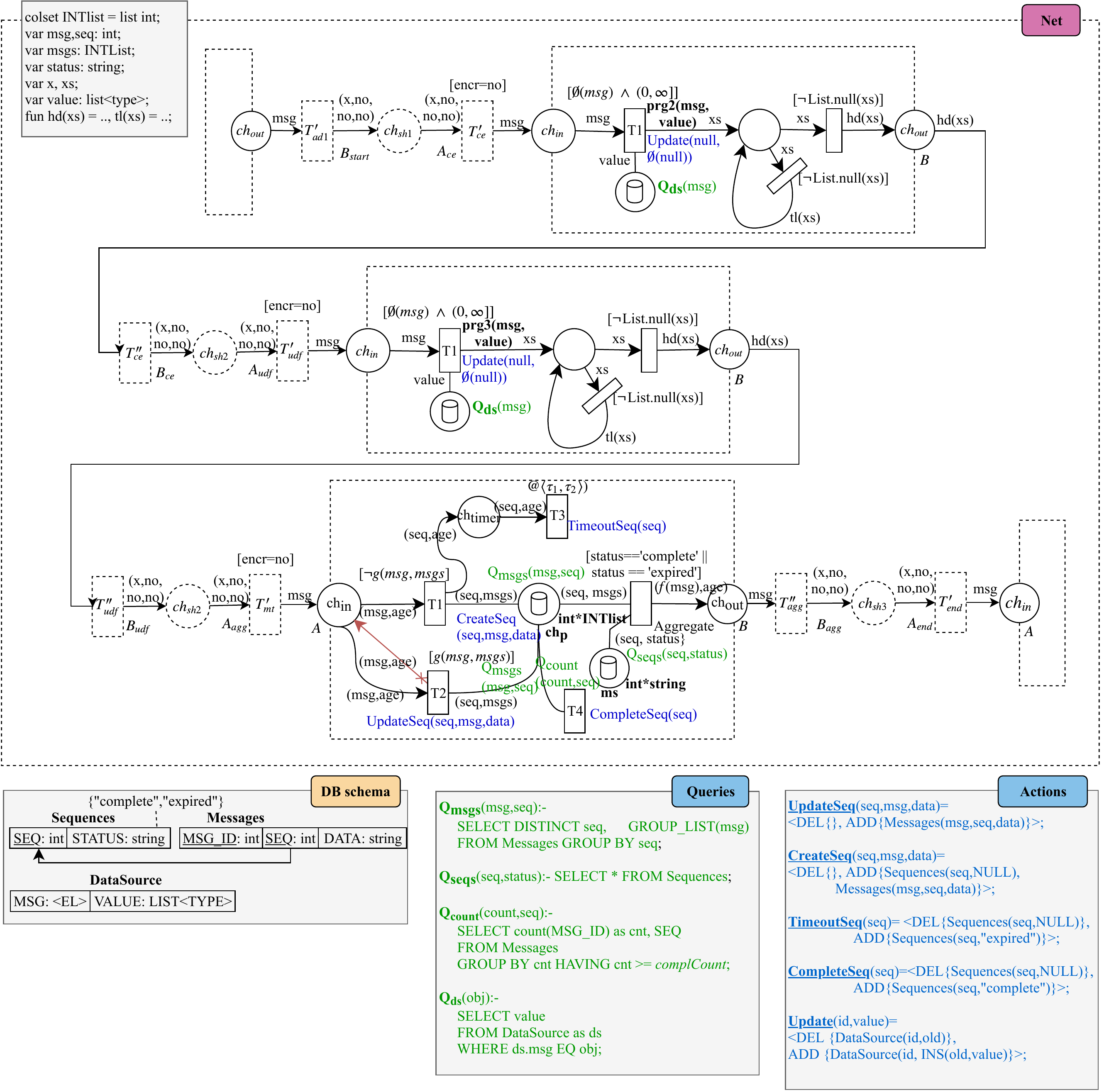}
	\caption{Predictive maintenance scenario as a \tdbnets with boundaries}
	\label{fig:pc:pdms_w_boundaries}
\end{figure*}
This ensures the correctness of the types of data exchanged between patterns, and follows directly from the correctness of the corresponding IPCG.
Other contract properties such as encryption, signatures, and encodings are checked through the transition guards.

\paragraph{Simulation} We illustrate the composition construction of the predictive maintenance scenario in~\cref{fig:pc:pdms_w_boundaries} through simulation in the form of a hierarchical \tdbnet model, shown in~\cref{fig:pc:create_notification_boundary_after}.
\begin{figure*}[bt]
	\centering
	\includegraphics[width=1\linewidth]{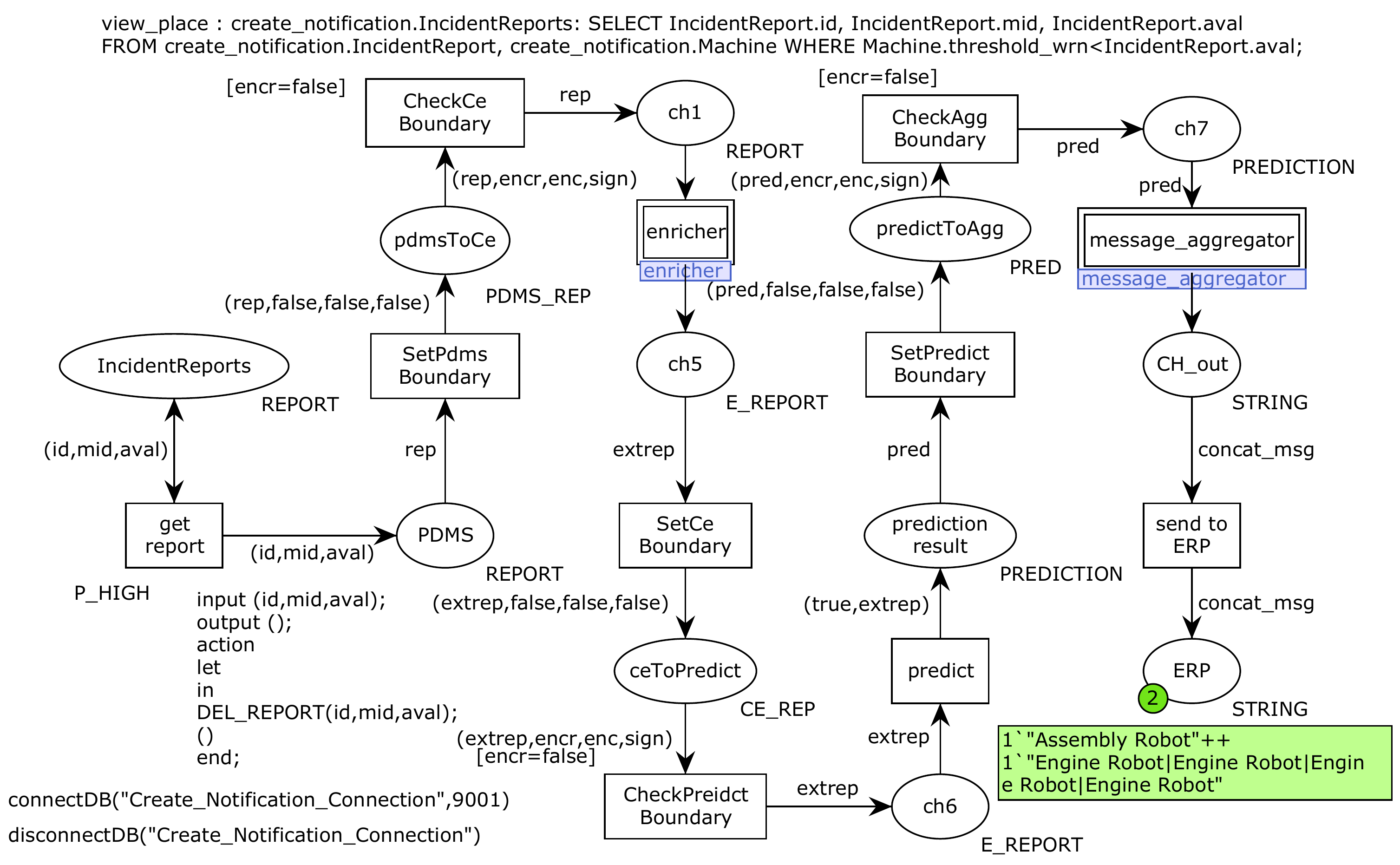}
	\caption{Predictive maintenance scenario simulation}
	\label{fig:pc:create_notification_boundary_after}
\end{figure*}
Again, all \tdbnet patterns are hierarchically represented by CPN Tool Subpage elements that are annotated with subpage tags \emph{enricher}, \emph{message\_aggregator}, respectively, and the user-defined function \emph{predict} is denoted by a transition.
The boundaries are constructed from~\cref{fig:pc:pdms_w_boundaries} by inserting \emph{SetPdmsBoundary} and \emph{pdmsToCe} as output boundary of \emph{get report}, which matches the input boundary of the subsequent \emph{enricher}, denoted by the \emph{CheckCeBoundary} transition.
Transition \emph{SetCeBoundary} and place \emph{ceToPredict} represent the output boundary of the enricher, which again match the input boundary of the \emph{predict} user-defined function pattern through transition \emph{CheckPredictBoundary}.
Finally, the output boundary of the predict step is ensured by transition \emph{SetPredictBoundary} and place \emph{predictToAgg}.
Again, it can be easily seen that the input boundary of the aggregator in the form of the \emph{CheckAggBoundary} transition matches, and thus the overall composition is correct.
Consequently, the simulation of the \tdbnet with these boundaries in \cref{fig:pc:create_notification_boundary_after} results in the same, correct output with the \tdbnets without boundaries in \cref{fig:pc:create_notification_boundary_after}.

\paragraph{Discussion} In this slightly more complex scenario, it becomes more obvious that the constructed IPCGs are quite technical as well and require a careful construction of pattern characteristics and contracts.
While this seems to be an ideal representation for checking the structural correctness of compositions, this should be no manual task for a user.
Especially for more complex scenarios, we found that the re-configurable pattern type-based translation works well.
However, the construction of the \tdbnets with boundaries corresponding to an IPCG would benefit from an automatic translation (\eg through tool support).

\labelsubtitle{Conclusions} (14) IPCGs are still quite technical, especially for more complex scenarios; (15) a tool support for automatic construction and translation is preferable.


%% file: relatedwork.tex
\section{Related Work}
\label{sec:relatedwork}
We presented related optimization techniques in \cref{sec:stategies}.
We now briefly situate our work within the context of other formalizations, beyond the already discussed BPMN~\cite{ritter2016exception} and PN~\cite{DBLP:conf/caise/FahlandG13} approaches, as summarized in \cref{tab:related_work}.
\begin{table}[bt]
	\centering
	\small
	\caption{Optimization Strategies in context of the objectives}
	\label{tab:related_work}
	\begin{tabular}{l|ccc}
		\hline
          Approach & Formal model & Optimizations & Correctness \\
		\hline
		EAI & \faThumbsOUp & \faThumbsDown & \faThumbsOUp \\
		BPM & \faThumbsUp & \faThumbsDown & \faThumbsOUp \\
		SPC & \faThumbsOUp & \faThumbsUp & \faThumbsUp \\
		AOS & \faThumbsDown & \faThumbsUp & \faThumbsDown \\
		GT & \faThumbsUp & \faThumbsDown & \faThumbsOUp \\
		\hline
		ReCO & \faThumbsUp & \faThumbsUp & \faThumbsUp \\
	\end{tabular}
\begin{tablenotes}
	\centering
	\small
	\item \faThumbsUp: covered, \faThumbsOUp: partially covered, \faThumbsDown: not covered 
\end{tablenotes}
\end{table}

\labeltitle{Enterprise Application Integration (EAI)} Similar to the BPMN and PN notations, several domain-specific languages (DSLs) have been developed that describe integration scenarios.
Apart from the EIP icon notation~\cite{hohpe2004enterprise}, there is also the Java-based Apache Camel DSL~\cite{Ibsen:2010:CA:1965487}, and the UML-based Guaran\'a DSL~\cite{frantz2011domain}.
However, none of these languages aim to be optimization-friendly formal integration scenario representations.
Conversely, we do not strive to build another integration DSL.
Instead we claim that all of the integration scenarios expressed in such languages can be formally represented in our formalism, so that optimizations can be determined that can be used to rewrite the scenarios.

There is work on formal representations of integration patterns, e.g.\ Mederly et al.~\cite{mederly2009construction} represents messages as first-order formulas and patterns as operations that add and delete formulas, and then applies AI planning to find an process with a minimal number of components.
While this approach shares the model complexity objective, our approach applies to a broader set of objectives and optimization strategies.
For the verification of service-oriented manufacturing systems,~Mendes et al.~\cite{mendes2012high} uses \enquote{high-level} Petri nets as a language instead of integration patterns, similar to the approach of Fahland and Gierds~\cite{DBLP:conf/caise/FahlandG13}.

\labeltitle{Business Process Management (BPM)} 
Sadiq and Orlowska~\cite{sadiq2000analyzing} applied reduction rules to workflow graphs for the visual identification of structural conflicts (\eg deadlocks) in business processes.
Compared to process control graphs, we use a similar base representation, which we extend by pattern characteristics and data contracts.
Furthermore, we use graph rewriting for optimization purposes.
In Cabanillas et al.~\cite{cabanillas2011automatic}, the structural aspects are extended by a data-centered view of the process that allows to analyze the life cycle of an object, and check data compliance rules.
This adds a view on the required data, but does not propose optimizations for the 
EIPs.
The main focus is rather on the object life cycle analysis of the process.

\labeltitle{Semantic Program Correctness (SPC)}
Semantic correctness plays a bigger role in the compiler construction and analysis domain. 
For example, Muchnick~\cite{muchnick1997advanced} provides an exhaustive catalog of optimizing transformations and states that the proof of the correctness of rewritings must be based on the (execution) semantics, and Nielson~\cite{nielson1981semantic} provides semantic correctness proofs using data-flow analysis, while Cousot~\cite{cousot2002systematic} provides a general framework for designing program transformations by analyzing abstract interpretations.
Although far simpler than general programming language transformations, our translation of IPCGs to timed db-nets with boundaries can be seen as a concretization in the sense of an abstract interpretation, and thus giving a similar notion of semantic correctness.

\labeltitle{Analysis and Optimization Structures (AOS)}
Transformation techniques for optimization have been employed by compiler construction, \eg for parallel~\cite{kuck1972number} or pipeline processing~\cite{kuck1980analysis},
where dependence graph representations become especially useful.
For example, Kuck et al.~\cite{kuck1981dependence} construct dependence graphs with output, anti, and flow dependencies as a foundation for optimizing transformations.
These kind of dependence graphs were also used by~B\"ohm et al.~\cite{DBLP:journals/is/BohmHPLW11}, however, they are \enquote{linearized} in the form of our pattern contracts.
This makes the decision of the optimization \enquote{local} and does not require dependence graph abstractions like intervals~\cite{cocke1970global} or scoping~\cite{zelkowitz1974optimization}.
More recently these techniques have been applied for business process optimization by Sadiq~\cite{sadiq2000analyzing}, Niedermann et al.~\cite{niedermann2011business,niedermann2011deep} or reductions to process tree structures~\cite{vanhatalo2009refined} with incremental transformations~\cite{hauser2008incremental}.
In our case the scope of the analysis is a local match of pattern contracts.

\labeltitle{Graph Transformations (GT)}
Similar to our approach, graph transformations have been used in related domains, \eg formalizing parts of the BPMN semantics by Dijkman et al.~\cite{DBLP:conf/bpmn/DijkmanG10}, who specify the execution semantics as graph rewrites.
Conformance is checked experimentally and verification is left for future work.
For the optimizations, we use the same visual notation and double-pushout rule application approach.
However, our execution semantics are given as timed db-net and can be formally analyzed.

%% file: discussion.tex
\section{Conclusions}
\label{sec:conclusion}
This work addresses an important shortcoming in EAI research, namely the lack of means for responsible or correct integration pattern compositions and the application of changes, \eg as optimization strategies, which preserve structural and semantic correctness, and thus ends the informality of descriptions of pattern compositions and optimizations (cf.\ Q1--Q3).

We approached the questions along a responsible pattern composition and optimization process (short ReCO), and started by compiling catalogs of integration pattern characteristics as well as optimization strategies from the literature.
We then developed a formalization of pattern compositions in order to precisely define optimizations as pattern contract graphs.
Then we extended the timed db-nets formalism, covering integration pattern semantics, into timed db-nets with boundaries, which resemble the contracts in the pattern graphs, and defined a mechanism to interpret the pattern graphs by them.

With the resulting formal framework, we proved that all defined optimizations preserve the meaning of compositions as timed db-nets. We evaluated the framework on data sets containing in total over 900 real world integration scenarios, and two brief case studies.
The responsible pattern composition part in ReCO was then studied for two integration processes down to the execution semantics, essentially showing ReCO from a modeling perspective.

We conclude that formalization and optimizations of integration processes in the form of integration pattern compositions --- using pattern contract graphs --- are relevant even for experienced integration experts (conclusions 1--2), with interesting choices (concls. 3--4, 6), implementation details (conclusions 5, 10) and trade-offs (concls. 7--9).
In the two additional case studies, we showed the suitability of our interpretation of pattern contract graphs in the newly defined timed db-nets with boundaries for pattern compositions (concls. (11,13)), model complexity considerations (concls. (12,14)) and desirable extensions like automatic translation (concl. (15)).
